\documentclass[pra,aps,nopacs,showkeys,onecolumn,twoside,superscriptaddress]{revtex4}

\usepackage{array}[=2016-10-06]

\usepackage{amsmath,amsfonts,amssymb,caption,color,epsfig,graphics,graphicx,hyperref,latexsym,mathrsfs,revsymb,theorem,url,verbatim,epstopdf,mathtools,enumerate}
\usepackage{subfigure}\usepackage{makecell}\usepackage{multirow}\usepackage{diagbox}
\usepackage{booktabs,tabularx,placeins}
\hypersetup{colorlinks,linkcolor={blue},citecolor={blue},urlcolor={blue}}

\usepackage[qm]{qcircuit}

\newtheorem{definition}{Definition}
\newtheorem{proposition}[definition]{Proposition}
\newtheorem{lemma}[definition]{Lemma}

\newtheorem{theorem}[definition]{Theorem}
\newtheorem{corollary}[definition]{Corollary}
\newtheorem{conjecture}[definition]{Conjecture}

\newtheorem{example}[definition]{Example}
\newtheorem{question}[definition]{Question}
\newtheorem{memo}[definition]{Memo}
\newtheorem{appendixlemma}{Lemma}[section]

\def\squareforqed{\hbox{\rlap{$\sqcap$}$\sqcup$}}
\def\qed{\ifmmode\squareforqed\else{\unskip\nobreak\hfil
\penalty50\hskip1em\null\nobreak\hfil\squareforqed
\parfillskip=0pt\finalhyphendemerits=0\endgraf}\fi}
\def\endenv{\ifmmode\;\else{\unskip\nobreak\hfil
\penalty50\hskip1em\null\nobreak\hfil\;
\parfillskip=0pt\finalhyphendemerits=0\endgraf}\fi}
\newenvironment{proof}{\noindent \textbf{{Proof.~} }}{\qed}
\newenvironment{resultproof}[1]{\noindent \textbf{Proof of #1.~}}{\qed}
\def\Dbar{\leavevmode\lower.6ex\hbox to 0pt
{\hskip-.23ex\accent"16\hss}D}
\makeatletter
\def\url@leostyle{%
  \@ifundefined{selectfont}{\def\UrlFont{\sf}}{\def\UrlFont{\small\ttfamily}}}
\makeatother
\def\bcj{\begin{conjecture}}
\def\ecj{\end{conjecture}}
\def\bcr{\begin{corollary}}
\def\ecr{\end{corollary}}
\def\bd{\begin{definition}}
\def\ed{\end{definition}}
\def\bea{\begin{eqnarray}}
\def\eea{\end{eqnarray}}
\def\beq{\begin{equation}}
\def\eeq{\end{equation}}
\def\bal{\begin{aligned}}
\def\eal{\end{aligned}}
\def\bem{\begin{enumerate}}
\def\eem{\end{enumerate}}
\def\bex{\begin{example}}
\def\eex{\end{example}}
\def\bim{\begin{itemize}}
\def\eim{\end{itemize}}
\def\bl{\begin{lemma}}
\def\el{\end{lemma}}
\def\bma{\begin{bmatrix}}
\def\ema{\end{bmatrix}}
\def\bpf{\begin{proof}}
\def\epf{\end{proof}}
\def\bpp{\begin{proposition}}
\def\epp{\end{proposition}}
\def\bqu{\begin{question}}
\def\equ{\end{question}}
\def\er{\end{remark}}
\def\bt{\begin{theorem}}
\def\et{\end{theorem}}
\def\bmm{\begin{memo}}
\def\emm{\end{memo}}

\def\btb{\begin{tabular}}
\def\etb{\end{tabular}}

\newcommand{\nc}{\newcommand}

\def\a{\alpha}
\def\b{\beta}

\def\t{\theta}

\def\r{\rho}
\def\s{\sigma}

\def\ps{\psi}

\nc{\bbA}{\mathbb{A}} \nc{\bbB}{\mathbb{B}} \nc{\bbC}{\mathbb{C}}
 \nc{\bbD}{\mathbb{D}} \nc{\bbE}{\mathbb{E}} \nc{\bbF}{\mathbb{F}}
 \nc{\bbG}{\mathbb{G}} \nc{\bbH}{\mathbb{H}} \nc{\bbI}{\mathbb{I}}
 \nc{\bbJ}{\mathbb{J}} \nc{\bbK}{\mathbb{K}} \nc{\bbL}{\mathbb{L}}
 \nc{\bbM}{\mathbb{M}} \nc{\bbN}{\mathbb{N}} \nc{\bbO}{\mathbb{O}}
 \nc{\bbP}{\mathbb{P}} \nc{\bbQ}{\mathbb{Q}} \nc{\bbR}{\mathbb{R}}
 \nc{\bbS}{\mathbb{S}} \nc{\bbT}{\mathbb{T}} \nc{\bbU}{\mathbb{U}}
 \nc{\bbV}{\mathbb{V}} \nc{\bbW}{\mathbb{W}} \nc{\bbX}{\mathbb{X}}
 \nc{\bbZ}{\mathbb{Z}}

 \nc{\bA}{{\bf A}} \nc{\bB}{{\bf B}} \nc{\bC}{{\bf C}}
 \nc{\bD}{{\bf D}} \nc{\bE}{{\bf E}} \nc{\bF}{{\bf F}}
 \nc{\bG}{{\bf G}} \nc{\bH}{{\bf H}} \nc{\bI}{{\bf I}}
 \nc{\bJ}{{\bf J}} \nc{\bK}{{\bf K}} \nc{\bL}{{\bf L}}
 \nc{\bM}{{\bf M}} \nc{\bN}{{\bf N}} \nc{\bO}{{\bf O}}
 \nc{\bP}{{\bf P}} \nc{\bQ}{{\bf Q}} \nc{\bR}{{\bf R}}
 \nc{\bS}{{\bf S}} \nc{\bT}{{\bf T}} \nc{\bU}{{\bf U}}
 \nc{\bV}{{\bf V}} \nc{\bW}{{\bf W}} \nc{\bX}{{\bf X}}
 \nc{\bZ}{{\bf Z}}

\nc{\cA}{{\cal A}} \nc{\cB}{{\cal B}} \nc{\cC}{{\cal C}}
\nc{\cD}{{\cal D}} \nc{\cE}{{\cal E}} \nc{\cF}{{\cal F}}
\nc{\cG}{{\cal G}} \nc{\cH}{{\cal H}} \nc{\cI}{{\cal I}}
\nc{\cJ}{{\cal J}} \nc{\cK}{{\cal K}} \nc{\cL}{{\cal L}}
\nc{\cM}{{\cal M}} \nc{\cN}{{\cal N}} \nc{\cO}{{\cal O}}
\nc{\cP}{{\cal P}} \nc{\cQ}{{\cal Q}} \nc{\cR}{{\cal R}}
\nc{\cS}{{\cal S}} \nc{\cT}{{\cal T}} \nc{\cU}{{\cal U}}
\nc{\cV}{{\cal V}} \nc{\cW}{{\cal W}} \nc{\cX}{{\cal X}}
\nc{\cZ}{{\cal Z}}

\nc{\hA}{{\hat{A}}} \nc{\hB}{{\hat{B}}} \nc{\hC}{{\hat{C}}}
\nc{\hD}{{\hat{D}}} \nc{\hE}{{\hat{E}}} \nc{\hF}{{\hat{F}}}
\nc{\hG}{{\hat{G}}} \nc{\hH}{{\hat{H}}} \nc{\hI}{{\hat{I}}}
\nc{\hJ}{{\hat{J}}} \nc{\hK}{{\hat{K}}} \nc{\hL}{{\hat{L}}}
\nc{\hM}{{\hat{M}}} \nc{\hN}{{\hat{N}}} \nc{\hO}{{\hat{O}}}
\nc{\hP}{{\hat{P}}} \nc{\hR}{{\hat{R}}} \nc{\hS}{{\hat{S}}}
\nc{\hT}{{\hat{T}}} \nc{\hU}{{\hat{U}}} \nc{\hV}{{\hat{V}}}
\nc{\hW}{{\hat{W}}} \nc{\hX}{{\hat{X}}} \nc{\hZ}{{\hat{Z}}}

\nc{\hn}{{\hat{n}}}

\def\Span{\mathop{\rm span}}
\def\dim{\mathop{\rm Dim}}

\def\max{\mathop{\rm max}}
\def\min{\mathop{\rm min}}

\def\rank{\mathop{\rm rank}}

\def\supp{\mathop{\rm supp}}
\def\tr{\mathop{\rm Tr}}

\def\dg{\dagger}

\def\ox{\otimes}

\def\ra{\rightarrow}

\newcommand{\bra}[1]{\langle#1|}
\newcommand{\ket}[1]{|#1\rangle}
\newcommand{\proj}[1]{| #1\rangle\!\langle #1 |}

\newcommand{\braket}[2]{\langle#1|#2\rangle}

\newcommand{\norm}[1]{\lVert#1\rVert}
\newcommand{\abs}[1]{\left|#1\right|}

\def\Dbar{\leavevmode\lower.6ex\hbox to 0pt
{\hskip-.23ex\accent"16\hss}D}

\begin{document}

\large

\title{
Rank and Range Criteria for Mixed-State Determination from Local Marginals}


\author{Xinyu Qiu}
\affiliation{LMIB(Beihang University), Ministry of Education, and School of Mathematical Sciences, Beihang University, Beijing 100191, China}

\author{Minglong Qin}
\affiliation{Centre for Quantum Technologies, National University of Singapore, Singapore 117543, Singapore}

\author{Caohan Cheng}
\affiliation{LMIB(Beihang University), Ministry of Education, and School of Mathematical Sciences, Beihang University, Beijing 100191, China}

\author{Lin Chen}\email[]{linchen@buaa.edu.cn (corresponding author)}
\affiliation{LMIB(Beihang University), Ministry of Education, and School of Mathematical Sciences, Beihang University, Beijing 100191, China} 

\author{Delin Chu} \affiliation{ Department of Mathematics, National University of Singapore, Singapore 119076, Singapore}

\begin{abstract}
Determining whether a mixed quantum state is uniquely determined among all states by its k-body marginals (k-UDA) is a fundamental problem in quantum system certification. We develop a range-based approach to this problem by analyzing the structure of the range of the global state. For three-qubit states, we show that states with GHZ-SLOCC-free ranges are 2-UDA at ranks one, three, and four. We derive a necessary and sufficient range criterion for rank-two 2-UDA states and reduce it to a finite quadratic-form test. 
To cover the remaining range configurations, we formulate an exact range-restricted semidefinite programming criterion and extend it to arbitrary finite-dimensional tripartite states. We also show that every three-qubit state of rank at least five is not 2-UDA, and further extend high-rank obstructions to multipartite systems. 
For a channel-based multipartite family, we characterize exactly when a state is $(n-1)$-UDA and show that lower-order marginals never suffice. 
Finally, we apply these results to the certification of genuine multipartite entanglement.
\end{abstract}

\keywords{quantum marginal problem, uniquely determined among all states,
mixed quantum states, semidefinite programming, genuine multipartite
entanglement}

\maketitle

\section{Introduction}
\label{sec:int}

Inferring global properties of a quantum system from local information
is a central task in quantum-system certification. The quantum marginal
problem first asks whether a prescribed
collection of reduced density matrices is compatible with a common
global state \cite{coleman1963structure,klyachko2006quantum}.
Spectral variants of the quantum marginal problem have been analyzed
using representation-theoretic and symplectic-geometric methods
\cite{christandl2006spectra,christandl2014eigenvalue}.
Important formulations of this problem, such as the
$N$-representability problem, are
QMA-complete \cite{liu2007nrepresentability}. We consider the stronger
question of uniqueness. A state is uniquely determined among all
states (UDA) by the chosen marginals if they are compatible with no
other global state. This uniqueness is the identifiability condition
that allows local data to certify the global state. It underlies
tomography from reduced density operators
\cite{swingle2014reconstructing,xin2017quantum,haah2017sample,gao2025optimal}, 
and also complements fidelity-estimation and
state-verification protocols based on local observables
\cite{flammia2011fidelity,pallister2018verification}. The UDA property is further
connected with certification by local Hamiltonians. A unique
ground state of a quasi-local Hamiltonian is UDA by the marginals on
the interaction supports, whereas the converse need not hold
\cite{from2012chen,karuvade2019uniquely}. Related uniqueness
conditions also govern whether a target state can be prepared as the
sole steady state of quasi-local dissipative dynamics
\cite{ticozzi2012stabilizing,karuvade2018generic}. Conversely, the
failure of small-subsystem marginals to distinguish different global
states is the local indistinguishability that protects quantum code
spaces and topologically ordered ground-state manifolds
\cite{knill1997theory,bravyi2010topological}. Finally, marginal
compatibility reflects the structure of multipartite correlations. It
can reveal irreducible higher-order information and certify global
entanglement, even when the individual marginals are separable
\cite{zhou2008irreducible,chen2014role,miklin2016emergent,
shi2025entangle,navascues2021entanglement}.

Pure states have a simpler structure and have therefore been
studied extensively since the UDA problem was first introduced,
although a complete characterization remains open in most settings.
Motivated by the characterization of three-qubit 2-UDA pure states
\cite{linden2002almost,diosi2004three}, generic multipartite pure-state uniqueness was
subsequently related to the size and number of the available marginals
\cite{linden2002parts,jones2005parts}. Complementary reducibility
results were obtained for $W$-type states and generic four-particle
pure states
\cite{parashar2009nqubit,rana2011optimal,wyderka2017almost}. For
$n$-qubit pure states, the generalized
GHZ states and their LU equivalents have been identified as the only exceptions, first
among pure states and then among arbitrary pure or mixed
states \cite{walck2008only,walck2009parts}. 
Technically, the collections of marginals that admit a global quantum-state
extension form a convex set. The boundary and facial structure of this
set, together with the fibers of the marginal map, provide a geometric
framework for determining whether local marginals uniquely specify the
global state \cite{chen2012erdahl}.

In realistic physical settings, quantum systems are typically described
by mixed states. It is therefore important to characterize mixed UDA
states, which have received less attention than their pure counterparts.
A bipartite mixed state is 1-UDA if and only if it has a pure one-body
marginal \cite{qiu2026mixed}.
As the next step, the 1-UDA problem
for three qubits is elementary. Up to a system permutation, a state is
1-UDA precisely when it has the form
$\proj{a}\ox\proj{b}\ox\tau_C$, where $\tau_C$ is an arbitrary
one-qubit state. Thus, the classification of three-qubit 2-UDA states is the first genuinely nontrivial mixed-state case and is the main focus of this work. A generic
low-rank tripartite result reduces to rank one when all three local
dimensions are two and therefore does not cover mixed three-qubit
states \cite{chen2013uniqueness}. 
Additivity constructions for mixed UDA states have been studied
\cite{shen2023additivity}, and several families of mixed UDA states
have been completely characterized. More
generally, a state cannot be $k$-UDA if the face determined by its
range contains a non-$k$-UDA state \cite{qiu2026mixed}. In particular,
the UDA status depends only on the range and not on the nonzero
spectrum. Together with the characterization of three-qubit pure
states \cite{linden2002almost,walck2009parts}, this face property shows
that a three-qubit state cannot be 2-UDA if its range contains a
GHZ-LU vector. The existing
mixed-state classifications \cite{qiu2026mixed}, however, rely on a prescribed product
structure and do not characterize a general three-qubit mixed state.
This gap motivates us to move toward a
complete characterization of the
three-qubit 2-UDA problem without assuming such a product factorization and,
more broadly, to extend the resulting criteria and methods to multipartite
systems.

In this work, we study mixed UDA states through the structure of their
ranges and then extend the resulting methods to multipartite systems.
First, we consider the 2-UDA problem for three-qubit mixed states. We
use rank as a coarse parameter and the structure of the state range as
the finer invariant. The known GHZ-LU obstruction directs attention to
ranges containing no GHZ-LU vectors, while its converse fails by  Example~\ref{exp:GHZ-LU_not2UDA}. We therefore
begin with the stronger condition that the range contains no
GHZ-SLOCC vector. Using projective geometry, we show that every
three- or four-dimensional range satisfying this condition is
contained, up to the relevant equivalences, in one of the canonical
subspaces $\cT$ in \eqref{eq:cT} and $\cC_\theta$ in
\eqref{eq:cC_t}. We then prove that every state whose range is
contained in either canonical subspace is 2-UDA. Consequently, every GHZ-SLOCC-free
three-qubit state of rank one, three, or four is 2-UDA.
For every rank-two state, we
derive a necessary and sufficient range criterion and reduce it to
a finite quadratic-form test. For states outside the structural classification, generalized state inversion restricts all compatible states to an enlarged range determined by the given state, within which the marginal-preserving directions form a compact spectrahedron. This yields an exact 2-UDA criterion in terms of finitely many range-restricted semidefinite programs (SDPs), covering in particular the remaining rank-three and rank-four cases. Finally, we show that every
five-dimensional three-qubit range contains a GHZ-LU vector. Hence
no three-qubit state of rank at least five is 2-UDA, and
$\rank\rho\leq4$ is the sharp universal rank bound. 

We next extend the preceding methods to multipartite systems, and establish general high-rank
obstructions for UDA states in arbitrary finite local dimensions. For $n$ qubits, we show that 
rank at least $2^n-n$ rules out $k$-UDA for every $k\leq n-2$. A
spin-flip refinement establishes the same threshold for
$(n-1)$-UDA when $n\geq6$ and obtains the near-threshold bound
$2^n-n+1$ when $n=4,5$. We complement these universal obstructions
with a channel-based family on
$(\bbC^d)^{\ox(n-1)}\ox\bbC^s$. Its $(n-1)$-UDA property is
characterized by the purity and pairwise nonorthogonality of the
channel outputs, while no lower-order marginals determine a state in
this family. 
Finally, we apply the UDA criteria to GME certification. We obtain GME detection length two for
three generally nonsymmetric three-qubit families and maximal
detection length $n$ for a coherent channel family. We also construct
an $n$-qubit rank-two family whose separable 2-marginals jointly
certify GME. For $n\geq4$, these marginals certify GME without uniquely
determining the global state. The main results of this work are summarized in Table~\ref{tab:main-contributions}.

\begin{table}[htbp]
\caption{Summary of the main results for three-qubit and multipartite UDA states and applications}
\label{tab:main-contributions}
\centering
\small
\renewcommand{\arraystretch}{1.12}
\renewcommand{\tabularxcolumn}[1]{m{#1}}
\setlength{\arrayrulewidth}{0.4pt}
\begin{tabularx}{0.98\textwidth}{|
>{\centering\arraybackslash}m{0.105\textwidth}|
>{\raggedright\arraybackslash}m{0.21\textwidth}|
>{\raggedright\arraybackslash}X|
>{\raggedright\arraybackslash}m{0.19\textwidth}|}
\hline
Role & Scope & Main conclusion & Location \\
\hline
\multirow[c]{3}{0.105\textwidth}[-4.5mm]{\centering Three-qubit 2-UDA criteria}
& GHZ-SLOCC-free states
& Canonical range inclusions establish 2-UDA for ranks one, three,
and four.
& Propositions~\ref{pro:GHZ_free_subspace_T_C}
and~\ref{pro:T_Ct_2UDA}, and Theorem~\ref{th:GHZ_free} \\
\cline{2-4}
& Rank-two states
& A complete range criterion and a finite quadratic-form test
characterize 2-UDA.
& Proposition~\ref{pro:rank_two_2UDA_parameter} \\
\cline{2-4}
& Remaining range configurations and high ranks
& Exact statewise SDP values characterize the remaining
rank-three and rank-four states. Every state of rank at least five is
not 2-UDA.
& Proposition~\ref{pro:support-reduced-SDP} and
Theorem~\ref{th:five_dim_contains_GHZ_LU} \\
\hline
Multipartite Extension
& Finite-dimensional multipartite systems
& Dimension and spin-flip arguments establish high-rank
obstructions. A channel construction has exact UDA order $n-1$.
& Lemma~\ref{le:dimension_non_UDA_criterion},
Theorem~\ref{th:high_rank_dimension_obstruction}, and
Proposition~\ref{pro:UDA_diagonal(n-1)} \\
\hline
Application
& GME certification
& UDA criteria determine exact GME detection lengths, and separable
2-marginals certify GME in an $n$-qubit family.
& Propositions~\ref{pro:GME_detection_length_families}
and~\ref{pro:GME-from-separable-marginals} \\
\hline
\end{tabularx}
\end{table}
\FloatBarrier

The rest of this paper is organized as follows. In
Sec.~\ref{sec:pre}, we introduce the notation and recall the
projective-geometric and characteristic-class results used in the
range analysis. In Sec.~\ref{sec:three-qubit-2UDA-characterization}, we 
develop the primary three-qubit line. We first analyze
GHZ-SLOCC-free ranges, then derive the rank-two criterion and the
range-restricted statewise criterion, and finally establish the sharp
rank bound. In Sec.~\ref{sec:multipartite-UDA-results}, we extend the
methods to high-rank multipartite states and to the channel-based
family. In Sec.~\ref{sec:applications-GME}, we apply the resulting UDA
criteria to GME certification.
We conclude in Sec.~\ref{sec:conclusion}.

\section{Preliminaries}
\label{sec:pre}

In this section, we introduce the notations, and 
recall the facts about three-qubit pure states used later. We also
introduce the Cayley hyperdeterminant and an LU criterion for
generalized GHZ states. We use the standard relative-invariance
properties of the hyperdeterminant from
\cite{gelfand1994discriminants,miyake2003classification}.
We finally recall the projective-geometric and characteristic-class
facts used in the range arguments.

Let $A_1,\ldots,A_n$ be quantum systems associated with
finite-dimensional Hilbert spaces
$\cH_{A_1},\ldots,\cH_{A_n}$, respectively. We set
$[n]:=\{1,\ldots,n\}$ and
$\cH:=\cH_{A_1}\ox\cdots\ox\cH_{A_n}$, and we write
$d_i:=\dim\cH_{A_i}$ for $i\in[n]$. For a subset
$\cS\subseteq[n]$, we write $\cS^c=[n]\setminus\cS$ and denote the
composite system $\ox_{i\in\cS}A_i$ by $A_{\cS}$, associated with the
Hilbert space $\cH_{A_{\cS}}:=\ox_{i\in\cS}\cH_{A_i}$. For any
operator $X$ on $\cH$, we write
$X_{A_{\cS}}:=\tr_{A_{\cS^c}}X$, and $I_{A_{\cS}}$ denotes the
identity operator on $\cH_{A_{\cS}}$. For the empty subsystem, we use
$\cH_{A_\varnothing}:=\bbC$, $I_{A_\varnothing}:=1$, and
$X_{A_\varnothing}:=\tr X$. For
$1\leq k\leq n-1$, we define the $k$-marginal map by
$$
\cM_k(X):=\bigl(X_{A_{\cS}}\bigr)_{\cS\subseteq[n],\,|\cS|=k}.
$$
A quantum state $\rho$ on $\cH$ is a positive semidefinite operator of trace
one. We denote its range and kernel by $\cR(\rho)$ and $\ker\rho$,
respectively. When $\rho$ is fixed, we write
$\cL:=\cR(\rho)=(\ker\rho)^\perp$. The entries of
$\cM_k(\rho)$ are the $k$-marginals of $\rho$.
We denote the real vector space of Hermitian operators on $\cH$ by
$\operatorname{Herm}(\cH)$. For a positive integer $d$, we write
$\bbM_d:=M_d(\bbC)$ for the algebra of $d\times d$ complex matrices.
We use $T_{A_i}$ to denote partial transposition on subsystem $A_i$
in a fixed local basis, and write $X^{T_{A_i}}:=T_{A_i}(X)$.

\begin{definition}
\label{def:k-UDA}
Let $1\leq k\leq n-1$. Two $n$-partite states $\rho$ and $\sigma$ are $k$-compatible if $\cM_k(\rho)=\cM_k(\sigma)$. A state $\rho$ is $k$-uniquely determined among all states, or $k$-UDA, if there is no other state $k$-compatible with $\rho$.
\end{definition}

For $1\leq k\leq n-1$, we define the real kernel of the
$k$-marginal map by
\begin{eqnarray}
\label{eq:N_k_general}
\cN_k
&:=&\ker\cM_k\cap\operatorname{Herm}(\cH)
\nonumber\\
&=&\Big\{
H\in\operatorname{Herm}(\cH):
\tr_{A_{\cS^c}}H=0,
\ \cS\subseteq[n],\ |\cS|=k
\Big\}.
\end{eqnarray}
Thus two states are $k$-compatible if and only if their difference
belongs to $\cN_k$. Every operator in $\cN_k$ is traceless. For a
subspace $\cU\subseteq\cH$, we write
$\operatorname{Herm}(\cU):=\{H\in\operatorname{Herm}(\cH):
\cR(H)\subseteq\cU\}$ for the real linear space of Hermitian
operators whose ranges are contained in $\cU$.

For $n$-qubit systems, we write
$\cH_n:=(\bbC^2)^{\ox n}$. Let $\sigma_0=I_2$,
$\sigma_1=\sigma_x$, $\sigma_2=\sigma_y$, and
$\sigma_3=\sigma_z$. For
$\boldsymbol{\alpha}=(\alpha_1,\ldots,\alpha_n)
\in\{0,1,2,3\}^n$, we define the Pauli string by
$P_{\boldsymbol{\alpha}}:=\sigma_{\alpha_1}\ox\cdots\ox
\sigma_{\alpha_n}$ and its weight by
$\operatorname{wt}(\boldsymbol{\alpha}):=
|\{i\in[n]:\alpha_i\neq0\}|$.
Let $\mathsf C$ denote complex conjugation in the computational basis.
We define the $n$-qubit spin flip by
\begin{eqnarray}
\label{eq:n-qubit-spin-flip}
\Theta_n:=(i\sigma_y)^{\ox n}\mathsf C,
\qquad
\Theta_n^2=(-1)^nI.
\end{eqnarray}
For a subspace $\cU\subseteq\cH_n$, we write
$\Theta_n\cU:=\{\Theta_nu:u\in\cU\}$. Every Pauli string satisfies
\begin{eqnarray}
\label{eq:n-qubit-spin-flip-Pauli}
\Theta_nP_{\boldsymbol{\alpha}}\Theta_n^{-1}
=(-1)^{\operatorname{wt}(\boldsymbol{\alpha})}
P_{\boldsymbol{\alpha}}.
\end{eqnarray}

We recall some basic properties for UDA states.  The $k$-UDA property is invariant under local unitary transformations and system permutations. Moreover, a $k$-UDA state is also $\ell$-UDA for every $k\leq\ell\leq n-1$. For a subspace $\cU$ of the global Hilbert space, its face is the set of states whose ranges are contained in $\cU$. 

\begin{lemma}[\cite{qiu2026mixed}]
\label{le:lr+1-ls}
Let $\rho$ and $\sigma$ be two $n$-partite states. If $\rho$ is not $k$-UDA, then $\lambda\rho+(1-\lambda)\sigma$ is not $k$-UDA for every $\lambda\in(0,1]$. Consequently, if the face determined by $\cR(\rho)$ contains a non-$k$-UDA state, then $\rho$ is not $k$-UDA. In particular, two states with the same range are either both $k$-UDA or both not $k$-UDA.
\end{lemma}

We next specialize to three-qubit systems. For
$(A_1,A_2,A_3)=(A,B,C)$, we write
$\cH_{ABC}:=\cH_A\ox\cH_B\ox\cH_C$. 
A GHZ-LU state is LU equivalent
to $a\ket{000}+b\ket{111}$ for some $ab\neq0$, and a GHZ-SLOCC state
is defined analogously. Local operators act on a subspace elementwise,
which defines LU and SLOCC equivalence for subspaces. We call a
subspace GHZ-SLOCC-free if it contains no nonzero GHZ-SLOCC vector.

\begin{lemma}[\cite{linden2002almost,walck2009parts}]
\label{le:3-qubit_UDA}
A three-qubit pure state is 2-UDA unless it is a GHZ-LU state.
\end{lemma}

For a three-qubit vector $\ket\psi=\sum_{i,j,k=0}^1a_{ijk}\ket{ijk}$, we denote the $2\times2\times2$ Cayley hyperdeterminant of its coefficient tensor by $\operatorname{Det}(\ket\psi)$. Its explicit quartic expression is standard and will not be needed here. We refer to Refs.~\cite{gelfand1994discriminants,miyake2003classification} for this expression. We use only its zero set and the following relative-invariance property.
\begin{lemma}[\cite{gelfand1994discriminants,miyake2003classification}]
\label{lem:relative-invariance}
For any $L_A,L_B,L_C\in GL(2,\bbC)$ and any three-qubit vector
$\ket\psi$, one has
\begin{eqnarray}
\label{eq:Det(ABC ps)}
\operatorname{Det}\bigl((L_A\ox L_B\ox L_C)\ket\psi\bigr)
=(\det L_A)^2(\det L_B)^2(\det L_C)^2
\operatorname{Det}(\ket\psi).
\end{eqnarray}
Consequently, $\ket\psi$ is a GHZ-SLOCC state if and only if $\operatorname{Det}(\ket\psi)\neq0$. Moreover, $|\operatorname{Det}(\ket\psi)|$ is invariant under local unitaries.
\end{lemma}
Lemma~\ref{lem:relative-invariance} converts the absence of GHZ-SLOCC vectors in a subspace into the vanishing of the hyperdeterminant throughout that subspace. This polynomial criterion will be used repeatedly to constrain and classify the ranges considered in Sec.~\ref{subsec:ghz-slocc-free-range}.

The vanishing of the hyperdeterminant determines the GHZ-SLOCC class, but it does not determine the finer LU class. The following criterion distinguishes GHZ-LU states using the one-qubit marginals.
\begin{lemma}
\label{le:LU-GHZ}
Let $\ket\psi$ be a normalized three-qubit state, let $\rho=\proj{\psi}$,
and set $\Delta_X:=\det\rho_X$ for $X=A,B,C$. Then $\ket\psi$ is a
GHZ-LU state if and only if $\Delta_A=\Delta_B=\Delta_C>0$ and
$\Delta_A=|\operatorname{Det}(\ket\psi)|$.
\end{lemma}
\begin{proof}
	All conditions in the lemma are LU invariant. For necessity, we use the
	representative $a\ket{000}+b\ket{111}$ and obtain the stated equalities
	directly. For sufficiency, we write $\ket\psi$ in the Acin form
	\cite{acin2000generalized}
	$$
	\ket{\psi_c}=\lambda_0\ket{000}
	+\lambda_1e^{i\varphi}\ket{100}
	+\lambda_2\ket{101}+\lambda_3\ket{110}+\lambda_4\ket{111},
	$$
	where $\lambda_i\geq0$. We assume that
	$\Delta_A=\Delta_B=\Delta_C>0$ and
	$\Delta_A=|\operatorname{Det}(\ket{\psi_c})|$. For the Acin form, we have
	$\operatorname{Det}(\ket{\psi_c})=\lambda_0^2\lambda_4^2$ and
	$\Delta_A=\lambda_0^2(\lambda_2^2+\lambda_3^2+\lambda_4^2)$. These
	equalities imply that $\lambda_0\lambda_4\neq0$ and 
	$\lambda_2=\lambda_3=0$. We then have
	$\Delta_A=\lambda_0^2\lambda_4^2$ and
	$\Delta_B=(\lambda_0^2+\lambda_1^2)\lambda_4^2$, so
	$\Delta_A=\Delta_B$ implies
	$\lambda_1=0$. Hence $\ket{\psi_c}$ is a generalized GHZ state, and
	we obtain the desired LU equivalence.
\end{proof}

We now introduce the structural tools used in the range analysis.
Vectors in the same range may be biseparable across different
bipartitions, which makes a direct range classification difficult.
The following result converts this pointwise condition into one fixed
bipartition. It also shows that a bipartite linear space consisting
entirely of product vectors must have one fixed local factor.

\begin{lemma}[\cite{harris1992algebraic}]
\label{lem:biseparable-product-loci}
(i) Let $\cV\subseteq\cH_A\ox\cH_B\ox\cH_C$ be a linear subspace.
If every nonzero vector in $\cV$ is product across at least one of the
three bipartitions, then there is one fixed bipartition across which
every vector in $\cV$ is product.

(ii) If every nonzero vector in a linear subspace
$\cX\subseteq\cH_1\ox\cH_2$ is a decomposable tensor, then either
$\cX=x_0\ox\cU_2$ or $\cX=\cU_1\ox y_0$ for suitable nonzero vectors
$x_0,y_0$ and suitable linear subspaces $\cU_1,\cU_2$.
\end{lemma}

We use Lemma~\ref{lem:biseparable-product-loci} to reduce pointwise
biseparability to one fixed bipartition and local factor. These
reductions underlie the range classification in
Lemmas~\ref{le:noSLOCC_dim<4=1} and~\ref{le:noSLOCC_dim<4} and the
canonical subspaces in Proposition~\ref{pro:GHZ_free_subspace_T_C}.
Proposition~\ref{pro:T_Ct_2UDA} then establishes 2-UDA on the
resulting canonical ranges, while Lemma~\ref{le:lr+1-ls} shows that
this property depends only on the range.

We next need a criterion ensuring that a prescribed range contains a
product vector. The following projective-dimension bound guarantees this
existence directly from the dimension of the range and the local
dimensions. It applies without any assumption on the vectors already
known to lie in the range.

\begin{lemma}[\cite{harris1992algebraic,parthasarathy2004maximal}]
\label{lem:product-vectors-from-dimension}
Let $d_i=\dim\cH_{A_i}$ and $d=\prod_{i=1}^n d_i$. Every subspace
$\cV\subseteq\cH$ satisfying
$
\dim\cV\geq d-\sum_{i=1}^n(d_i-1)
$
contains a nonzero fully product vector.
\end{lemma}

We use Lemma~\ref{lem:product-vectors-from-dimension} to show that the
ranges under consideration intersect the relevant product varieties.
The resulting product vectors enable the local normalizations and
range reductions that lead to the canonical inclusions in
Proposition~\ref{pro:GHZ_free_subspace_T_C}.

To obtain the stronger GHZ-LU conclusion and the sharp high-rank
obstruction, we use characteristic classes to detect unavoidable zeros
of a vector bundle section. The following facts provide the
topological framework for that argument.

\begin{lemma}[\cite{bott1982differential}]
\label{lem:characteristic-class-tools}
Let $E$ be a complex rank-$q$ vector bundle over a closed oriented
manifold $M$.

(i) The zero locus of a transverse section of $E$ represents the
Poincar\'e dual of $c_q(E)$. If a section is nowhere zero outside a
closed neighborhood $N$, its relative Euler class localizes this
class to $H^{2q}(M,M\setminus\operatorname{int}N;\bbZ)$.

(ii) Chern and Stiefel--Whitney classes satisfy the Whitney product
formula and $w_2(E_{\bbR})=c_1(E)\pmod2$, where $E_{\bbR}$ denotes
the underlying real vector bundle of $E$.

(iii) Let $\pi:\bbP(E)\ra M$ be the projective bundle of lines in
$E$, let $\mathscr T$ be its tautological line bundle, and set
$\xi=c_1(\mathscr T^*)$. The tautological exact sequence determines
the Chern-class relation on $\bbP(E)$. Moreover, $\xi$ restricts to
the hyperplane class on every fiber and
$\pi_*(\xi^{q-1})=1$.
\end{lemma}

Lemma~\ref{lem:characteristic-class-tools} is applied in the proof of
Theorem~\ref{th:five_dim_contains_GHZ_LU}. It provides the relative
Euler-class framework for a zero set whose points away from the
product endpoints parametrize GHZ-LU vectors. The involution and
parity calculation  then force such a point. This proves that every
five-dimensional three-qubit range contains a GHZ-LU vector. Combining 
with Lemmas~\ref{le:lr+1-ls} and~\ref{le:3-qubit_UDA}, this range
statement rules out 2-UDA states of rank at least five.

\section{Range Structure and Criteria for Three-Qubit 2-UDA States}
\label{sec:three-qubit-2UDA-characterization}

In this section, we ananyze the three-qubit 2-UDA states from the perspective of range structure. We first
analyze the structure of GHZ-SLOCC-free ranges.
The resulting canonical inclusions settle the rank-one, rank-three,
and rank-four cases in this class. Rank-two states exhibit a finer
behavior, since a GHZ-SLOCC-free range may correspond to either a
2-UDA or a non-2-UDA state. We therefore derive a necessary and
sufficient range criterion for every rank-two state. For the remaining rank-three and rank-four states, 
we obtain 
range-restricted criterion for tripartite states through
finitely many exact SDP optimal values. We
finally prove that every three-qubit state of rank at least five is not
2-UDA. These results establish the sharp universal rank bound and
complete the statewise three-qubit framework.

\subsection{GHZ-SLOCC-Free Ranges and 2-UDA States of Rank at Most Four}
\label{subsec:ghz-slocc-free-range}

We begin with a structurally tractable class of three-qubit ranges.
If $\cR(\rho)$ contains a GHZ-LU state, then $\rho$ cannot be 2-UDA
\cite{qiu2026mixed}. Example~\ref{exp:GHZ-LU_not2UDA} shows that the
absence of GHZ-LU vectors from the range is not sufficient. We
therefore impose the stronger condition that $\cR(\rho)$ contains no
GHZ-SLOCC vector. We classify the resulting ranges up to local
unitary equivalence and system permutation, and then apply this
structure to the UDA problem.

We first consider states whose ranges contain a GHZ-LU state.
Lemma \ref{le:lr+1-ls} shows that if $\r$ is 2-UDA, then $\cR(\r)$ does
not contain any GHZ-LU states. The following example shows that the converse
does not hold. In particular, the absence of GHZ-LU states from the range does not guarantee
that a three-qubit state is 2-UDA.
\begin{example}
\label{exp:GHZ-LU_not2UDA}
Let $H = 2\s_x \ox \s_x \ox \s_y - 2\s_x \ox \s_z \ox \s_x + \s_y \ox \s_x \ox \s_x - \s_y \ox \s_y \ox \s_z - 2\s_z \ox \s_x \ox \s_z\in \cN_2$ be a Hermitian operator acting on $\cH_{ABC}$.  Suppose $H_+$ and $H_-$ are the positive and negative parts of $H$, respectively. Then $H_\pm \ge 0$ and $H=H_+-H_-$. Let
$\r=\frac{H_-}{12}$ and $ \s=\frac{H_+}{12}$.
Then $\s$ is 2-compatible with $\r$ and
$\s-\r=\frac{H}{12}\neq0$, so $\r$ is not 2-UDA.
One can verify that $\cR(\r)$ does not contain any GHZ-LU states by Lemma~\ref{le:LU-GHZ}.
  \qed
\end{example} 
Clearly, if $\cR(\rho)$ contains no GHZ-SLOCC states, then it contains no
GHZ-LU states. The converse does not generally hold. This motivates us to
consider the UDA property of states whose ranges contain no GHZ-SLOCC states. The range of a three-qubit state is naturally a subspace of $(\bbC^2)^{\ox 3}$. Since the UDA property of a state is invariant up to LU
equivalence and system permutations, 
we now classify GHZ-SLOCC-free three-qubit subspaces up to such operations, and then apply the classification of range 
to the UDA problem of quantum states.

First, we show that one fully product vector imposes a common coordinate
restriction on a GHZ-SLOCC-free subspace, while two linearly
independent fully product vectors further constrain its dimension and
canonical form.
\begin{lemma}
\label{le:noSLOCC_dim<4=1}
Let $\cV\subseteq\cH_{ABC}$ be a GHZ-SLOCC-free subspace.

(i) If $\cV$ contains the fully product vector $\ket{000}$, then $\cV$ is orthogonal to 
$\ket{111}$ and at least one of $\ket{011}$, $\ket{101}$ and $\ket{110}$.
In particular, up to a system permutation, $\cV$ is orthogonal to
$\ket{011}$ and $\ket{111}$.

(ii) If $\cV$ contains two linearly independent fully product vectors, then $\dim \cV\le4$.

(iii) If $\cV$ contains two linearly independent fully product vectors and $\dim \cV=4$, then up to SLOCC equivalence and a system permutation, $\cV=\Span\{\ket{000},\ket{001},\ket{100},\ket{101}\}$, $\Span\{\ket{000},\ket{001},\ket{010},\ket{100}\}$, or $\bbC^2\ox\Span\{\ket{00},\ket{01}+\ket{10}\}$.
\end{lemma}
The proof of Lemma~\ref{le:noSLOCC_dim<4=1} is given in
Appendix~\ref{app:proofs-ghz-slocc-free-range}.

We next present two consequences for a GHZ-SLOCC-free three-qubit
subspace. Such a subspace has dimension at most four. Moreover, in
dimensions three and four, it necessarily contains a fully product vector.
\begin{lemma}
\label{le:noSLOCC_dim<4}
Let $\cV\subseteq\cH_{ABC}$ be a GHZ-SLOCC-free subspace. Then 

(i) $\dim \cV\leq 4$.

(ii)\label{le:noSLOCC_dim<4_product}
If $\dim\cV=3$ or $4$, then $\cV$ contains a fully product vector.
\end{lemma}
The proof of Lemma~\ref{le:noSLOCC_dim<4} is given in
Appendix~\ref{app:proofs-ghz-slocc-free-range}. 

Lemma~\ref{le:noSLOCC_dim<4} (ii) provides the starting point
for the canonical classification below. We first map the fully product
vector to $\ket{000}$ by local unitaries, after which
Lemma~\ref{le:noSLOCC_dim<4=1} (i) places the subspace in a smaller
coordinate subspace. Under this restriction, the vanishing of the
hyperdeterminant forces a relevant projection of the subspace to
consist entirely of decomposable two-qubit tensors.
Lemma~\ref{lem:biseparable-product-loci}(ii) then yields a fixed local
factor, which leads to the canonical inclusions into $\cT$ or
$\cC_\theta$ in Proposition~\ref{pro:GHZ_free_subspace_T_C}.
Besides, the product-vector conclusion does not extend to dimensions one and
two. This is immediate in the one-dimensional case. The following is
a counterexample for $\dim \cV=2$.
\begin{example}
\label{exp:two_dim_no_GHZ_product}
Let $\cV=\Span\{u,v\}\subseteq\cH_{ABC}$, where
\begin{eqnarray}
\label{eq:u=ket001}
&&u=|001\rangle+|010\rangle+|100\rangle,
   \\
   && 
\label{eq:v=ket011}v=|011\rangle+|101\rangle+4|110\rangle.
\end{eqnarray}
We claim that $\cV$ contains neither GHZ-SLOCC states nor fully
product states. For any
$|\phi\rangle=\alpha u+\beta v\in\cV$, direct substitution into
the Cayley hyperdeterminant shows
$\operatorname{Det}(\ket\phi)=0$. Hence, by
Lemma~\ref{lem:relative-invariance}, $\cV$ contains no
GHZ-SLOCC states.

If a nonzero $|\phi\rangle=\alpha u+\beta v$ were fully product, then
its coefficient matrix across the bipartition $A:BC$,
$$
M_{A:BC}(\ket\phi)=
\bma
0 & \alpha & \alpha & \beta \\
\alpha & \beta & 4\beta & 0
\ema
$$
would have rank one. However, its minors in columns $(1,2)$ and $(2,4)$
are $-\alpha^2$ and $-\beta^2$, respectively. Their vanishing forces
$\alpha=\beta=0$, a contradiction. Thus $\cV$ contains no fully
product states.
\end{example}

We next identify the canonical subspaces that govern the remaining
classification. Lemma~\ref{le:noSLOCC_dim<4=1} (i) constrains a
GHZ-SLOCC-free subspace once a fully product vector is fixed, while
Lemma~\ref{le:noSLOCC_dim<4_product} (ii) guarantees such a vector in
dimensions three and four. We define
\begin{eqnarray}
\label{eq:cT}
&&\cT:=\Span\{\ket{000}, \ket{100},
\ket{010}, \ket{001}\},
\\
\label{eq:cC_t}
&&\cC_\t:=\bbC^2\ox \Span\{\ket{00}, 
\cos\t\ket{01}+\sin\t\ket{10}\},\; \t\in [0,\frac{\pi}{4}].
\end{eqnarray}
\begin{proposition}
\label{pro:GHZ_free_subspace_T_C}
The canonical subspaces $\cT$ and $\cC_\theta$ in
\eqref{eq:cT} and \eqref{eq:cC_t} are GHZ-SLOCC-free.
Moreover, let $\cV\subseteq\cH_{ABC}$ be a GHZ-SLOCC-free subspace.
Then, up to LU equivalence and a system permutation,

(i) If $\cV$ contains a fully product vector, then
$\cV\subseteq\cT$ or $\cV\subseteq\cC_\theta$.

(ii) If $\dim\cV=3$ or $4$, then
$\cV\subseteq\cT$ or $\cV\subseteq\cC_\theta$. If
$\dim\cV=4$, the corresponding inclusion is an equality.
\end{proposition}
The proof of Proposition~\ref{pro:GHZ_free_subspace_T_C} is presented
in Appendix~\ref{app:proofs-ghz-slocc-free-range}.

Proposition \ref{pro:GHZ_free_subspace_T_C} (i) will be used below to treat two-dimensional subspaces that
contain a fully product vector. Part (ii) reduces the study of
three- and four-dimensional GHZ-SLOCC-free ranges to states whose
ranges are contained in $\cT$ or $\cC_\theta$. We combine this reduction with
Proposition~\ref{pro:T_Ct_2UDA} to obtain their 2-UDA property.
By Lemma~\ref{le:noSLOCC_dim<4} (i), no higher-dimensional
GHZ-SLOCC-free subspace exists. The one-dimensional case is immediate,
whereas the two-dimensional case is not classified here.

We apply these canonical inclusion results to the range of three-qubit
mixed states to analyze their UDA property.
We show that every state whose range is a subspace of $\cT$ or
$\cC_\t$ is 2-UDA.
\begin{proposition}
\label{pro:T_Ct_2UDA}
If a three-qubit state $\rho$ satisfies $\cR(\rho)\subseteq \cT$ or $\cR(\rho)\subseteq \cC_\t$, then $\rho$ is 2-UDA.
\end{proposition}
The proof of Proposition~\ref{pro:T_Ct_2UDA} is shown in
Appendix~\ref{app:proofs-ghz-slocc-free-range}.

This proof illustrates a general range-restriction and injectivity
criterion for arbitrary finite-dimensional multipartite systems. Let
$\cU$ be a subspace of the global Hilbert space and suppose
$\cR(\rho)\subseteq\cU$. Suppose the ranges of the prescribed
$k$-marginals, together with positivity, force
$\cR(\sigma)\subseteq\cU$ for every state $\sigma$ that is
$k$-compatible with $\rho$. If the
restriction of the marginal map $\cM_k$ to $\operatorname{Herm}(\cU)$ is
injective, then $\rho$ is $k$-UDA. Indeed, for every $k$-compatible
state $\sigma$, the difference $H=\sigma-\rho$ belongs to
$\operatorname{Herm}(\cU)\cap\ker\cM_k$ and hence vanishes. In
Proposition~\ref{pro:T_Ct_2UDA}, we verify the range condition from
the kernels of the 2-marginals and the injectivity condition by direct
partial traces.

Proposition~\ref{pro:T_Ct_2UDA} covers rank-two ranges contained in the
canonical subspaces $\cT$ and $\cC_\theta$, but not all
two-dimensional GHZ-SLOCC-free ranges. The following examples show
that the absence of GHZ-SLOCC states from a rank-two range alone does
not determine the UDA property.
\begin{example}
\label{exp:rank_two_GHZ_free_UDA}
Let $\omega=\exp(2\pi i/3)$, 
$\ket{e_1}=\ket{001}+\ket{010}+\ket{100}$,
and
$\ket{e_2}=\omega\ket{011}+\ket{101}+\omega^2\ket{110}$.
We consider the two rank-two states
$$
\rho_1=\frac{1}{6}\bigl(\ket{e_1}\bra{e_1}
+\ket{e_2}\bra{e_2}\bigr),
\qquad
\rho_2=\frac{1}{6}\ket u\bra u+\frac{1}{36}\ket v\bra v,
$$
where $u$ and $v$ are shown in \eqref{eq:u=ket001} and
\eqref{eq:v=ket011}. We have 
$\operatorname{Det}(a\ket{e_1}+b\ket{e_2})=0$ for all
$a,b\in\bbC$, while Example~\ref{exp:two_dim_no_GHZ_product} shows
the same conclusion for $\Span\{u,v\}$. Thus, by
Lemma~\ref{lem:relative-invariance}, both ranges are GHZ-SLOCC-free. By 
Proposition~\ref{pro:rank_two_2UDA_parameter}, the corresponding test
admits a nonzero compatible perturbation for $\rho_1$ but none for
$\rho_2$. Consequently, $\rho_1$ is not 2-UDA whereas $\rho_2$ is
2-UDA.
\end{example}
Example~\ref{exp:rank_two_GHZ_free_UDA} shows that the absence of
GHZ-SLOCC vectors does not determine the UDA property at rank two. The
preceding structural results nevertheless settle all other possible
ranks in the GHZ-SLOCC-free class. Lemma~\ref{le:noSLOCC_dim<4} (i)
first restricts the rank to at most four. At rank one, the range
contains no GHZ-LU vector, so Lemma~\ref{le:3-qubit_UDA} establishes
2-UDA. At ranks three and four,
Lemma~\ref{le:noSLOCC_dim<4} (ii) and
Proposition~\ref{pro:GHZ_free_subspace_T_C} reduce the range to a
subspace of $\cT$ or $\cC_\theta$. Proposition~\ref{pro:T_Ct_2UDA}
then establishes 2-UDA. To conclude, we have the following fact.
\begin{theorem}
\label{th:GHZ_free}
Let $\r$ be a three-qubit mixed state whose range contains no GHZ-SLOCC
states. If $\rank \r=1,3$ or $4$, then $\r$ is 2-UDA.
\end{theorem}
The proof of Theorem~\ref{th:GHZ_free} is given in
Appendix~\ref{app:proofs-ghz-slocc-free-range}.

Theorem~\ref{th:GHZ_free} therefore isolates rank two as the only
GHZ-SLOCC-free case in which the range condition alone does not settle
the UDA property. Since Example~\ref{exp:rank_two_GHZ_free_UDA} shows
that both behaviors occur, we analyze rank-two states separately.

\subsection{A Range Criterion for Rank-Two Three-Qubit 2-UDA States}
\label{subsec:rank-two-three-qubit-states}

For a rank-two state $\r$, if $\cR(\r)\subset \cT$ in \eqref{eq:cT} or
$\cR(\r)\subset \cC_\t$ in \eqref{eq:cC_t} up to LU equivalence and system permutation, then
$\r$ is 2-UDA. Outside these canonical cases, the UDA property is not
determined by the absence of GHZ-SLOCC states from the range alone. We
next derive a necessary and sufficient condition for a rank-two state
to be 2-UDA.

We introduce some necessary notations. The criterion below can be stated without fixing coordinates for the
two-dimensional range. Let
$
\rho=\mathsf E\mathsf E^\dagger$,
 $\mathsf E=(\ket{x},\ket{y})$, and
$
\cL=\cR(\rho),
$
where $\ket{x}$ and $\ket{y}$ are linearly independent vectors. We
use the three-qubit spin flip $\Theta_3$ in
\eqref{eq:n-qubit-spin-flip}. For a vector $\ket v$, a Hermitian
operator $X$, and a subspace $\cU$, we write
$\ket{\widetilde v}:=\Theta_3\ket v$,
$\widetilde X:=\Theta_3X\Theta_3^{-1}$, and
$\widetilde{\cU}:=\Theta_3\cU$. We set
$$
\widetilde{\mathsf E}
=(\ket{\widetilde x},\ket{\widetilde y}).
$$
The scalar needed in our main result is
\begin{equation}
	\label{eq:delta_E_rank_two}
	\delta_{\mathsf E}
	:=
	\norm{x}^{2}\norm{y}^{2}
	-|\braket{x}{y}|^{2}
	-|\braket{\widetilde x}{y}|^{2}.
\end{equation}
Only the distinction between $\delta_{\mathsf E}=0$ and
$\delta_{\mathsf E}>0$ will matter. Its proof below shows that this
distinction depends only on $\cL$, although the value of
$\delta_{\mathsf E}$ is not determined by $\cL$ alone.
When $\delta_{\mathsf E}>0$, for $F=F^\dagger\in\bbC^{2\times2}$ and
$G=G^T\in\bbC^{2\times2}$, we set
\begin{equation}
	\label{eq:Phi_rank_two}
	\Phi_{\mathsf E}(F,G)
	:=
	\mathsf E F\mathsf E^\dagger
	-\widetilde{\mathsf E}\,\overline F\,\widetilde{\mathsf E}^{\dagger}
	+\mathsf E G\widetilde{\mathsf E}^{\dagger}
	+\widetilde{\mathsf E}G^\dagger\mathsf E^\dagger
\end{equation}
and define the real linear space 
\begin{equation}
	\label{eq:F_space_rank_two}
	\cF_{\mathsf E}
	:=
	\left\{
	F=F^\dagger:
	\ \Phi_{\mathsf E}(F,G)\in\cN_2
	\text{ for some }G=G^T
	\right\}.
\end{equation}
Thus $\cF_{\mathsf E}$ consists precisely of the upper-left
$2\times2$ blocks of the marginal-preserving perturbations relevant to
$\rho$. Each rank-two state admits such a factorization, for example,
by its spectral decomposition. Moreover, any two factorizations of the
same state differ by a unitary transformation on the right. Hence the
choice of $\mathsf E$ involves no loss of generality. We verify that the conditions below are independent not only of this
factorization, but also of the choice of the full-rank state with range
$\cL$.

\begin{proposition}
	\label{pro:rank_two_2UDA_parameter}
	Let $\rho$ be a rank-two three-qubit state, and let
	$\rho=\mathsf E\mathsf E^\dagger$ be any factorization as above. Then
	$\rho$ is 2-UDA if and only if both of the following conditions hold.
	
	(i) $\cL=\cR(\rho)$ contains no GHZ-LU state.
	
	(ii) either $\delta_{\mathsf E}=0$, or
	$
	\det F\leq0
	\text{ for every }F\in\cF_{\mathsf E}.
	$
	Both conditions are independent of the chosen factorization and depend
	only on the range $\cL$.
	\end{proposition}
The proof of Proposition~\ref{pro:rank_two_2UDA_parameter} is
presented in Appendix~\ref{app:proofs-rank-two-three-qubit-states}.

Proposition \ref{pro:rank_two_2UDA_parameter} can be used to verify the
2-UDA property of a rank-two state as follows. If
$\delta_{\mathsf E}=0$, no further test is required. Suppose that
$\delta_{\mathsf E}>0$, and 
$$
F(z)=
\bma
z_1&z_3-iz_4\\
z_3+iz_4&z_2
\ema,
\qquad
G(w)=
\bma
w_1+iw_2&w_5+iw_6\\
w_5+iw_6&w_3+iw_4
\ema,
$$
where $z=(z_1,\ldots,z_4)^T\in\bbR^4$ and
$w=(w_1,\ldots,w_6)^T\in\bbR^6$. Let $P_1,\ldots,P_9$ be the nine
weight-one three-qubit Pauli strings. 
By Eq.~\eqref{eq:n-qubit-spin-flip-Pauli},
$\Phi_{\mathsf E}(F(z),G(w))$ is odd under the spin flip, so its
Pauli expansion contains only terms of weights one and three. Hence the
condition
$\Phi_{\mathsf E}(F(z),G(w))\in\cN_2$ is equivalent to
\begin{eqnarray}
	\label{eq:tr(PiPhi_E)}
	\tr\!\left(P_i\Phi_{\mathsf E}(F(z),G(w))\right)=0,
	\qquad i=1,\ldots,9.
\end{eqnarray}
Let $e_\alpha$ and $f_\beta$ denote the standard basis vectors of
$\bbR^4$ and $\bbR^6$, respectively. We define
$M_F\in\bbR^{9\times4}$ and $M_G\in\bbR^{9\times6}$ by
$$
(M_F)_{i\alpha}
:=
\tr\!\left(P_i\Phi_{\mathsf E}(F(e_\alpha),0)\right),
\qquad
(M_G)_{i\beta}
:=
\tr\!\left(P_i\Phi_{\mathsf E}(0,G(f_\beta))\right).
$$
These entries are real because the operators inside each trace are
Hermitian.
The equations in \eqref{eq:tr(PiPhi_E)} are equivalent to
$M_Fz+M_Gw=0$. Let 
$$
\cZ_{\mathsf E}
:=
\left\{
z\in\bbR^4:
\ M_Fz+M_Gw=0
\text{ for some }w\in\bbR^6
\right\}.
$$
Equivalently, $\cZ_{\mathsf E}$ is the projection of
$\ker(M_F\;M_G)$ onto the $z$-coordinates. Let
$Z\in\bbR^{4\times\dim\cZ_{\mathsf E}}$ be any matrix whose columns
form a basis of $\cZ_{\mathsf E}$. Since
$$
\det F(z)=z_1z_2-z_3^2-z_4^2=z^TJ_Fz,
\qquad
J_F=\tfrac12\sigma_x\oplus(-I_2),
$$
condition (ii) in Proposition \ref{pro:rank_two_2UDA_parameter} is
equivalent to
$
Z^TJ_FZ\leq0.
$
Thus we can directly verify the coordinate-free criterion from the
matrices $M_F$, $M_G$, and $Z$.

\subsection{An SDP Criterion for Tripartite 2-UDA States}
\label{subsec:support-reduced-SDP}

The preceding analytic criteria rely on special canonical forms or on
structural properties specific to two-dimensional ranges. They do not
directly cover rank-three and rank-four ranges that contain GHZ-SLOCC vectors but no
GHZ-LU vector. We now replace structural classification by a statewise
convex criterion. Generalized state inversion confines every compatible
state to an enlarged range, while the remaining marginal-preserving
directions form a compact spectrahedron. The resulting criterion is
stated for arbitrary finite-dimensional tripartite states and, in
particular, decides every prescribed state in the remaining
three-qubit regime. Any positive optimum produces an explicit
compatible state.

Let $\rho$ be a tripartite state on $\cH$. For each
$Y\in\{A,B,C\}$ and every linear operator $Z$ on $\cH_Y$, we define
the local universal inverter \cite{rungta2001universal} by
$\mathcal I_Y(Z):=\tr(Z)I_Y-Z$. For every linear operator $X$ on
$\cH$, we define the generalized tripartite state inversion by
\begin{eqnarray}
\label{eq:generalized-tripartite-state-inversion}
\widetilde X
:=
(\mathcal I_A\ox\mathcal I_B\ox\mathcal I_C)(X).
\end{eqnarray} 
For three qubits and Hermitian $X$, Eq.~
\eqref{eq:n-qubit-spin-flip} implies
$\widetilde X=\Theta_3X\Theta_3^{-1}$. We retain the inverter formulation here
because Proposition~\ref{pro:support-reduced-SDP} applies to arbitrary
finite local dimensions.
In the proof of Proposition~\ref{pro:support-reduced-SDP}, we will show
that this transformation preserves positive semidefiniteness. Besides, one can verify that 
$H\in\cN_2$ in \eqref{eq:N_k_general} implies that $\widetilde H=-H$. 
We set $r:=\rank\rho$ and $\cL:=\cR(\rho)$, and denote the orthogonal
projection onto $\cL$ by $P_{\cL}$. We define the maximally mixed
state on $\cL$ and the state-inversion-enlarged range by
\begin{eqnarray}
\label{eq:state-inversion-reduced-support}
\tau_{\cL}
:=
\frac1rP_{\cL},
\qquad
\widehat{\cL}
:=
\cR\bigl(\tau_{\cL}+\widetilde{\tau_{\cL}}\bigr),
\qquad
d_{\cL}:=\dim\widehat{\cL}.
\end{eqnarray}
We choose a real basis $H_1,\ldots,H_m$ of
$\cN_2\cap\operatorname{Herm}(\widehat{\cL})$ and an isometry
\begin{eqnarray}
\label{eq:support-reduced-isometry}
U_{\cL}:\bbC^{d_{\cL}}\ra\cH,
\qquad
\cR(U_{\cL})=\widehat{\cL},
\qquad
U_{\cL}^{\dagger}U_{\cL}=I_{d_{\cL}},
\qquad
U_{\cL}U_{\cL}^{\dagger}=P_{\widehat{\cL}}.
\end{eqnarray}
Here $P_{\widehat{\cL}}$ is the orthogonal projection onto
$\widehat{\cL}$. Equivalently, the columns of $U_{\cL}$ form an
orthonormal basis of $\widehat{\cL}$. Let
$R_{\cL}:=U_{\cL}^{\dagger}\tau_{\cL}U_{\cL}$ and 
$G_j:=U_{\cL}^{\dagger}H_jU_{\cL}$. We define the spectrahedron
\begin{eqnarray}
\label{eq:state-inversion-spectrahedron}
\Omega_{\cL}
:=
\Big\{
x\in\bbR^m:
R_{\cL}+\sum_{j=1}^m x_jG_j\geq0
\Big\}.
\end{eqnarray}

\begin{proposition}
\label{pro:support-reduced-SDP}
Let $\rho$ be an arbitrary finite-dimensional tripartite state, and
let the preceding objects be defined from $\cL=\cR(\rho)$. Then 

(i) Every state $\sigma$ that is 2-compatible with $\rho$ satisfies
$\cR(\sigma)\subseteq\widehat{\cL}$.

(ii) The spectrahedron $\Omega_{\cL}$ is compact, and $\rho$ is 2-UDA
if and only if $\Omega_{\cL}=\{0\}$.

(iii) Suppose $x^*\in\Omega_{\cL}\setminus\{0\}$ and set
$H^*:=\sum_{j=1}^m x_j^*H_j$. Then
$\rho^*:=\rho+\epsilon H^*$, where
$\epsilon:=\frac r2\lambda_{\min}^{+}(\rho)$ and
$\lambda_{\min}^{+}(\rho)$ is the smallest strictly positive
eigenvalue of $\rho$, is a state distinct from and 2-compatible with
$\rho$.
\end{proposition}
The proof of Proposition~\ref{pro:support-reduced-SDP} is presented in
Appendix~\ref{app:proofs-support-reduced-SDP}.

We next formulate Proposition~\ref{pro:support-reduced-SDP} (ii) in
terms of finitely many SDP optimal values. Suppose $m\geq1$, and
let $e_1,\ldots,e_m$ be the standard basis of $\bbR^m$. We set
$c^{(0)}:=-(1,\ldots,1)^T$ and $c^{(j)}:=e_j$ for $j=1,\ldots,m$.
For $a=0,\ldots,m$, we consider the semidefinite program
\begin{eqnarray}
\label{eq:support-reduced-SDP-objectives}
\delta_a
:=
\max_{x\in\bbR^m}
\Big\{
(c^{(a)})^Tx:
R_{\cL}+\sum_{j=1}^m x_jG_j\geq0
\Big\}.
\end{eqnarray}
Since $0\in\Omega_{\cL}$, every $\delta_a$ is nonnegative. If
$\Omega_{\cL}$ contains a nonzero $x$, then either $x$ has a positive
coordinate or all its coordinates are nonpositive and
$-(1,\ldots,1)^Tx>0$. Hence
$$
\rho\text{ is 2-UDA}
\quad\Longleftrightarrow\quad
\max_{0\leq a\leq m}\delta_a=0.
$$
Thus, for exact input data and exact optimal values, $\rho$ is 2-UDA
if and only if all $m+1$ optimal values vanish. A positive optimum and
its optimizer produce the competing state in
Proposition~\ref{pro:support-reduced-SDP} (iii).

We briefly record the size of the criterion. Let
$d_X:=\dim\cH_X$ for $X=A,B,C$. We have
$
\dim_{\bbR}\cN_2=(d_A^2-1)(d_B^2-1)(d_C^2-1)$ and
$m\leq
\min\left\{
(d_A^2-1)(d_B^2-1)(d_C^2-1),
d_{\cL}^2-1
\right\}.
$
Each SDP in \eqref{eq:support-reduced-SDP-objectives} has $m$ real
variables and one $d_{\cL}\times d_{\cL}$ linear matrix inequality.
For three qubits, this inversion agrees with the spin flip in
\eqref{eq:n-qubit-spin-flip} and preserves rank, so
$d_{\cL}\leq\min\{2r,8\}$ and $m\leq27$. Consequently, at most $28$
SDPs occur in the characterization, each with a linear matrix
inequality of order at most $\min\{2r,8\}$. 
The range restriction in \eqref{eq:state-inversion-reduced-support} is the
main structural content of this criterion. It confines all compatible
states to a subspace determined only by the range, while the variables
in \eqref{eq:state-inversion-spectrahedron} describe genuine three-body
directions invisible to the 2-marginals. In particular, it provides an exact
characterization of individual rank-three or rank-four states in the
structurally open regime.

The range-restricted criterion therefore completes the statewise
analysis of the remaining low-rank three-qubit regime without requiring
a canonical classification of their ranges. It remains to determine
whether rank alone eventually precludes 2-UDA. We now return to
the structure of three-qubit ranges and establish the sharp universal rank
bound.

\subsection{The Sharp Rank Bound for Three-Qubit 2-UDA States}
\label{subsec:three-qubit-rank-bound}

The preceding results treat all three-qubit states of rank at most
four, either analytically by their range structure or statewise
through the range-restricted criterion. We now ask whether a state of
higher rank can be 2-UDA. A generalized GHZ-LU vector determines a
non-2-UDA pure state by Lemma~\ref{le:3-qubit_UDA}. If such a vector
lies in the range of a mixed state, Lemma~\ref{le:lr+1-ls} excludes
2-UDA for that mixed state. We therefore determine how large a
three-qubit subspace can be without containing a generalized GHZ-LU
vector.

We previously showed that three-qubit states of rank seven or eight are
not 2-UDA \cite{qiu2026mixed}. The following fact settles the remaining
ranks five and six and shows the sharp rank bound for 2-UDA states. We
combine projective geometry with a characteristic-class argument to
establish the required existence of a GHZ-LU vector.
\begin{theorem}
\label{th:five_dim_contains_GHZ_LU}
Let $\cV\subseteq\cH_{ABC}$ be a five-dimensional three-qubit
subspace. Then there exist $a,b\in\bbC$ with $ab\neq0$ such that
$\cV$ contains a vector LU equivalent to
$a\ket{000}+b\ket{111}$. Consequently, every three-qubit state $\rho$
of rank at least five is not 2-UDA.
\end{theorem}
The proof of Theorem~\ref{th:five_dim_contains_GHZ_LU} is presented in
Appendix~\ref{app:proofs-three-qubit-rank-bound}.

We associate with $\cV$ a section of a rank-three vector bundle over a
projective bundle that parametrizes complementary product directions.
A zero away from the two product endpoint sections determines a GHZ-LU
vector in $\cV$. If every zero were confined to the endpoints, the
orthogonal-complement involution and a parity argument would force the
relevant characteristic number to be divisible by four. A direct
Chern-class computation shows that this number equals six, which is
impossible. We then obtain the rank consequence by applying the
existence result to a five-dimensional subspace of $\cR(\rho)$ and
using Lemmas~\ref{le:lr+1-ls} and~\ref{le:3-qubit_UDA}. 

Theorem~\ref{th:five_dim_contains_GHZ_LU} shows that every 2-UDA
three-qubit state has rank at most four. This obstruction depends only
on the range and not on the nonzero eigenvalues of the state. In this
sense, the 2-UDA property is highly restrictive and therefore rare
among mixed states. The set of rank-deficient density matrices has
Lebesgue measure zero in the affine space of density operators and
hence has measure zero under every full-dimensional unitarily invariant
measure that is absolutely continuous with respect to Lebesgue measure.

\section{Extensions to Multipartite Systems}
\label{sec:multipartite-UDA-results}

Having completed the statewise three-qubit analysis, we now extend its
methods to multipartite systems. We first compare the space of
range-constrained Hermitian directions with the kernel of the
marginal map. This comparison yields general high-rank obstructions,
and the multiqubit spin flip sharpens them near the rank threshold
considered below. We then complement these universal obstructions with an
explicit channel-based family whose exact UDA order can be
characterized.

\subsection{High-Rank Obstructions to Multipartite $k$-UDA}
\label{subsec:high-rank-non-2UDA}

We first derive a dimension criterion that applies to arbitrary
finite-dimensional multipartite systems. It requires no canonical form
for the state range. The criterion detects non-UDA states whenever
the range carries more Hermitian directions than the prescribed
marginals can distinguish. For multiqubit systems, we further use the
spin flip to incorporate the position of the kernel and obtain sharper
rank thresholds.
\begin{lemma}
\label{le:dimension_non_UDA_criterion}
Let $\rho$ be an $n$-partite state on
$\cH=\cH_{A_1}\ox\cdots\ox\cH_{A_n}$ of rank $r$, set
$\cL:=\cR(\rho)$, and let
$d_i:=\dim\cH_{A_i}$ for $i\in[n]$. We use the convention that the
product over the empty set is one. Then

(i) Suppose $1\leq k\leq n-1$. If
\begin{eqnarray}
\label{eq:dimension_non_UDA_criterion}
r^2>
\sum_{\substack{\cS\subseteq[n]\\|\cS|\leq k}}
\prod_{i\in\cS}(d_i^2-1),
\end{eqnarray}
then $\rho$ is not $k$-UDA. 
In particular, if all local dimensions are equal to $d$, then the
condition in \eqref{eq:dimension_non_UDA_criterion} becomes
$r^2>\sum_{j=0}^{k}\binom{n}{j}(d^2-1)^j$.

(ii) Suppose that $n\geq2$ and $\rho$ is an $n$-qubit state on $\cH_n$. We set
$\cK:=\ker\rho$, $\widehat{\cK}:=\cK+\Theta_n\cK$, and
$t:=\dim\widehat{\cK}$, where $\Theta_n$ is defined in
\eqref{eq:n-qubit-spin-flip}. If
\begin{eqnarray}
\label{eq:spin-flip-refined-obstruction}
3^n>
\frac{t(2^{n+1}-t+1)}{2},
\end{eqnarray}
then $\rho$ is not $(n-1)$-UDA. When $n$ is even, the same conclusion
holds if equality occurs in
\eqref{eq:spin-flip-refined-obstruction}.
\end{lemma}
The proof of Lemma~\ref{le:dimension_non_UDA_criterion} is given in
Appendix~\ref{app:proofs-high-rank-non-2UDA}. Part (i) converts the multipartite UDA problem
into a comparison between the dimensions of the range-constrained
Hermitian directions and the marginal kernel. We obtain the bound by comparing the
$r^2$-dimensional space $\operatorname{Herm}(\cL)$ with the space of
operator coefficients detected by the $k$-marginals. For $(n-1)$-marginals, the relevant marginal-preserving operators are
the full-weight Pauli directions. Eqs.~\eqref{eq:n-qubit-spin-flip} and
\eqref{eq:n-qubit-spin-flip-Pauli} show that if such a direction
annihilates a subspace $\cK$, then it also annihilates
$\Theta_n\cK$. We use this observation to replace
$\cK$ by the spin-flip invariant space
$\cK+\Theta_n\cK$ in part (ii) and sharpen the dimension comparison.
We now apply both criteria to multiqubit systems.

\begin{theorem}
	\label{th:high_rank_dimension_obstruction}
	Let $n\ge3$, and let $\rho$ be an $n$-qubit state on $\cH_n$. Then
	
(i)  If $\rank\rho\ge 2^n-n$, then $\rho$ is not $k$-UDA for every 
		$1\le k\le n-2$.

(ii) If either $n\ge6$ and $\rank\rho\ge 2^n-n$, or
$n\in\{4,5\}$ and $\rank\rho\ge 2^n-n+1$, then $\rho$ is not
$k$-UDA for every $1\le k\le n-1$.
\end{theorem}
The proof of Theorem~\ref{th:high_rank_dimension_obstruction} is given
in Appendix~\ref{app:proofs-high-rank-non-2UDA}.

From a physical perspective, Theorem~\ref{th:high_rank_dimension_obstruction}
shows that a sufficiently large state range necessarily contains global
perturbation directions invisible to the prescribed marginals. This
conclusion improves Proposition 5 in \cite{qiu2026mixed}. The
spin-flip refinement further shows that the position of the kernel
relative to its spin flip controls the boundary cases. Hence UDA is
intrinsically a low-rank phenomenon in this regime. Local data cannot uniquely
recover a global state once the dimension of its range reaches the stated threshold. 
Motivated by this high-rank obstruction, we ask whether every
$n$-qubit state with $\rank\rho\ge2^n-n$ is not $k$-UDA for any
$1\le k\le n-1$. Theorem~\ref{th:high_rank_dimension_obstruction}
answers this threshold question affirmatively for all $k\le n-2$.
For $k=n-1$, part (ii) gives an affirmative answer when $n\ge6$ and
yields the near-threshold bound
$\rank\rho\ge2^n-n+1$ when $n=4,5$. For $n=3$,
Theorem~\ref{th:five_dim_contains_GHZ_LU} shows that every
three-qubit state of rank at least $5=2^3-3$ is not 2-UDA. Together
with Theorem~\ref{th:high_rank_dimension_obstruction}~(i), this
resolves the threshold question for $n=3$. At the threshold rank
$\rank\rho=2^n-n$, we set $\cK:=\ker\rho$. The same spin-flip
criterion also rules out 3-UDA for $n=4$ when
$\dim(\cK\cap\Theta_4\cK)\ge2$, and it rules out 4-UDA for $n=5$ when
$\cK\cap\Theta_5\cK\neq\{0\}$. Indeed,
$\dim(\cK+\Theta_4\cK)\le6$ in the first case. In the second case,
$\Theta_5^2=-I$ implies that the nonzero intersection has even
dimension, and hence $\dim(\cK+\Theta_5\cK)\le8$.
Lemma~\ref{le:dimension_non_UDA_criterion}(ii) then applies. Thus the
remaining threshold cases are
\begin{eqnarray}
\label{eq:remaining_rank_threshold_cases}
&&n=4,\quad \rank\rho=12,\quad
\dim(\cK\cap\Theta_4\cK)\in\{0,1\},
\nonumber\\
&&n=5,\quad \rank\rho=27,\quad
\cK\cap\Theta_5\cK=\{0\}.
\end{eqnarray}

\subsection{A Channel-Based Family with Exact UDA Order}
\label{subsec:multipartite-non-2UDA}

The preceding results are universal but primarily obstructive. They
identify ranks above which the marginals cannot determine a state, but
they do not show how the exact UDA order depends on finer multipartite
structure. We now complement them with an explicit channel-based
family. The first $n-1$ systems share a classical label, while the last
system carries a channel output determined by that label. We show that
the purity and pairwise overlap of these outputs determine whether the
state is fixed by its $(n-1)$-marginals. We also prove that marginals of
every lower order fail throughout this family.

We first show that a channel maps the coherent and incoherent versions
of a generalized GHZ state to states with identical
$(n-1)$-marginals. From this observation, we obtain the basic compatible
pair used throughout this subsection. We denote the standard basis of
$\bbC^d$ by $\{\ket j\}_{j=1}^d$ and the identity map on $\bbM_d$ by
$\operatorname{id}_{\bbM_d}$. Here a channel is a completely positive
trace-preserving linear map between matrix algebras.
\begin{lemma}
\label{le:rho_sigma_compatible}
Let $n\geq 3$, $d,s\geq 2$, and let $p_j>0$ satisfy
$\sum_{j=1}^d p_j=1$. We set $p:=(p_1,\ldots,p_d)$ and define
$\ket{\psi_p}:=\sum_{j=1}^d\sqrt{p_j}\ket{j}^{\ox n}$. For a channel
$\Lambda:\bbM_d\ra\bbM_s$, we define two $n$-partite states
\begin{eqnarray}
\label{eq:rho_sigma_channel_states}
\rho_{\Lambda,p}^{(n)}
&:=&\sum_{j=1}^d p_j
\big(\ket{j}\bra{j}\big)^{\ox(n-1)}
\ox\Lambda(\ket{j}\bra{j}),\nonumber\\
\sigma_{\Lambda,p}^{(n)}
&:=&(\operatorname{id}_{\bbM_d}^{\ox(n-1)}\ox\Lambda)
(\proj{\psi_p})\nonumber\\
&=&\sum_{j,k=1}^d\sqrt{p_jp_k}
\big(\ket{j}\bra{k}\big)^{\ox(n-1)}
\ox\Lambda(\ket{j}\bra{k}).
\end{eqnarray}
Then $\rho_{\Lambda,p}^{(n)}$ and $\sigma_{\Lambda,p}^{(n)}$ have the same $(n-1)$-marginals. Consequently, they are also 2-compatible.
\end{lemma}
The proof of Lemma~\ref{le:rho_sigma_compatible} is shown in
Appendix~\ref{app:proofs-multipartite-non-2UDA}.

Lemma~\ref{le:rho_sigma_compatible} converts any nonzero off-diagonal
action of $\Lambda$ into a state distinct from
$\rho_{\Lambda,p}^{(n)}$ but with the same $(n-1)$-marginals.
It therefore reduces the UDA problem for this construction to the
freedom in choosing the off-diagonal action of a channel.

Next, we determine when the prescribed images of the diagonal matrix
units fix a channel uniquely. The following lemma shows a complete
dichotomy in terms of the output states. 
\begin{lemma}
\label{lem:unique_channel_from_diagonal}
Let $d,s\geq 2$, and let $\alpha_1,\ldots,\alpha_d\in\bbM_s$ be states. We define the admissible channel set
\begin{eqnarray}
\label{eq:channel_set_diagonal}
\cA(\alpha_1,\ldots,\alpha_d)
&:=&\{\Lambda:\bbM_d\ra\bbM_s:\ \Lambda\text{ is a channel}, 
\Lambda(\ket{j}\bra{j})=\alpha_j,\ j=1,\ldots,d\}.
\end{eqnarray}
Then  
(i) The set $\cA(\alpha_1,\ldots,\alpha_d)$ is nonempty.

(ii) The set $\cA(\alpha_1,\ldots,\alpha_d)$ contains a unique channel if and only if every $\alpha_j$ is pure and $\tr(\alpha_j\alpha_k)>0$ for any $j\neq k$. In this case, the unique channel is
\begin{eqnarray}
\label{eq:Lambda(X)=}
\Lambda_0(X)=\sum_{j=1}^d X_{j,j}\alpha_j,\qquad
X=(X_{j,k})\in\bbM_d.
\end{eqnarray}

(iii) If the condition in (ii) fails, then
$\cA(\alpha_1,\ldots,\alpha_d)$ contains infinitely
many channels. Moreover, there exist
$\Lambda\in\cA(\alpha_1,\ldots,\alpha_d)$ and
$r\neq t$ such that $\Lambda(\ket{r}\bra{t})\neq 0$.
\end{lemma}
The proof of Lemma~\ref{lem:unique_channel_from_diagonal} is presented
in Appendix~\ref{app:proofs-multipartite-non-2UDA}.

Lemma~\ref{lem:unique_channel_from_diagonal} shows that uniqueness occurs
exactly when the diagonal outputs are pure and pairwise nonorthogonal.
If this condition fails, we obtain not merely a second compatible
channel, but infinitely many choices, including one that preserves a
nonzero off-diagonal matrix unit.

We now combine Lemmas~\ref{le:rho_sigma_compatible}
and~\ref{lem:unique_channel_from_diagonal} to characterize the
$k$-UDA property of the corresponding multipartite states.
\begin{proposition}
\label{pro:UDA_diagonal(n-1)}
Let $n\geq 3$, $d,s\geq 2$, and let $p_j>0$ satisfy
$\sum_{j=1}^d p_j=1$. Suppose
$\alpha_1,\ldots,\alpha_d\in\bbM_s$ are states. We set
$\alpha:=(\alpha_1,\ldots,\alpha_d)$ and
$p:=(p_1,\ldots,p_d)$. We define the state
$\rho_{\alpha,p}^{(n)}$ on $(\bbC^d)^{\ox(n-1)}\ox\bbC^s$ by
\begin{eqnarray}
\label{eq:rho_alpha_p_n}
\rho_{\alpha,p}^{(n)}
=\sum_{j=1}^d p_j
\big(\ket{j\cdots j}\bra{j\cdots j}\big)_{A_1\cdots A_{n-1}}
\ox\alpha_j.
\end{eqnarray}
Then

(i) The state $\rho_{\alpha,p}^{(n)}$ is $(n-1)$-UDA if and only if every $\alpha_j$ is pure and $\tr(\alpha_j\alpha_k)>0$ for any $j\neq k$.

(ii) For every $1\leq k\leq n-1$, the state $\rho_{\alpha,p}^{(n)}$ is $k$-UDA if and only if $k=n-1$ and the condition in (i) holds.
\end{proposition}
The proof of Proposition~\ref{pro:UDA_diagonal(n-1)} is given in
Appendix~\ref{app:proofs-multipartite-non-2UDA}.

Proposition~\ref{pro:UDA_diagonal(n-1)} shows a complete criterion for
this family. In particular, we obtain $(n-1)$-UDA states precisely from
pure, pairwise nonorthogonal output states, whereas no state in the
family is determined by marginals of order below $n-1$. Thus the family
exhibits a sharp separation between determination by
$(n-1)$-marginals and determination by lower-order marginals.

The two multipartite extensions play complementary roles. The
dimension arguments show when UDA is impossible because the range
contains marginal-preserving directions. The channel construction
identifies an explicit family in which the exact UDA order is governed
by a transparent output-state condition. Together with the
three-qubit criteria, these results provide the structural input for
the applications below.

\section{Applications to GME Certification}
\label{sec:applications-GME}

The UDA property of quantum states naturally meets the identifiability
requirement of tomography by marginals, for which estimating
lower-dimensional marginals can require fewer samples than
reconstructing an unrestricted global state
\cite{xin2017quantum,qiu2026mixed}. We now apply the preceding UDA
results to GME certification. We first determine exact GME detection
lengths and then construct an $n$-qubit family whose separable
2-marginals jointly certify GME. For $n\geq4$, this certification
occurs even though the same marginals do not uniquely determine the
global state.

\subsection{GME Detection Length from UDA Criteria}
\label{sec:application-GME-length}

The entanglement detection length (EDL), introduced in
Ref.~\cite{shi2025entangle} quantifies the minimum number of parties
that must be measured jointly in order to certify the genuine multipartite entanglement (GME) of a state.
For a collection $\cJ$ of subsets of $[n]$, we denote the set of
compatible global extensions by
$$
\cE_{\cJ}(\rho)
:=\{\sigma:\sigma\text{ is a state on }\cH,
\ \sigma_{A_{\cS}}=\rho_{A_{\cS}}
\text{ for every }\cS\in\cJ\}.
$$
We say that $\cJ$ detects the GME of $\rho$ if every state in
$\cE_{\cJ}(\rho)$ is GME. The GME detection length of a GME
state $\rho$ is therefore
$$
\ell_{\rm GME}(\rho)
:=\min_{\cJ:\,\cJ\text{ detects the GME of }\rho}
\ \max_{\cS\in\cJ}|\cS|.
$$

We first show a direct connection with UDA. The  1-marginals cannot
certify GME, since they are always compatible with the fully product
state $\bigotimes_{i=1}^n\rho_{A_i}$. On the other hand, if a GME state is
2-UDA, its full collection of 2-marginals determines the state and
hence certifies its
GME. Consequently, every GME state that is 2-UDA satisfies
$\ell_{\rm GME}(\rho)=2$.
Ref.~\cite{shi2025entangle} shows a complete EDL criterion for
symmetric states and an SDP upper bound for general states. We use the
results in this paper to obtain exact analytic values for several
generally nonsymmetric families. Part (i) below shows three families of
three-qubit mixed states, while part (ii) concerns the multipartite
channel construction in \eqref{eq:rho_sigma_channel_states}. The
symmetric special cases are covered by Ref.~\cite{shi2025entangle},
whereas the generally nonsymmetric families and the channel
formulation below do not appear there.
\begin{proposition}
\label{pro:GME_detection_length_families}

(i) Let $\rho$ be a three-qubit mixed state. Suppose  one of the
following conditions holds. Then $\rho$ is GME and $\ell_{\rm GME}(\rho)=2$.

(a) For some $b,c,d\in\bbC$ with $bcd\neq0$,
$
\cR(\rho)
=\Span\{\ket{000},
b\ket{100}+c\ket{010}+d\ket{001}\}.
$

(b) For some $0<\theta\leq\pi/4$,
$\cR(\rho)\subseteq\cC_\theta$ and $\rho^{T_A}\not\geq0$.

(c) For $u$ in \eqref{eq:u=ket001} and $v$ in
\eqref{eq:v=ket011}, $\cR(\rho)=\Span\{u,v\}$.

\noindent
(ii) Let $n\geq3$, $d,s\geq2$, and let
$p=(p_1,\ldots,p_d)$ satisfy $p_j>0$ and
$\sum_{j=1}^dp_j=1$. Let $\Lambda:\bbM_d\ra\bbM_s$ be a channel,
and let $\sigma_{\Lambda,p}^{(n)}$ be defined in
\eqref{eq:rho_sigma_channel_states}. We set
$\cQ=\Span\{\ket{j}^{\ox(n-1)}:j=1,\ldots,d\}$. If
$\sigma_{\Lambda,p}^{(n)}$, regarded as a bipartite state on
$\cQ\ox\bbC^s$, is entangled across $\cQ:A_n$, then it is GME and
$
\ell_{\rm GME}\big(\sigma_{\Lambda,p}^{(n)}\big)=n.
$
When $d=s=2$, the entanglement condition is equivalent to
$\big(\sigma_{\Lambda,p}^{(n)}\big)^{T_{A_n}}\not\geq0$.
\end{proposition}
The proof of Proposition~\ref{pro:GME_detection_length_families} is
presented in Appendix~\ref{app:proofs-GME-applications}.

 Proposition~\ref{pro:GME_detection_length_families} (i) shows 
three generally nonsymmetric families of three-qubit mixed states whose
GME detection length is exactly two. More generally,
Theorem~\ref{th:GHZ_free} shows that every GME
three-qubit mixed state of rank three or four with a GHZ-SLOCC-free
range has GME detection length two, while for rank-two states,
Proposition~\ref{pro:rank_two_2UDA_parameter} shows the corresponding
range criterion, i.e.  every GME state satisfying its two conditions also 
has $\ell_{\rm GME}=2$. These conclusions do not require permutation
symmetry and determine the EDL exactly rather than merely establishing
an upper bound. 
Proposition~\ref{pro:GME_detection_length_families} (ii) turns to the
$n$-partite construction in
Sec.~\ref{subsec:multipartite-non-2UDA} and exhibits the opposite
extreme, i.e. the coherent channel family has maximal GME detection length
$n$, since no collection of proper marginals can certify its GME.

\subsection{GME Certification from Separable Marginals}
\label{subsec:GME-separable-marginals}

We next show that separable 2-marginals can certify GME throughout an
$n$-qubit rank-two family. Although no individual 2-marginal contains
bipartite entanglement, their joint compatibility excludes every
biseparable global completion. Thus the GME is encoded in the
compatibility constraints among the marginals rather than in any
marginal separately. This phenomenon extends the separation between
local and global correlation structures considered in
Refs.~\cite{chen2014role,miklin2016emergent}. In contrast to the
three-qubit case below, GME certification for $n\geq4$ does not arise
from the 2-UDA property.

\begin{proposition}
\label{pro:GME-from-separable-marginals}
Let $n\geq3$ and let
$\boldsymbol{x}=(x_1,\ldots,x_n)$ satisfy $x_i>0$,
$\sum_{i=1}^nx_i=1$, and $x_r\neq x_s$ for some distinct
$r,s\in[n]$. For
$i\in[n]$, we set
$\ket{e_i}:=\ket{0}^{\ox(i-1)}\ox\ket{1}\ox\ket{0}^{\ox(n-i)}$
and define 
\begin{eqnarray}
\label{eq:w-x-emergent-GME}
&&\ket{w_{\boldsymbol{x}}}
:=\sum_{i=1}^n\sqrt{x_i}\ket{e_i},
\\
\label{eq:rho-px-emergent-GME}
&&\rho_{p;\boldsymbol{x}}
:=p\proj{w_{\boldsymbol{x}}}+(1-p)\proj{1}^{\ox n},
\qquad 0<p\leq p_*(\boldsymbol{x}),
\end{eqnarray}
where
$
p_*(\boldsymbol{x})
:=\min_{1\leq i<j\leq n}
\frac{1-x_i-x_j}{1-x_i-x_j+x_ix_j}.
$
Then all 2-marginals of $\rho_{p;\boldsymbol{x}}$ are separable and
jointly certify its GME. Equivalently, every $n$-qubit state
2-compatible with $\rho_{p;\boldsymbol{x}}$ is GME. Consequently,
$\ell_{\rm GME}(\rho_{p;\boldsymbol{x}})=2$. If $n\geq4$, however,
$\rho_{p;\boldsymbol{x}}$ is not 2-UDA.
\end{proposition}

The proof of Proposition~\ref{pro:GME-from-separable-marginals} is
presented in Appendix~\ref{app:proofs-GME-applications}.
When $n=3$, the assumption that the weights are not all equal allows
us to verify both conditions in
Proposition~\ref{pro:rank_two_2UDA_parameter}. The one-qubit
determinant test excludes GHZ-LU vectors from the range, while the
marginal equations leave only a $\ket{000}\bra{111}$ coherence and
its adjoint as possible marginal-preserving directions, which
positivity excludes.
Hence 
$\rho_{p;\boldsymbol{x}}$ is 2-UDA. Thus, in this case, the separable
2-marginals not only certify GME but determine the global state
uniquely. Since the 1-marginals admit the fully product completion
$\rho_{A_1}\ox\rho_{A_2}\ox\rho_{A_3}$, we consequently have
$\ell_{\rm GME}(\rho_{p;\boldsymbol{x}})=2$. For example, we may choose
$\boldsymbol{x}=(1/14,2/7,9/14)$, for which
$p_*(\boldsymbol{x})=7/25$.

\section{Conclusion and Further Perspectives}
\label{sec:conclusion}

In this work, we developed a range-based framework for analyzing UDA states. For three-qubit states, we obtained canonical inclusions for
three- and four-dimensional GHZ-SLOCC-free ranges and showed that states with such ranges are 2-UDA. We derived a necessary and sufficient range criterion for every rank-two state
and reduced it to a finite quadratic-form test. Without assuming a canonical range form, we also obtained an
exact range-restricted spectrahedral criterion that decides each
prescribed state in the remaining rank-three and rank-four regime. Besides, every three-qubit state of rank at least five is shown to be not 2-UDA. 
We next extended the underlying methods to multipartite systems. For multiqubit systems, the spin flip sharpened this
comparison and established the threshold $\rank\rho\geq2^n-n$ for
$n\geq6$. It also reduced the four- and five-qubit threshold cases to
two explicit kernel configurations. For a channel-based multipartite
family, we characterized exactly when $(n-1)$-marginals determine the
state and showed that marginals of every lower order fail. 
Finally, we applied the UDA criteria to GME certification. We determined exact GME
detection lengths for several generally nonsymmetric families and for
a coherent channel construction. We also constructed an $n$-qubit
rank-two family whose separable 2-marginals jointly certify GME. For
$n\geq4$, these marginals certify GME without uniquely determining the
global state. This separation shows that joint marginal compatibility
can certify global entanglement beyond the correlations visible in any
individual marginal.

Many problems arising from this work can be further
explored.
First,  an analytic classification of rank-three and
rank-four ranges that contain GHZ-SLOCC vectors but no GHZ-LU vector remains to be analyzed.
Second, the multipartite extensions also suggest further structural questions.
We may seek higher-dimensional canonical forms and sharper rank
obstructions. The two kernel configurations in
\eqref{eq:remaining_rank_threshold_cases} are the only threshold cases
for $n=4,5$ that are not covered by the present results. Their
resolution requires information beyond dimension counting and provides
concrete test problems for the structure of multiqubit subspaces.
Finally, the GME applications motivate a complementary robustness question.
The present conclusions concern exact marginal data, and the criteria
do not by themselves control numerical conditioning or
finite-sample performance. It would be useful to derive robust versions
of the range criteria, quantify GME certification under noisy
marginals, and determine when separable marginals retain their
certifying power under perturbations \cite{shi2025entangle}.

\section*{Acknowledgment}
This work was supported by  the National Natural Science Foundation of
China under Grant 12471427, the National University of Singapore Research under Grant A-8003570-00-00 and China Scholarship
Council Program under Grant 202506020134.
The authors used ChatGPT 5.6 Sol (OpenAI) to improve the proofs of
Propositions~\ref{pro:GHZ_free_subspace_T_C},
\ref{pro:rank_two_2UDA_parameter}, and~\ref{pro:support-reduced-SDP},
and Theorem~\ref{th:five_dim_contains_GHZ_LU}, and to assist with
language refinement. All mathematical statements and proofs were independently verified by the authors, who take responsibility for the final manuscript.

\appendix

\section{Proofs for Sec.~\ref{subsec:ghz-slocc-free-range}}
\label{app:proofs-ghz-slocc-free-range}

\begin{resultproof}{Lemma~\ref{le:noSLOCC_dim<4=1}}
(i) 
For any
$|\psi\rangle=\sum_{i,j,k\in\{0,1\}}a_{ijk}|ijk\rangle\in\cV$, we have
$\ket{\phi}=\ket{000}+t\ket{\psi}\in\cV$ for every $t\in\bbC$. Since
$\cV$ contains no GHZ-SLOCC states, Lemma~\ref{lem:relative-invariance} shows
$\operatorname{Det}(\ket{\phi})=0$ for every $t\in\bbC$. The coefficient
of $t^2$ is $a_{111}^2$, so it must vanish. Hence $\cV$ is orthogonal to
$\ket{111}$. The coefficient of $t^3$ is
$4a_{011}a_{101}a_{110}$, so it must also vanish. Let
$\ell_{011},\ell_{101},\ell_{110}\in\cV^*$ be the corresponding
coordinate linear forms. Their product is the zero polynomial on
$\cV$. Since the polynomial ring on $\cV$ is an integral domain, one
of these linear forms vanishes identically. Up to a system permutation, we may
assume that $\ell_{011}=0$. Hence $\cV$ is orthogonal to $\ket{011}$,
which proves (i).

(ii) Let $\ket{x_1,x_2,x_3}$ and $\ket{y_1,y_2,y_3}$ be two linearly
independent product vectors in $\cV$. Let $r$ be the number of subsystems
on which $\ket{x_i}$ and $\ket{y_i}$ are linearly independent. If $r=3$,
then, up to SLOCC equivalence, the two product vectors are $\ket{000}$ and
$\ket{111}$, contradicting (i). Hence $r=1$ or $r=2$. Up to SLOCC
equivalence and a system permutation, these two possibilities reduce to
\begin{eqnarray}
\label{eq:lem16_ii_two_cases}
\ket{000},\ket{001}\in \cV,\qquad \text{or}\qquad \ket{000},\ket{011}\in \cV .
\end{eqnarray}

We first consider the case $\ket{000},\ket{001}\in \cV$. By applying (i) to $\ket{000}$ and $\ket{001}$, respectively, we obtain that $\cV$ is orthogonal to $\ket{111}$ and $\ket{110}$. Hence
\begin{eqnarray}
\label{eq:lem16_ii_case_a_support}
\cV\subset \Span\{\ket{000},\ket{001}\}\oplus \cW,
\end{eqnarray}
where $\cW=\Span\{\ket{010},\ket{011},\ket{100},\ket{101}\}$. For any $u=\sum_{i,j,k}u_{ijk}\ket{ijk}\in \cV$, Lemma~\ref{lem:relative-invariance} and the assumption on $\cV$ imply that $\operatorname{Det}(u)=0$. Since $u_{110}=u_{111}=0$, this identity becomes
\begin{eqnarray}
\label{eq:lem16_ii_case_a_det}
0=\operatorname{Det}(u)=(u_{010}u_{101}-u_{011}u_{100})^2 .
\end{eqnarray}
Let $\cU$ be the image of $\cV$ under the projection onto $\cW$. By
\eqref{eq:lem16_ii_case_a_det}, every vector in $\cU$ has zero
determinant as a vector in the bipartite space
$\Span\{\ket{01},\ket{10}\}\ox\bbC^2$. Hence $\cU$ is a subspace of
$\bbC^2\ox\bbC^2$ containing only product vectors. By
Lemma~\ref{lem:biseparable-product-loci}(ii), such a subspace has a
fixed factor and dimension at most two. Since
$\Span\{\ket{000},\ket{001}\}\subset\cV$, we have
\begin{eqnarray}
\label{eq:lem16_ii_case_a_dim}
\dim \cV=2+\dim \cU\le 4 .
\end{eqnarray}

We next consider the case $\ket{000},\ket{011}\in \cV$. By (i), $\cV$
is orthogonal to $\ket{111}$ and at least one of $\ket{101}$ and
$\ket{110}$. Up to interchanging the second and third subsystems, assume
that $\cV$ is orthogonal to $\ket{101}$. Applying (i) to the product
vector $\ket{011}$ shows that $\cV$ is orthogonal to $\ket{100}$. Hence
\begin{eqnarray}
\label{eq:lem16_ii_case_b_support}
\cV\subset \Span\{\ket{000},\ket{001},\ket{010},\ket{011},\ket{110}\}.
\end{eqnarray}
For any $u\in \cV$, Lemma~\ref{lem:relative-invariance} again implies that $\operatorname{Det}(u)=0$, which reduces to
\begin{eqnarray}
\label{eq:lem16_ii_case_b_det}
0=\operatorname{Det}(u)=(u_{001}u_{110})^2 .
\end{eqnarray}
Since the polynomial ring on $\cV$ is an integral domain, one of the
two coordinate functions in \eqref{eq:lem16_ii_case_b_det} vanishes
identically on $\cV$. Thus the projection of $\cV$ onto
$\Span\{\ket{001},\ket{110}\}$ has dimension at most one. By
\eqref{eq:lem16_ii_case_b_support}, the remaining part is contained in
$\Span\{\ket{000},\ket{010},\ket{011}\}$, which is three-dimensional. Therefore $\dim \cV\le 3+1=4$.
Combining the two cases in \eqref{eq:lem16_ii_two_cases}, we have $\dim \cV\le 4$.

(iii) Up to SLOCC equivalence and system permutation, it suffices to consider the two cases in \eqref{eq:lem16_ii_two_cases}.
We first consider the case $\ket{000},\ket{001}\in \cV$. Let $\cW$ and $\cU$ be defined as in the proof of (ii). Since $\Span\{\ket{000},\ket{001}\}\subset \cV$ and $\dim \cV=4$, the equality in \eqref{eq:lem16_ii_case_a_dim} yields $\dim \cU=2$. Moreover, since $\Span\{\ket{000},\ket{001}\}\subset \cV$, every vector in $\cU$ also belongs to $\cV$ after subtracting its $\Span\{\ket{000},\ket{001}\}$ part. Hence
\begin{eqnarray}
\label{eq:lem16_iii_case_a_decomp}
\cV=\Span\{\ket{000},\ket{001}\}\oplus \cU .
\end{eqnarray}
The space $\cU$ is a two-dimensional subspace of
$\bbC^2\ox\bbC^2$ containing only product vectors. By
Lemma~\ref{lem:biseparable-product-loci}(ii), all vectors in $\cU$
have one fixed local factor. Therefore one of the following holds,
\begin{eqnarray}
\label{eq:lem16_iii_U_forms}
&&\cU=\Span\{\a\ket{01}+\b\ket{10}\}\ox \bbC^2, \quad (\a,\b)\neq(0,0)
\\
\label{eq:lem16_iii_U_forms_2}
&&\cU=\Span\{\ket{01},\ket{10}\}\ox \Span\{\ket{\phi}\}, \quad \ket{\phi}\neq 0.
\end{eqnarray}

In the case \eqref{eq:lem16_iii_U_forms}, we have
\begin{eqnarray}
\label{eq:lem16_iii_first_subcase}
\cV=\Span\{\ket{00},\a\ket{01}+\b\ket{10}\}\ox \bbC^2 .
\end{eqnarray}
If $\a\b=0$, then \eqref{eq:lem16_iii_first_subcase} is, up to system permutation, 
$\Span\{\ket{000},\ket{001},\ket{100},\ket{101}\}$. If $\a\b\neq 0$, then local diagonal invertible operators on the first two systems transform $\a\ket{01}+\b\ket{10}$ into $\ket{01}+\ket{10}$. Hence, after moving the free qubit to the first system, $\cV$ is SLOCC equivalent to
\begin{eqnarray}
\label{eq:lem16_iii_C_type}
\bbC^2\ox \Span\{\ket{00},\ket{01}+\ket{10}\}.
\end{eqnarray}
In the case \eqref{eq:lem16_iii_U_forms_2}, a local invertible operator on the third system maps $\ket{\phi}$ to $\ket{0}$. Hence \eqref{eq:lem16_iii_case_a_decomp} becomes
\begin{eqnarray}
\label{eq:lem16_iii_T_type}
\Span\{\ket{000},\ket{001},\ket{010},\ket{100}\}.
\end{eqnarray}

We next consider the case $\ket{000},\ket{011}\in \cV$. By the proof of (ii), up to interchanging the second and third systems, $\cV$ is contained in
$\Span\{\ket{000},\ket{001},\ket{010},\ket{011},\ket{110}\}$, and either $\cV$ is orthogonal to $\ket{001}$ or $\ket{110}$. Since $\dim \cV=4$, $\cV$ is equal to one of the following two four-dimensional coordinate subspaces,
\begin{eqnarray}
\label{eq:lem16_iii_case_b_two_forms}
&&\Span\{\ket{000},\ket{010},\ket{011},\ket{110}\},
\\
&&\Span\{\ket{000},\ket{001},\ket{010},\ket{011}\}.
\end{eqnarray}
A bit flip on the second system maps the first subspace in \eqref{eq:lem16_iii_case_b_two_forms} to
$\Span\{\ket{000},\ket{001},\ket{010},\ket{100}\}$, and interchanging the first and second systems maps the second subspace to
\begin{eqnarray}
\label{eq:span_000,001,100,101}
\Span\{\ket{000},\ket{001},\ket{100},\ket{101}\}.    
\end{eqnarray}
Combining the above cases, up to SLOCC equivalence and system permutation, $\cV$ is one of \eqref{eq:lem16_iii_C_type}, \eqref{eq:lem16_iii_T_type} and \eqref{eq:span_000,001,100,101}.
This proves (iii).
\end{resultproof}

\begin{resultproof}{Lemma~\ref{le:noSLOCC_dim<4}(i)}
We prove the claim by contradiction. Suppose that $\dim \cV\ge 5$.
Then $\cV$ contains a five-dimensional subspace. By
Lemma~\ref{lem:product-vectors-from-dimension}, this subspace, and
hence $\cV$, contains a fully product vector. Up to SLOCC equivalence,
we may assume that this product vector is $\ket{000}$.

Let $W=\Span\{\ket{001},\ket{010},\ket{100},\ket{101},\ket{110}\}$. By Lemma \ref{le:noSLOCC_dim<4=1} (i), $\cV$ is orthogonal to both $\ket{011}$ and
$\ket{111}$. Hence $\cV$ is contained in the six-dimensional subspace
$\Span\{\ket{000}\}\oplus W$. The part of $\cV$ lying in
$W$ has dimension at least $\dim\cV-1\ge 4$. Equivalently,
$U:=\cV\cap W$ is a subspace of $W$ with $\dim U\ge 4$.

Now consider the two-dimensional subspace
$L=\Span\{\ket{100},\ket{101}\}$ of $W$. Since both $U$ and $L$ are
subspaces of the five-dimensional space $W$, the dimension formula yields
$\dim(U\cap L)\ge \dim U+\dim L-\dim W\ge 4+2-5=1$. Thus $U\cap L$
contains a nonzero vector. Therefore there exist a product vector $v=a\ket{100}+b\ket{101}\in U\cap L\subseteq \cV$, for $a,b\in \bbC$.
Moreover, $v$ is
linearly independent of $\ket{000}$. Hence $\cV$ contains two linearly independent
product vectors, namely $\ket{000}$ and $v$. This contradicts Lemma
\ref{le:noSLOCC_dim<4=1} (ii). Therefore the assumption $\dim\cV\ge 5$
is impossible, and the claim follows.
\end{resultproof}

To prove Lemma~\ref{le:noSLOCC_dim<4_product}(ii), we first establish
the following auxiliary fact.
\begin{lemma}
\label{le:bilinear_form_dim<=2}
Let $\cV=\bbC^4$,  $(x,y,z,w)\in \cV$ and nonzero $\alpha,\beta\in\bbC$. 
If $\cW\subset \cV$ is a linear subspace such that the quadratic form $q(x,y,z,w)=(\beta x-\alpha w)^2+4\alpha\beta yz$ vanishes identically on $\cW$, i.e., 
$q|_{\cW}\equiv 0$, then $\dim \cW\leq 2$. 
\end{lemma}

\begin{resultproof}{Lemma~\ref{le:bilinear_form_dim<=2}}
Let $r=\beta x-\alpha w$ and $p=\beta x+\alpha w$. Since $\alpha\beta\neq 0$, the variables $(p,r,y,z)$ form a linear coordinate system on $\cV$. In these coordinates, the quadratic form is
\begin{eqnarray}
\label{eq:q_pryz}
q&=&r^2+4\alpha\beta yz .
\end{eqnarray}
Let $B_q$ be the symmetric bilinear form associated with $q$, namely $B_q(u,v)=\frac{1}{2}[q(u+v)-q(u)-q(v)]$. For $u=(p,r,y,z)$ and $v=(p',r',y',z')$, by \eqref{eq:q_pryz} one has
\begin{eqnarray}
\label{eq:Bq_pryz}
B_q(u,v)&=&rr'+2\alpha\beta(yz'+y'z).
\end{eqnarray}
It follows from \eqref{eq:Bq_pryz} that the radical of $q$ is
\begin{eqnarray}
\label{eq:rad_q}
\cD_q:=\operatorname{rad}(q)&=&\{(p,r,y,z):r=y=z=0\}.
\end{eqnarray}
Thus $\cD_q$ is one-dimensional. We note that the variable $p$ does not occur in \eqref{eq:q_pryz}. Therefore, the $p$-axis is precisely the radical direction.

Let $\overline{\cV}=\cV/\cD_q$ be the quotient space. Since $\cD_q$ is the radical, $q$ descends to a quadratic form $\overline q$ on $\overline{\cV}$, defined by $\overline q(v+\cD_q)=q(v)$. This is well-defined because, for any $h\in \cD_q$, one has $q(v+h)=q(v)+2B_q(v,h)+q(h)=q(v)$. In the quotient coordinates $(r,y,z)$, the descended quadratic form is
\begin{eqnarray}
\label{eq:qbar_ryz}
\overline q(r,y,z)&=&r^2+4\alpha\beta yz .
\end{eqnarray}
Its associated bilinear form has matrix
$
\bma
1 & 0 & 0\\
0 & 0 & 2\alpha\beta\\
0 & 2\alpha\beta & 0
\ema.
$
The determinant of this matrix is $-4\alpha^2\beta^2\neq 0$. Hence $\overline q$ is nondegenerate on the three-dimensional vector space $\overline{\cV}$.

Suppose $\cW\subset \cV$ is a linear subspace with $q|_{\cW}\equiv 0$. We denote by
\begin{eqnarray}
\label{eq:Wbar_def}
\overline{\cW}&=&(\cW+\cD_q)/\cD_q\subset \overline{\cV}
\end{eqnarray}
the image of $\cW$ in the quotient space $\overline{\cV}$. Since $q$ vanishes identically on $\cW$, the descended quadratic form $\overline q$ vanishes identically on $\overline{\cW}$. Hence $\overline q|_{\overline{\cW}}\equiv 0$. 
Let $\overline B$ be the associated bilinear form of $\overline q$. For any $\overline u,\overline v\in\overline{\cW}$, since $\overline{\cW}$ is a linear subspace and $\overline q$ vanishes identically on it, one has $\overline q(\overline u)=\overline q(\overline v)=\overline q(\overline u+\overline v)=0$. Therefore $\overline B(\overline u,\overline v)=0$. Thus $\overline{\cW}\subset \overline{\cW}^\perp$ with respect to the nondegenerate bilinear form $\overline B$. 
Since $\overline B$ is nondegenerate on $\overline{\cV}$ and
$\dim\overline{\cV}=3$, we have
$\dim\overline{\cW}+\dim\overline{\cW}^\perp=3$. The inclusion
$\overline{\cW}\subset\overline{\cW}^\perp$ implies
$\dim\overline{\cW}\leq\dim\overline{\cW}^\perp$, and hence
$2\dim\overline{\cW}\leq3$. Therefore $\dim\overline{\cW}\leq1$.

Finally, by \eqref{eq:Wbar_def}, the natural map $\cW\to \overline{\cW}$ has kernel $\cW\cap \cD_q$. Hence, $\dim \cW=\dim(\cW\cap \cD_q)+\dim\overline{\cW}$. Since $\dim \cD_q=1$, we have $\dim(\cW\cap \cD_q)\leq 1$. From $\dim\overline{\cW}\leq 1$, we have $\dim \cW\leq 2$.
\end{resultproof}

\begin{resultproof}{Lemma~\ref{le:noSLOCC_dim<4_product}(ii)}
The proof is carried out separately for the two cases (i) $\dim \cV = 3$ and (ii) $\dim \cV = 4$.

(i) Suppose $\dim \cV=3$. We prove by contradiction. Suppose that $\cV$ contains no nonzero fully product vector.
First, we show that $\cV$ contains a state in the W SLOCC class.
We write $\bbP(\cV)=(\cV\setminus\{0\})/\bbC^\times$.
If $\bbP(\cV)$ contains no state in the W SLOCC class, then every nonzero vector in $\cV$ is
biseparable with respect to at least one of the three bipartitions by
the three-qubit SLOCC classification
\cite{dur2000three,miyake2003classification}. We have
$$
\bbP(\cV)\subseteq \Sigma_{A:BC}\cup \Sigma_{B:AC}\cup \Sigma_{C:AB},
$$
where, for example, $\Sigma_{A:BC}$ denotes the projective variety of states separable across the bipartition $A:BC$.
By Lemma~\ref{lem:biseparable-product-loci}(i), $\bbP(\cV)$ is
contained in one of these three varieties. Up to system permutation, we assume that
$\bbP(\cV)\subseteq \Sigma_{A:BC}$. 

By Lemma~\ref{lem:biseparable-product-loci}(ii), we have
$\cV=\ket{a_0}_A\ox \cW$ or
$\cV=\cU_A\ox \ket{\eta}_{BC}$.
The latter one is impossible since $\dim\cV=3>\dim\cH_A=2$. Hence, 
\begin{eqnarray}
\label{eq:V_fixed_A_factor}
\cV=\ket{a_0}_A\ox \cW,
\end{eqnarray}
where $\cW\subseteq \cH_{BC}\simeq \bbC^2\ox \bbC^2$ is a three-dimensional subspace. 
By Lemma~\ref{lem:product-vectors-from-dimension}, the space $\cW$
contains a two-qubit product vector. Eq.
\eqref{eq:V_fixed_A_factor} then shows
a fully product vector in $\cV$, a contradiction. Therefore
$\cV$ contains a state in the W SLOCC class.

Both the condition $\operatorname{Det}|_{\cV}=0$ and the property of containing no nonzero fully product vector are invariant under invertible local transformations. Thus, after applying a SLOCC transformation to the whole subspace, we may assume that
\begin{eqnarray}
\label{eq:standard_W_in_V}
\ket{W}:=\ket{001}+\ket{010}+\ket{100}\in \cV.
\end{eqnarray}

Let $x=\sum_{i,j,k=0}^1 x_{ijk}\ket{ijk}\in\cV$. Since $\cV$ is
linear, $\ket{W}+tx\in\cV$ for every $t\in\bbC$. Hence
$\operatorname{Det}(\ket{W}+tx)\equiv 0$.
Therefore the coefficient of $t$ shows $x_{111}=0$. Substituting this into the coefficient of $t^2$, we obtain
\begin{eqnarray}
\label{eq:h_and_q_zero}
q(x_{011},x_{101},x_{110}):=x_{011}^2-2x_{011}(x_{101}+x_{110})+x_{101}^2-2x_{101}x_{110}+x_{110}^2=0.
\end{eqnarray}
Let $\cW_1:=\Span\{\ket{011},\ket{101},\ket{110}\}$. We define the projection
\begin{eqnarray}
\label{eq:pi_projection_U}
\pi:\cV\longrightarrow \cW_1,\quad 
\pi(x)=x_{011}\ket{011}+x_{101}\ket{101}+x_{110}\ket{110}.
\end{eqnarray}
Since $\pi$ is linear, $\pi(\cV)$ is a linear subspace of $\cW_1$. Because
\eqref{eq:h_and_q_zero} holds for every $x\in\cV$, the quadratic form $q$
vanishes identically on $\pi(\cV)$.

The quadratic form $q(x_{011},x_{101},x_{110})$ has symmetric matrix
\begin{eqnarray}
\label{eq:q_dfg_matrix}
M_q=
\bma
1&-1&-1\\
-1&1&-1\\
-1&-1&1
\ema,
\quad \det M_q=-4\neq 0.
\end{eqnarray} 
For $u,v\in\pi(\cV)$ and $a,b\in\bbC$, the vector
$au+bv$ also belongs to $\pi(\cV)$. Hence,
$$
0=q(au+bv)=a^2q(u)+b^2q(v)+2ab\,u^TM_qv
=2ab\,u^TM_qv.
$$
Thus the symmetric bilinear form $B(u,v):=u^TM_qv$ vanishes on
$\pi(\cV)\times\pi(\cV)$, so
$\pi(\cV)\subseteq\pi(\cV)^\perp$. Since $B$ is nondegenerate by
$\det M_q\neq0$, we have
$\dim\pi(\cV)+\dim\pi(\cV)^\perp=3$. Therefore,
\begin{eqnarray}
\label{eq:dim_pi_leq_one}
\dim \pi(\cV)\leq 1.
\end{eqnarray}

If $\dim\pi(\cV)=0$, then
$\cV\subseteq\cT$ in \eqref{eq:cT}.
Inside $\cT$, we have the two-dimensional product subspace
$\cP_0:=\Span\{\ket{000},\ket{100}\}=\bbC^2\ox \ket{00}$.
Since $\dim\cV=3$, $\dim \cP_0=2$, and $\dim\cT=4$, we have
\begin{eqnarray}
\label{eq:V_intersect_P0}
\dim(\cV\cap \cP_0)&=&\dim \cV +\dim \cP_0-\dim(\cV+\cP_0)
\\
&\ge& \dim \cV +\dim \cP_0-\dim \cT
\\
&=& 1.
\end{eqnarray}
Thus $\cV$ contains a nonzero fully product vector, which is a contradiction. Hence,
\begin{eqnarray}
\label{eq:dim_pi_equal_one}
\dim\pi(\cV)=1.
\end{eqnarray}
Let
$\cV_0:=\cV\cap\cT=\ker\pi$.
By \eqref{eq:dim_pi_equal_one}, we have $\dim \cV_0=\dim\cV-\dim \pi(\cV)=2$. We choose $w\in \cV\setminus\cT$ and write
$w=A\ket{000}+B\ket{001}+C\ket{010}+D\ket{011}
+E\ket{100}+F\ket{101}+G\ket{110}$,
where $(D,F,G)\neq (0,0,0)$. For any
\begin{eqnarray}
\label{eq:u_in_V0}
u=a\ket{000}+b\ket{001}+c\ket{010}+e\ket{100}\in \cV_0,
\end{eqnarray}
we have $u+tw\in\cV$ for every $t\in\bbC$. The coefficient of $t^2$ in $\operatorname{Det}(u+tw)$ is
\begin{eqnarray}
\label{eq:Q_DFG_zero}
Q_{D,F,G}(b,c,e):=(De-Fc-Gb)^2-4FGbc=0,\quad \forall u\in \cV_0.
\end{eqnarray}

We consider the projection $\varpi:\cV_0\to \bbC^3$ defined by $\varpi(u)=(b,c,e)$. If $\dim\varpi(\cV_0)\leq 1$, then $\ker\varpi$ contains a nonzero vector. We thus obtain a nonzero vector of the form $a\ket{000}$ in $\cV_0$, which is fully product. Hence
\begin{eqnarray}
\label{eq:varpi_V0_dim_two}
\dim\varpi(\cV_0)=2.
\end{eqnarray}
The quadratic form $Q_{D,F,G}$ vanishes on a two-dimensional linear subspace of the $(b,c,e)$-space. With respect to the coordinates $(e,c,b)$, the symmetric matrix of $Q_{D,F,G}$ is
\begin{eqnarray}
\label{eq:M_DFG_matrix}
M(D,F,G)=
\bma
D^2&-DF&-DG\\
-DF&F^2&-FG\\
-DG&-FG&G^2
\ema.
\end{eqnarray}
Then $\det M(D,F,G)=-4D^2F^2G^2$. 
If $DFG\neq 0$, then $Q_{D,F,G}=0$ defines a smooth conic in $\bbP^2$,
which cannot contain a projective line \cite{harris1992algebraic}. This contradicts
\eqref{eq:varpi_V0_dim_two}, because $\bbP(\varpi(\cV_0))\subset\bbP^2$ is a
projective line and $Q_{D,F,G}\equiv0$ on $\varpi(\cV_0)$ implies
$\bbP(\varpi(\cV_0))\subset\{Q_{D,F,G}=0\}$. Hence
\begin{eqnarray}
\label{eq:DFG_zero}
DFG=0.
\end{eqnarray}
On the other hand, applying \eqref{eq:h_and_q_zero} to $w$ yields
\begin{eqnarray}
\label{eq:q_DFG_zero}
D^2-2D(F+G)+F^2-2FG+G^2=0.
\end{eqnarray}
Combining \eqref{eq:DFG_zero} and \eqref{eq:q_DFG_zero}, and using a system permutation if necessary, we may assume that $G=0$ and $D=F\neq 0$. Rescaling $w$, we set
\begin{eqnarray}
\label{eq:DF_normal_form}
G=0,\quad D=F=1.
\end{eqnarray}

Then \eqref{eq:Q_DFG_zero} becomes $Q_{1,1,0}(b,c,e)=(e-c)^2=0$. Therefore every $u\in \cV_0$ in \eqref{eq:u_in_V0} satisfies $e=c$, and hence
\begin{eqnarray}
\label{eq:V0_subset_final}
\cV_0\subset \cX_0:=\Span\{\ket{000},\ket{001},\ket{010}+\ket{100}\}.
\end{eqnarray}
The three-dimensional space $\cX_0$ contains the two-dimensional product subspace
$\cP_1:=\Span\{\ket{000},\ket{001}\}$.
Since $\dim \cV_0=2$, $\dim \cP_1=2$, and $\dim \cX_0=3$, we have
\begin{eqnarray}
\label{eq:V0_intersect_P1_simplified}
\dim(\cV_0\cap \cP_1)\geq 2+2-3=1.
\end{eqnarray}
Thus $\cV_0$ contains a nonzero fully product vector, which is a contradiction. Therefore $\cV$ must contain a nonzero fully product vector. 

(ii) Suppose $\dim\cV=4$. Consider $\cV$ across the bipartition $A:BC$.
Then $\cV\subseteq\cH_A\ox\cH_{BC}$. By
Lemma~\ref{lem:product-vectors-from-dimension}, $\cV$ contains a
product vector
$\ket{\phi}_A\ox\ket{\eta}_{BC}\in\cV$. If $\ket{\eta}_{BC}$ is a
product vector, then $\ket{\phi}_A\ox\ket{\eta}_{BC}$ is fully product.
Otherwise, $\ket{\eta}_{BC}$ is entangled. Up to LU equivalence, we may
write
\begin{eqnarray}
\label{eq:s=ket0_ox_a00+b11}
s=\ket{0}_A\ox (\a\ket{00}+\b\ket{11}), \quad \a\b\neq 0.
\end{eqnarray}
Write $\cV=\Span\{s\}\oplus\cU$, where $\dim\cU=3$. For any
\begin{eqnarray}
 \label{eq:u in cU}
u=\sum_{i,j,k}u_{i,j,k}\ket{ijk}\in \cU,   
\end{eqnarray}
the set $\{s+tu:t\in\bbC\}\subset\cV$ contains no GHZ-SLOCC states.
Thus, by Lemma~\ref{lem:relative-invariance},
$\operatorname{Det}(s+tu)\equiv0$ for every $t\in\bbC$. Consequently,
the coefficient of $t^2$ vanishes, and we obtain
\begin{eqnarray}
 \label{eq:alphaU111-betaU100}
(\a u_{111}-\b u_{100})^2+4\a\b u_{101} u_{110}=0, \quad \forall u\in \cU.
\end{eqnarray} 
For $t=0,1$, define the projection
$$
\pi_t:\cU\ra\cH_B\ox\cH_C,\qquad
\sum_{i,j,k}u_{i,j,k}\ket{i,j,k}
\longmapsto \sum_{j,k}u_{t,j,k}\ket{j,k}.
$$
Condition \eqref{eq:alphaU111-betaU100} depends only on $\pi_1(u)$. By
Lemma \ref{le:bilinear_form_dim<=2},
$\dim\pi_1(\cU)\le2$. From $\dim\cU=3$, we have 
$$
\dim\ker\pi_1=\dim\cU-\dim\pi_1(\cU)\ge1.
$$
Therefore, there is a nonzero 
$u=\sum_{j,k}u_{0,j,k}\ket{0,j,k} \in \cU$  
such that $\pi_1(u)=0$. By \eqref{eq:s=ket0_ox_a00+b11}, the two-dimensional space
$$
\cV_1:=\Span\{s,u\}
=\ket{0}_A\ox
\Span\Big\{\a\ket{00}+\b\ket{11},
\sum_{j,k}u_{0,j,k}\ket{j,k}\Big\}_{BC}
$$
contains a product vector $\ket{x}_B\ox\ket{y}_C$ in its two-qubit
factor by Lemma~\ref{lem:product-vectors-from-dimension}. Hence
$\ket{0}_A\ox\ket{x}_B\ox\ket{y}_C\in
\cV_1\subseteq\cV$ is fully product. This
completes the proof.
\end{resultproof}

\begin{resultproof}{Proposition~\ref{pro:GHZ_free_subspace_T_C}}
A direct evaluation of the Cayley hyperdeterminant shows that every
vector in $\cT$ or $\cC_\theta$ has zero hyperdeterminant. Hence, by
Lemma~\ref{lem:relative-invariance}, these
canonical subspaces contain no GHZ-SLOCC states.

(i) Suppose that $\cV$ contains a fully product vector. Up to LU
equivalence, we may take this vector to be $\ket{000}$. By
Lemma~\ref{le:noSLOCC_dim<4=1}(i), after a system permutation if
necessary, $\cV$ is orthogonal to $\ket{011}$ and $\ket{111}$. Thus
$$
\cV\subseteq\cU_0\oplus\cW_0,
\qquad
\cU_0:=\cH_A\ox\Span\{\ket{00}\},
\qquad
\cW_0:=\cH_A\ox\Span\{\ket{01},\ket{10}\}.
$$
For every $u=\sum_{i,j,k}u_{ijk}\ket{ijk}\in\cV$,
Lemma~\ref{lem:relative-invariance} and the assumption on $\cV$ imply
$\operatorname{Det}(u)=0$. Under the above coordinate restriction, this
identity reduces to
$
0=\operatorname{Det}(u)
=
\bigl(u_{001}u_{110}-u_{010}u_{101}\bigr)^2.
$
Consequently, under the natural identification
$\cW_0\simeq\bbC^2\ox\bbC^2$, the projection of $\cV$ onto
$\cW_0$ is a linear subspace consisting only of product vectors.
By Lemma~\ref{lem:biseparable-product-loci}(ii), such a subspace has
a fixed factor on one side.

There are therefore two cases. In the first case, for some nonzero
$\ket\phi\in\cH_A$,
$$
\cV\subseteq
\cU_0\oplus
\left(
\Span\{\ket\phi\}\ox\Span\{\ket{01},\ket{10}\}
\right),
$$
which is LU equivalent to a subspace of $\cT$. In the second case, for
some nonzero
$\ket\varphi\in\Span\{\ket{01},\ket{10}\}$,
$$
\cV\subseteq
\cH_A\ox
\Span\{\ket{00},\ket\varphi\}.
$$
Local phase transformations and, if necessary, an interchange of the
second and third systems bring
$\ket\varphi$ to
$\cos\theta\ket{01}+\sin\theta\ket{10}$ with
$0\leq\theta\leq\pi/4$. Hence $\cV$ is LU equivalent to a subspace
of $\cC_\theta$. 

(ii) If $\dim\cV=3$ or $4$, then
Lemma~\ref{le:noSLOCC_dim<4_product}(ii) guarantees that $\cV$
contains a fully product vector. The assertion therefore follows from
(i). When $\dim\cV=4$, both $\cT$ and $\cC_\theta$ are
four-dimensional, so the corresponding inclusion is an equality.
\end{resultproof}

\begin{resultproof}{Proposition~\ref{pro:T_Ct_2UDA}}
Let $\sigma$ be a three-qubit state that is 2-compatible with $\rho$.
We show that $\sigma=\rho$.

We first consider $\cR(\rho)\subseteq\cT$. The vector $\ket{11}$
belongs to the kernel of every 2-marginal of $\rho$. The same holds
for $\sigma$. We set
$
P_A=I_A\ox\ket{11}\bra{11}_{BC}$,
$P_B=\ket{1}\bra{1}_A\ox I_B\ox\ket{1}\bra{1}_C$, and
$P_C=\ket{11}\bra{11}_{AB}\ox I_C.$ 
Positivity and $\tr(P_X\sigma)=0$ for $X=A,B,C$ imply
$\cR(\sigma)\subseteq\ker P_X$, and hence
$$
\cR(\sigma)\subseteq\bigcap_{X=A,B,C}\ker P_X
=\Span\{\ket{000},\ket{001},\ket{010},\ket{100}\}=\cT.
$$
We set $H=\sigma-\rho$. Then $H\in\cN_2$ and $\cR(H)\subseteq\cT$.
With $e_0=\ket{000}$, $e_1=\ket{001}$, $e_2=\ket{010}$, and
$e_3=\ket{100}$, we write
$H=\sum_{i,j=0}^3h_{ij}\ket{e_i}\bra{e_j}$. Comparing the matrix
entries in $\tr_A(H)=\tr_B(H)=\tr_C(H)=0$, we have 
$$
h_{ij}=0\quad\text{for }(i,j)\neq(0,0),\qquad
h_{00}+h_{kk}=0\quad\text{for }k=1,2,3.
$$
Thus $h_{00}=0$, so $H=0$.

We next consider $\cR(\rho)\subseteq\cC_\t$. We write
$\cC_\t=\cH_A\ox\cU_{BC}(\t)$, where
$$
\cU_{BC}(\t)
=\Span\{\ket{00},\cos\t\ket{01}+\sin\t\ket{10}\}.
$$
Since $\sigma_{BC}=\rho_{BC}$, the ranges of both marginals are
contained in $\cU_{BC}(\t)$. From $\sigma\ge 0$, we have 
$\cR(\sigma)\subseteq\cH_A\ox\cU_{BC}(\t)$. We again set
$H=\sigma-\rho$ and write
$H=\sum_{i,j=0}^1\ket{i}\bra{j}_A\ox K_{ij}$, where
$K_{ij}:\cU_{BC}(\t)\ra\cU_{BC}(\t)$ is linear. The condition
$\tr_B(H)=0$ implies
$\tr_B(K_{ij})=0$ for every $i,j$.

We choose $\ket{\ps_0}=\ket{00}$ and
$\ket{\ps_1}=\cos\t\ket{01}+\sin\t\ket{10}$. Every
linear operator $K:\cU_{BC}(\t)\ra\cU_{BC}(\t)$ has the form
$K=a\proj{\ps_0}+b\proj{\ps_1}
+c\ket{\ps_0}\bra{\ps_1}+d\ket{\ps_1}\bra{\ps_0}$, and
$$
\tr_B(K)=
\bma
a+b\sin^2\t&c\cos\t\\
d\cos\t&b\cos^2\t
\ema_C.
$$
Since $\cos\t>0$, the equation $\tr_B(K)=0$ implies
$a=b=c=d=0$. Thus every $K_{ij}=0$, and hence $H=0$. In both cases,
$\sigma=\rho$, which proves that $\rho$ is 2-UDA.
\end{resultproof}

\begin{resultproof}{Theorem~\ref{th:GHZ_free}}
Since the three-qubit subspace $\cR(\r)$ contains no GHZ-SLOCC states,
Lemma~\ref{le:noSLOCC_dim<4}(i) implies
$\dim\cR(\r)\le4$, i.e., $\rank(\r)\le4$.

The claim holds if $\rank(\r)=1$, i.e. $\r=\proj{\ps}$. 
It is known that a three-qubit pure state $\ket{\ps}$ is 2-UDA if and only
if it is not LU equivalent to a GHZ-type state
$a\ket{000}+b\ket{111}$ with $ab\neq0$.
  If $\cR(\r)$ has no GHZ-SLOCC state, then $\ket{\ps}$ is not a GHZ-SLOCC state, and hence it is not LU equivalent to  a GHZ type state.   Therefore, $\r$ is 2-UDA. 

Next, assume that $\rank(\r)=3$ or $4$. By
Proposition~\ref{pro:GHZ_free_subspace_T_C}(ii), up to LU equivalence
and system permutation, $\cR(\r)$ is a subspace of $\cT$ in
\eqref{eq:cT} or $\cC_\t$ in \eqref{eq:cC_t}. By
Proposition~\ref{pro:T_Ct_2UDA}, $\r$ is 2-UDA.
\end{resultproof}

\section{Proofs for Sec.~\ref{subsec:rank-two-three-qubit-states}}
\label{app:proofs-rank-two-three-qubit-states}

\begin{appendixlemma}
\label{lem:generalized-schur-complement}
Let $\cH_1$ and $\cH_2$ be finite-dimensional Hilbert spaces. Suppose
$X$ and $D$ are Hermitian operators on $\cH_1$ and $\cH_2$,
respectively, and $B:\cH_1\ra\cH_2$ is linear. Then
$
\bma
X&B^\dagger\\
B&D
\ema\geq0
$
if and only if $D\geq0$, $\cR(B)\subseteq\cR(D)$, and
$X-B^\dagger D^+B\geq0$, where $D^+$ is the Moore-Penrose inverse of
$D$.
\end{appendixlemma}
This is the standard generalized Schur-complement criterion, and its
proof is omitted.

\begin{resultproof}{Proposition~\ref{pro:rank_two_2UDA_parameter}}
	We first record two consequences of the spin flip in
	\eqref{eq:n-qubit-spin-flip}. Since $\Theta_3^2=-I$, one has
	\begin{equation}
		\label{eq:braketz=0}
		\braket{v}{\widetilde v}=0,
		\quad \forall \text{ three-qubit vector }\ket v.
	\end{equation}

	Since $\Theta_3$ is antiunitary, we have
	$\norm{\widetilde x}=\norm{x}$. We set
	$\ket{e_1}:=\ket{x}/\norm{x}$ and
	$\ket{e_2}:=\ket{\widetilde x}/\norm{x}$. By
	\eqref{eq:braketz=0}, these two vectors are orthonormal. Bessel's
	inequality states that the squared norm of the orthogonal projection of
	$\ket y$ onto their span does not exceed $\norm{y}^2$. Hence, 
	$$
	|\braket{x}{y}|^2+|\braket{\widetilde x}{y}|^2
	\leq \norm{x}^2\norm{y}^2.
	$$
	By \eqref{eq:delta_E_rank_two}, this inequality is exactly
	$\delta_{\mathsf E}\geq0$. Equality holds if and only if
	$\ket y\in\Span\{\ket x,\ket{\widetilde x}\}$. Since $\ket x$ and
	$\ket y$ are linearly independent and
	$\Theta_3\ket{\widetilde x}=-\ket x$, this equality condition is
	equivalent to $\cL=\widetilde{\cL}$. Therefore,
	\begin{equation}
		\label{eq:delta_zero_spin_invariant}
		\delta_{\mathsf E}=0
		\quad\Longleftrightarrow\quad
		\cL=\widetilde{\cL}.
	\end{equation}
	In particular, the alternatives $\delta_{\mathsf E}=0$ and
	$\delta_{\mathsf E}>0$ depend only on the range $\cL$.
	
	We next justify that the factorization entails no loss of generality.
	Let $\rho'$ be another rank-two state with
	$\cR(\rho')=\cL$, and write
	$\rho'=\mathsf E'\mathsf E'^\dagger$. Since $\mathsf E$ and
	$\mathsf E'$ have the same range, there is an invertible matrix
	$R\in GL(2,\bbC)$ such that $\mathsf E'=\mathsf E R$. From
	\eqref{eq:delta_E_rank_two}, we have
	$\delta_{\mathsf E'}=|\det R|^2\delta_{\mathsf E}$,
	$F'=R^{-1}F (R^\dagger)^{-1}$, 
	$G'=R^{-1}G R^{-T}$,   
	$\Phi_{\mathsf E'}(F',G')=\Phi_{\mathsf E}(F,G)$, and 
	$\det F'=\frac{\det F}{|\det R|^2}$.
	Since $R$ is invertible and $G'=(G')^T$, this transformation induces a
	bijection between $\cF_{\mathsf E}$ and $\cF_{\mathsf E'}$. It follows
	that the alternatives $\delta_{\mathsf E}=0$ and
	$\delta_{\mathsf E}>0$, as well as the sign condition on $\det F$, depend
	only on $\cL$. In particular, if $\mathsf E'$ is another factorization
	of the same $\rho$, then $R$ is unitary. This also agrees with the
	range invariance of the UDA property \cite{qiu2026mixed}.
	
	We will also use the following consequence of
	\eqref{eq:delta_zero_spin_invariant}. If $\delta_{\mathsf E}>0$, then
	\begin{equation}
		\label{eq:rank_two_spin_direct_sum}
		\cL\cap\widetilde{\cL}=\{0\}.
	\end{equation}
	Indeed, a nonzero vector in the intersection and its spin flip would be
	linearly independent vectors in the same intersection, forcing
	$\cL=\widetilde{\cL}$. Hence
	$
	\mathsf W_{\mathsf E}:=(\mathsf E,\widetilde{\mathsf E})
	\in\bbC^{8\times4}
	$
	has full column rank whenever $\delta_{\mathsf E}>0$.
	
	We first prove the ``only if'' part. Suppose that $\rho$ is 2-UDA.
	Condition (i) follows from
	Lemmas~\ref{le:3-qubit_UDA} and~\ref{le:lr+1-ls}. In fact, a GHZ-LU
	vector in $\cL$ determines a non-2-UDA pure state in the face
	containing $\rho$, so $\rho$ cannot be 2-UDA. 
	It remains to prove (ii). If $\delta_{\mathsf E}=0$, then the first
	alternative in (ii) already holds. It  remains to consider
	the case $\delta_{\mathsf E}>0$. If (ii) fails, then some
	$F\in\cF_{\mathsf E}$ satisfies $\det F>0$. We choose a
	symmetric $G$ for which
	$H=\Phi_{\mathsf E}(F,G)\in\cN_2$. Since a Hermitian $2\times2$
	matrix with positive determinant is definite, replacing $(F,G)$ by
	$(-F,-G)$ if necessary allows us to assume $F<0$. For all sufficiently
	small $t>0$,
	\begin{equation}
		\label{eq:small_t_Q_positive}
		Q_t=
		\bma
		I_2+tF&tG\\
		tG^\dagger&-t\overline F
		\ema
		\geq0.
	\end{equation}
	By Lemma~\ref{lem:generalized-schur-complement}, this follows because
	$-t\overline F>0$, and its Schur complement in $Q_t$ is
	$$
	I_2+t\left(
	F-G(-\overline F)^{-1}G^\dagger
	\right)>0
	$$
	for sufficiently small $t$. Therefore
	$\rho+tH=\mathsf W_{\mathsf E}Q_t\mathsf W_{\mathsf E}^\dagger$ is a state with the same
	2-marginals as $\rho$. Moreover, \eqref{eq:rank_two_spin_direct_sum}
	implies that $\mathsf W_{\mathsf E}$ has full column rank, so $F\neq0$ implies
	$H\neq0$. Hence $\rho+tH\neq\rho$, contradicting the assumption that
	$\rho$ is 2-UDA. This proves (ii), and thus completes the ``only if''
	part.
	
	We next prove the ``if'' part. Suppose (i) and (ii) hold. Let $\sigma$ be a state
	2-compatible with  $\rho$, and set $H=\sigma-\rho$. Then
	$H\in\cN_2$. Every Pauli term of an element of $\cN_2$ has weight
	three. Eqs.~\eqref{eq:n-qubit-spin-flip} and
	\eqref{eq:n-qubit-spin-flip-Pauli} therefore imply
	$
	\widetilde H=-H.
	$
	It follows that
	$\sigma+\widetilde\sigma=\rho+\widetilde\rho$, 
	and therefore
	\begin{equation}
		\label{eq:sigma_support_spin_sum}
		\cR(\sigma)\subseteq
		\cR(\rho+\widetilde\rho)
		=\cL+\widetilde{\cL}.
	\end{equation}
	
	Suppose first that $\delta_{\mathsf E}=0$. By
	\eqref{eq:delta_zero_spin_invariant} and
	\eqref{eq:sigma_support_spin_sum}, both $\cR(\sigma)$ and $\cR(H)$ are
	contained in the two-dimensional space $\cL$. If $H\neq0$, then $H$ is traceless
	and has the form
	$$
	H=\lambda
	\bigl(
	\ket{\phi_0}\bra{\phi_0}
	-\ket{\phi_1}\bra{\phi_1}
	\bigr)
	$$
	for some $\lambda>0$ and orthonormal vectors
	$\ket{\phi_0},\ket{\phi_1}\in\cL$. Since $H\in\cN_2$, these two pure
	states have the same 2-marginals. From Lemma~\ref{le:3-qubit_UDA},
	$\ket{\phi_0}$ and $\ket{\phi_1}$ are GHZ-LU states, contradicting
	(i). Thus $H=0$.
	
	Suppose next that $\delta_{\mathsf E}>0$. By
	\eqref{eq:rank_two_spin_direct_sum}, $\mathsf W_{\mathsf E}$ has full column rank.
	By writing $H$ in this basis and imposing
	$H=H^\dagger$ and $\widetilde H=-H$, we have  a unique representation
	\begin{equation}
		\label{eq:H_AB_rank_two}
		H=\Phi_{\mathsf E}(F,G),
		\qquad
		F=F^\dagger,\quad G=G^T.
	\end{equation}
	From $H\in\cN_2$, the definition in
	\eqref{eq:F_space_rank_two} implies 
	$F\in\cF_{\mathsf E}$. 
	Since $\rho=\mathsf E\mathsf E^\dagger$, we also have
	\begin{equation}
		\label{eq:Q_AB_positive}
		\sigma=\mathsf W_{\mathsf E}Q\mathsf W_{\mathsf E}^\dagger,
		\qquad
		Q=
		\bma
		I_2+F&G\\
		G^\dagger&-\overline F
		\ema
		\geq0.
	\end{equation}
	The full column rank of $\mathsf W_{\mathsf E}$ makes the positivity of $\sigma$
	equivalent to that of $Q$. In particular, $F\leq0$. If $F=0$, then
	$Q\geq0$ forces $G=0$, and hence $H=0$.
	
	It remains to exclude $\rank F=1$. In that case, we write
	$F=-uu^\dagger$. Positivity of $Q$ implies
	$\cR(G^\dagger)\subseteq\Span\{\overline u\}$. Since
	$G=G^T$, we then obtain $G=\beta uu^T$ for some $\beta\in\bbC$. Let 
	$\ket\psi=\mathsf Eu\in\cL$.  Eq.~\eqref{eq:H_AB_rank_two} becomes
	\begin{equation}
		\label{eq:H_rank_one_A}
		H=
		-\ket\psi\bra\psi
		+\ket{\widetilde\psi}\bra{\widetilde\psi}
		+\beta\ket\psi\bra{\widetilde\psi}
		+\overline\beta\ket{\widetilde\psi}\bra\psi.
	\end{equation}
	By \eqref{eq:braketz=0}, the first two vectors are linearly
	independent, and the coefficient matrix in
	\eqref{eq:H_rank_one_A} has determinant $-1-|\beta|^2$. Thus
	$\rank H=2$. The two orthogonal eigenvectors of $H$ have identical
	2-marginals. After a common local unitary, such a compatible pair
	differs only in the relative phase of a generalized GHZ state. Hence
	their span is $\Span\{\ket{000},\ket{111}\}$. Every nonzero vector in this span is
	either GHZ-LU or proportional to one of the two product endpoints.
	The first case contradicts (i). In the second, the cross terms
	in \eqref{eq:H_rank_one_A} vanish under every single-qubit partial
	trace, whereas the first two terms have a nonzero partial trace. This
	contradicts $H\in\cN_2$. Therefore $\rank F\neq1$.
	Consequently, a nonzero compatible perturbation would require $F<0$
	and hence $\det F>0$, which is contrary to (ii). Thus $H=0$, and 
	$\rho$ is 2-UDA.
\end{resultproof}

\section{Proofs for Sec. ~\ref{subsec:support-reduced-SDP}}
\label{app:proofs-support-reduced-SDP}

\begin{resultproof}{Proposition~\ref{pro:support-reduced-SDP}}
We first establish two properties of the generalized state inversion
in \eqref{eq:generalized-tripartite-state-inversion} that
will be used throughout the proof. For $Y\in\{A,B,C\}$, we write
$d_Y:=\dim\cH_Y$. With respect to a fixed orthonormal basis of
$\cH_Y$, we set
$
A_{pq}^{(Y)}
:=
\ket p\bra q-\ket q\bra p,
$ for $
1\leq p<q\leq d_Y.
$
For every linear operator $Z$ on $\cH_Y$, we have 
$
\mathcal I_Y(Z)
=
\sum_{p<q}
A_{pq}^{(Y)}Z^T\bigl(A_{pq}^{(Y)}\bigr)^\dagger.
$
Consequently, $\widetilde X$ is obtained by applying a completely
positive map to the global transpose $X^T$. Since transposition
preserves positive semidefiniteness,
$X\geq0$ implies $\widetilde X\geq0$. For
$\cS\subseteq[3]$, we use the marginal notation from
Sec.~\ref{sec:pre}, with $(A_1,A_2,A_3)=(A,B,C)$ and the
empty-subsystem conventions stated there. Expanding the three
local inverters in \eqref{eq:generalized-tripartite-state-inversion}, we
obtain that 
\begin{eqnarray}
\label{eq:generalized-inverter-expansion}
\widetilde X
=
\sum_{\cS\subseteq[3]}
(-1)^{|\cS|}X_{A_{\cS}}\ox I_{A_{\cS^c}},
\end{eqnarray}
where every term is embedded in the natural tensor order. If
$H\in\cN_2$, then all 1- and 2-marginals of $H$ vanish. Thus only the
term $\cS=[3]$ remains in
\eqref{eq:generalized-inverter-expansion}, which proves the identity
$\widetilde H=-H$ stated before Proposition~\ref{pro:support-reduced-SDP}.

(i) Since $\rho$ is positive definite on $\cL$, there are constants
$a,b>0$ such that
$
a\tau_{\cL}\leq\rho\leq b\tau_{\cL}.
$
Since the generalized state inversion is positive and linear, applying
it to the preceding inequalities yields
$a\widetilde{\tau_{\cL}}\leq\widetilde\rho
\leq b\widetilde{\tau_{\cL}}$. Hence
\begin{eqnarray}
\label{eq:rho-tau-state-inversion-comparison}
a\bigl(\tau_{\cL}+\widetilde{\tau_{\cL}}\bigr)
\leq
\rho+\widetilde\rho
\leq
b\bigl(\tau_{\cL}+\widetilde{\tau_{\cL}}\bigr).
\end{eqnarray}
Then the positive semidefinite operators  $\rho+\widetilde\rho$ and $\tau_{\cL}+\widetilde{\tau_{\cL}}$ have the same
kernel and the same range $\widehat{\cL}$. 
Let $\sigma$ be a state 2-compatible with $\rho$, and 
$H:=\sigma-\rho$. Then $H\in\cN_2$ and hence
$\widetilde H=-H$. It follows that
\begin{eqnarray}
\label{eq:rho-compatible-state-inversion-sum}
\sigma+\widetilde\sigma
=
\rho+\widetilde\rho.
\end{eqnarray}
Since $\widetilde\sigma\geq0$, we have
$0\leq\sigma\leq\rho+\widetilde\rho$. For positive semidefinite
operators $0\leq X\leq Y$, one has $\ker Y\subseteq\ker X$.
Eqs.~\eqref{eq:rho-tau-state-inversion-comparison} and
\eqref{eq:rho-compatible-state-inversion-sum} therefore imply
$\cR(\sigma)\subseteq\widehat{\cL}$.

(ii) By Lemma~\ref{le:lr+1-ls}, $\rho$ and $\tau_{\cL}$ have the same UDA
property. We characterize the states 2-compatible with
$\tau_{\cL}$. By part (i), every such state $\sigma$ satisfies
$\cR(\sigma)\subseteq\widehat{\cL}$. Therefore
$H:=\sigma-\tau_{\cL}$ belongs to
$\cN_2\cap\operatorname{Herm}(\widehat{\cL})$, and there is a unique
$x\in\bbR^m$ such that $H=\sum_jx_jH_j$. Using
\eqref{eq:support-reduced-isometry}, we then have
\begin{eqnarray}
\label{eq:compatible-state-SDP-coordinate}
\sigma
=
U_{\cL}
\left(
R_{\cL}+\sum_{j=1}^m x_jG_j
\right)
U_{\cL}^\dagger.
\end{eqnarray}
Thus $\sigma\geq0$ if and only if $x\in\Omega_{\cL}$ in
\eqref{eq:state-inversion-spectrahedron}. Conversely,
every $x\in\Omega_{\cL}$ defines through
\eqref{eq:compatible-state-SDP-coordinate} a positive semidefinite
operator of trace one, because every $H_j$ is traceless. It also has
the same 2-marginals as $\tau_{\cL}$ and equals $\tau_{\cL}$ if and
only if $x=0$. We conclude that $\rho$ is 2-UDA if and only if
$\Omega_{\cL}=\{0\}$.

It remains to prove that $\Omega_{\cL}$ is compact. We define the real
linear map
$$
\mathsf T_{\cL}:\bbR^m\ra\operatorname{Herm}(\bbC^{d_{\cL}}),
\qquad
\mathsf T_{\cL}(x):=\sum_{j=1}^m x_jG_j.
$$
This map is injective. Indeed, $\mathsf T_{\cL}(x)=0$ and the fact that
$\cR(H_j)\subseteq\widehat{\cL}$ for every $j$ imply
$
\sum_{j=1}^m x_jH_j
=
U_{\cL}\mathsf T_{\cL}(x)U_{\cL}^\dagger
=0,
$
and the real linear independence of $H_1,\ldots,H_m$ implies $x=0$.
Consequently, there is a constant $\gamma_{\cL}>0$ such that
$\|\mathsf T_{\cL}(x)\|_{\rm HS}\geq\gamma_{\cL}\|x\|_2$, where
$\|Y\|_{\rm HS}:=\sqrt{\tr(Y^\dagger Y)}$ denotes the
Hilbert--Schmidt norm.
For $x\in\Omega_{\cL}$, we set
$A(x):=R_{\cL}+\mathsf T_{\cL}(x)$. Since $A(x)\geq0$ and
$\tr A(x)=1$, we have $\|A(x)\|_{\rm HS}\leq1$. We also have
$\|R_{\cL}\|_{\rm HS}\leq1$, and hence
$
\gamma_{\cL}\|x\|_2
\leq
\|\mathsf T_{\cL}(x)\|_{\rm HS}
\leq2.
$
Thus $\Omega_{\cL}$ is bounded. It is closed as the inverse image of
the positive semidefinite cone under a continuous affine map, and
hence it is compact.

(iii) Let $x^*$ and $H^*$ be as assumed, and set
$\sigma^*:=\tau_{\cL}+H^*$. By the correspondence established in the
proof of (ii), $\sigma^*$ is a state distinct from and 2-compatible
with $\tau_{\cL}$. With $\epsilon$ as in (iii), we obtain
\begin{eqnarray}
\label{eq:original-state-SDP-witness-positive}
\rho^*
=
\rho+\epsilon H^*
=
\bigl(\rho-\epsilon\tau_{\cL}\bigr)
+\epsilon\sigma^*
\geq0,
\end{eqnarray}
because
$\rho-\epsilon\tau_{\cL}
=\rho-\frac12\lambda_{\min}^{+}(\rho)P_{\cL}\geq0$.
Moreover, $H^*$ is nonzero, traceless, and belongs to $\cN_2$.
Therefore $\rho^*$ is a state distinct from and 2-compatible with
$\rho$.
\end{resultproof}

\section{Proofs for Sec. ~\ref{subsec:three-qubit-rank-bound}}
\label{app:proofs-three-qubit-rank-bound}

\begin{resultproof}{Theorem~\ref{th:five_dim_contains_GHZ_LU}}
	We use the characteristic-class facts collected in
	Lemma~\ref{lem:characteristic-class-tools}.
	Let $\cH=(\bbC^2)^{\ox 3}$ and $X=(\mathbb{CP}^1)^3$. For $j=1,2,3$, let $\mathcal L_j$ be the pullback to $X$ of the tautological line bundle on the $j$-th factor, and let $h_j=c_1(\mathcal L_j^*)$. We set $h=h_1+h_2+h_3$. The integral cohomology ring of $X$ is
	\begin{eqnarray}
		\label{eq:cohomology_X_relative}
		H^*(X;\bbZ)=\bbZ[h_1,h_2,h_3]/(h_1^2,h_2^2,h_3^2),
		\qquad \int_Xh_1h_2h_3=1.
	\end{eqnarray}
	We define the smooth complex line bundles
	\begin{eqnarray}
		\label{eq:A_B_bundles_relative}
		\mathcal A=\mathcal L_1\ox\mathcal L_2\ox\mathcal L_3,
		\qquad
		\mathcal B=\mathcal L_1^\perp\ox\mathcal L_2^\perp\ox\mathcal L_3^\perp,
	\end{eqnarray}
	where the orthogonal complements are taken with respect to the standard Hermitian inner product on $\bbC^2$. Since $\mathcal L_j\oplus\mathcal L_j^\perp$ is trivial, we have $c_1(\mathcal L_j)=-h_j$ and $c_1(\mathcal L_j^\perp)=h_j$. Hence, for $\mathcal E=\mathcal A\oplus\mathcal B$,
	\begin{eqnarray}
		\label{eq:Chern_E_relative}
		c_1(\mathcal A)=-h,\qquad c_1(\mathcal B)=h,\qquad
		c_1(\mathcal E)=0,\qquad c_2(\mathcal E)=-h^2.
	\end{eqnarray}
	
	Let $\cW=\cH/\cV$, so $\dim\cW=3$. Fiberwise inclusion into $\cH$, followed by the quotient map $\cH\to\cW$, defines a complex bundle morphism $\Phi:\mathcal E\to X\times\cW$. Let $\pi:M=\mathbb P(\mathcal E)\to X$ be the bundle of complex lines, let $\mathscr T\subset\pi^*\mathcal E$ be the tautological line bundle, and set $\xi=c_1(\mathscr T^*)$. The morphism $\Phi$ induces a section $s$ of the complex rank-three bundle
	\begin{eqnarray}
		\label{eq:F_and_s_relative}
		\mathcal F=\mathscr T^*\ox\cW\simeq(\mathscr T^*)^{\oplus3},
		\qquad s(x,K)=\left.\Phi_x\right|_K,\quad (x,K)\in M,
	\end{eqnarray}
	where $K\subseteq\mathcal E_x$ is a complex line.
	
	For $x=(\ell_1,\ell_2,\ell_3)\in X$, a line in $\mathcal E_x=\mathcal A_x\oplus\mathcal B_x$ outside $\mathbb P(\mathcal A_x)\cup\mathbb P(\mathcal B_x)$ is spanned by a vector $au+bv$, where $ab\neq0$, $u\in\ell_1\ox\ell_2\ox\ell_3$, and $v\in\ell_1^\perp\ox\ell_2^\perp\ox\ell_3^\perp$. Consequently, a zero of $s$ outside the two endpoint sections
	$$D_{\mathcal A}=\mathbb P(\mathcal A),\qquad D_{\mathcal B}=\mathbb P(\mathcal B)$$
	yields the required GHZ-LU state. Suppose, to the contrary, that no required state exists. We then have
	\begin{eqnarray}
		\label{eq:zero_set_endpoints_relative}
		Z(s):=\{(x,K)\in M:s(x,K)=0\}
		\subset D_{\mathcal A}\sqcup D_{\mathcal B}.
	\end{eqnarray}
	
	We define the orthogonal-complement involution
	$$
	\iota:X\longrightarrow X,\qquad
	(\ell_1,\ell_2,\ell_3)\longmapsto
	(\ell_1^\perp,\ell_2^\perp,\ell_3^\perp).
	$$
	On each copy of $\mathbb{CP}^1\simeq S^2$, this is the antipodal antiholomorphic involution. Therefore,
	\begin{eqnarray}
		\label{eq:iota_on_h_relative}
		\iota^*h_j=-h_j,\qquad \iota^*h=-h,
	\end{eqnarray}
	and $\iota$ reverses the orientation of the complex threefold $X$.
	Since $\mathcal A_{\iota(x)}=\mathcal B_x$ and
	$\mathcal B_{\iota(x)}=\mathcal A_x$, the involution lifts to
	$\widehat\iota:M\to M$ by
	$$
	\widehat\iota\bigl(x,[u+v]\bigr)
	=\bigl(\iota(x),[v+u]\bigr),\qquad
	u\in\mathcal A_x,\quad v\in\mathcal B_x.
	$$
	It exchanges $D_{\mathcal A}$ and $D_{\mathcal B}$ and acts
	complex-linearly on every projective fiber. Hence $\widehat\iota$
	reverses the orientation of the real eight-dimensional manifold $M$.
	The induced maps on $\mathscr T$ and $\mathcal F$ are complex linear.
	The action on the factor $\cW$ is the identity, and $s$ is equivariant
	because $u+v=v+u$ in $\cH/\cV$.
	
	We choose disjoint closed tubular neighborhoods with smooth boundary
	\begin{eqnarray}
		\label{eq:endpoint_neighborhoods_relative}
		N_{\mathcal A}\subset M\setminus D_{\mathcal B},\qquad
		N_{\mathcal B}=\widehat\iota(N_{\mathcal A})\subset M\setminus D_{\mathcal A}
	\end{eqnarray}
	of $D_{\mathcal A}$ and $D_{\mathcal B}$, respectively, and set $N=N_{\mathcal A}\sqcup N_{\mathcal B}$ and $C=M\setminus\operatorname{int}N$. By \eqref{eq:zero_set_endpoints_relative}, the section $s$ is nowhere zero on $C$. Let $U_{\mathcal F}\in H^6(\mathcal F,\mathcal F\setminus M;\bbZ)$ be the Thom class. By Lemma~\ref{lem:characteristic-class-tools}(i), pulling back the Thom class along the map of pairs $s:(M,C)\to(\mathcal F,\mathcal F\setminus M)$ defines the relative Euler class
	\begin{eqnarray}
		\label{eq:relative_Euler_class}
		e(\mathcal F,s)=s^*U_{\mathcal F}\in H^6(M,C;\bbZ),
	\end{eqnarray}
	whose image in $H^6(M;\bbZ)$ is $e(\mathcal F)=c_3(\mathcal F)$. Excision shows
	\begin{eqnarray}
		\label{eq:relative_excision}
		H^6(M,C;\bbZ)&\simeq&H^6(N_{\mathcal A},\partial N_{\mathcal A};\bbZ)
		\oplus H^6(N_{\mathcal B},\partial N_{\mathcal B};\bbZ).
	\end{eqnarray}
	We denote the two components by $\varepsilon_{\mathcal A}$ and $\varepsilon_{\mathcal B}$, and define
	\begin{eqnarray}
		\label{eq:endpoint_correction_numbers}
		I_{\mathcal A}&=&\left\langle\left.\pi^*h\right|_{N_{\mathcal A}}\smile\varepsilon_{\mathcal A},
		[N_{\mathcal A},\partial N_{\mathcal A}]\right\rangle,
		\nonumber\\
		I_{\mathcal B}&=&\left\langle\left.\pi^*h\right|_{N_{\mathcal B}}\smile\varepsilon_{\mathcal B},
		[N_{\mathcal B},\partial N_{\mathcal B}]\right\rangle.
	\end{eqnarray}
	The relative-to-absolute map and excision imply
	\begin{eqnarray}
		\label{eq:endpoint_correction_identity}
		\left\langle\pi^*h\smile c_3(\mathcal F),[M]\right\rangle
		=I_{\mathcal A}+I_{\mathcal B}.
	\end{eqnarray}
	A perturbation of $s$ in $N_{\mathcal A}$, fixed near $\partial N_{\mathcal A}$, can be chosen transverse to the zero section. We transport it to $N_{\mathcal B}$ by $\widehat\iota$ and keep $s$ unchanged on $C$. The resulting equivariant section $\widetilde s$ is homotopic to $s$ relative to $C$ and has a zero set consisting of closed oriented surfaces
	\begin{eqnarray}
		\label{eq:relative_zero_surfaces}
		Z(\widetilde s)=\Sigma_{\mathcal A}\sqcup\Sigma_{\mathcal B},\qquad
		\Sigma_{\mathcal A}\subset\operatorname{int}N_{\mathcal A},\qquad
		\Sigma_{\mathcal B}\subset\operatorname{int}N_{\mathcal B}.
	\end{eqnarray}
	Poincar\'e--Lefschetz duality shows
	\begin{eqnarray}
		\label{eq:corrections_as_integrals}
		I_{\mathcal A}=\int_{\Sigma_{\mathcal A}}\pi^*h,
		\qquad I_{\mathcal B}=\int_{\Sigma_{\mathcal B}}\pi^*h.
	\end{eqnarray}
	Since $\widehat\iota$ reverses the orientation of $M$ and preserves the orientation of $\mathcal F$, it reverses the induced orientation of the zero surface. Thus
	\begin{eqnarray}
		\label{eq:orientation_Sigma_relative}
		[\Sigma_{\mathcal B}]=-\widehat\iota_*[\Sigma_{\mathcal A}].
	\end{eqnarray}
	
	The vertical tangent bundle is $T_{M/X}=\operatorname{Hom}(\mathscr T,\pi^*\mathcal E/\mathscr T)$. By $c_1(\mathcal E)=0$ and $c_1(\mathscr T)=-\xi$, we have $c_1(T_{M/X})=2\xi$. Moreover, $c_1(TX)=2h$. The tangent bundle sequence of $\pi$ therefore shows
	\begin{eqnarray}
		\label{eq:M_spin_relative}
		w_2(TM)=\pi^*w_2(TX)+w_2(T_{M/X})
		=2\pi^*h+2\xi=0.
	\end{eqnarray}
	For either $\Sigma_*$ in \eqref{eq:relative_zero_surfaces}, transversality identifies its real normal bundle with $\mathcal F|_{\Sigma_*}$. Since $c_1(\mathcal F)=3\xi$, we have $w_2(\mathcal F)=\xi\pmod2$. Every closed oriented surface satisfies $w_2(T\Sigma_*)=0$. The Whitney formula and the identity for $w_2$ in
	Lemma~\ref{lem:characteristic-class-tools}(ii), together with
	\eqref{eq:M_spin_relative}, then imply
	\begin{eqnarray}
		\label{eq:xi_even_on_Sigma_relative}
		\left.\xi\right|_{\Sigma_*}=0
		\quad\text{in}\quad H^2(\Sigma_*;\bbZ/2).
	\end{eqnarray}
	Consequently, the integral of $\xi$ over every component of $\Sigma_*$ is even.
	
	On $M\setminus D_{\mathcal B}$, projection onto $\mathcal A$ identifies $\mathscr T$ with $\pi^*\mathcal A$, whereas on $M\setminus D_{\mathcal A}$, projection onto $\mathcal B$ identifies $\mathscr T$ with $\pi^*\mathcal B$. It follows from \eqref{eq:Chern_E_relative} that
	\begin{eqnarray}
		\label{eq:xi_endpoint_charts_relative}
		\left.\xi\right|_{M\setminus D_{\mathcal B}}=\pi^*h,
		\qquad
		\left.\xi\right|_{M\setminus D_{\mathcal A}}=-\pi^*h.
	\end{eqnarray}
	We set $d=\int_{\Sigma_{\mathcal A}}\pi^*h$. Eqs.~\eqref{eq:endpoint_neighborhoods_relative}, \eqref{eq:xi_even_on_Sigma_relative}, and \eqref{eq:xi_endpoint_charts_relative} show that $d$ is even. By \eqref{eq:iota_on_h_relative} and \eqref{eq:orientation_Sigma_relative},
	\begin{eqnarray}
		\label{eq:relative_mod_four}
		I_{\mathcal B}
		=\int_{\Sigma_{\mathcal B}}\pi^*h
		=-\int_{\Sigma_{\mathcal A}}\widehat\iota^*\pi^*h
		=d=I_{\mathcal A}.
	\end{eqnarray}
	Hence the endpoint correction identity \eqref{eq:endpoint_correction_identity} satisfies
	\begin{eqnarray}
		\label{eq:endpoint_correction_mod_four}
		\left\langle\pi^*h\smile c_3(\mathcal F),[M]\right\rangle
		=I_{\mathcal A}+I_{\mathcal B}=2d\equiv0\pmod4.
	\end{eqnarray}
	
	Let $\mathcal Q_{\mathcal E}:=\pi^*\mathcal E/\mathscr T$. Since
	$c_1(\mathcal Q_{\mathcal E})=\xi$, the tautological exact sequence in
	Lemma~\ref{lem:characteristic-class-tools}(iii) and
	\eqref{eq:Chern_E_relative} yield
	\begin{eqnarray}
		\label{eq:projective_bundle_relation_relative}
		c_2(\pi^*\mathcal E)=c_1(\mathscr T)c_1(\mathcal Q_{\mathcal E})
		=-\xi^2=-\pi^*h^2,
		\qquad \xi^2=\pi^*h^2.
	\end{eqnarray}
	By Lemma~\ref{lem:characteristic-class-tools}(iii), $\xi$ restricts
	to $c_1(\mathcal O_{\mathbb{CP}^1}(1))$ on every fiber and
	$\pi_*(\xi)=1$. Eq.~\eqref{eq:projective_bundle_relation_relative} yields $\pi_*(\xi^3)=h^2$. On the other hand, $c_3(\mathcal F)=\xi^3$. Hence, by \eqref{eq:cohomology_X_relative},
	\begin{eqnarray}
		\label{eq:Chern_number_six_relative}
		\left\langle\pi^*h\smile c_3(\mathcal F),[M]\right\rangle
		&=&\int_M\pi^*h\,\xi^3
		=\int_Xh\,\pi_*(\xi^3)
		\nonumber\\
		&=&\int_Xh^3=6.
	\end{eqnarray}
	Eq.~\eqref{eq:Chern_number_six_relative} contradicts the divisibility by four in \eqref{eq:endpoint_correction_mod_four}. Therefore, the original section $s$ has a zero outside $D_{\mathcal A}\cup D_{\mathcal B}$.
	For the UDA consequence, let $\rho$ be a three-qubit state with
	$\rank(\rho)\geq5$, and choose a five-dimensional subspace
	$\cV_0\subseteq\cR(\rho)$. By the preceding result, $\cV_0$, and hence
	$\cR(\rho)$, contains a GHZ-LU vector. This vector determines a
	non-2-UDA pure state in the face containing $\rho$. Therefore, by
	Lemmas~ \ref{le:lr+1-ls} and \ref{le:3-qubit_UDA}, $\rho$ is not 2-UDA.
\end{resultproof}

\section{Proofs for Sec. ~\ref{subsec:high-rank-non-2UDA}}
\label{app:proofs-high-rank-non-2UDA}

\begin{resultproof}{Lemma~\ref{le:dimension_non_UDA_criterion}}
(i) Let $\cL=\cR(\rho)$. For each $i\in[n]$, we choose an orthonormal
Hermitian basis
$
F_0^{(i)}=\frac{I_{A_i}}{\sqrt{d_i}},
$ for $
F_1^{(i)},\ldots,F_{d_i^2-1}^{(i)},
$
where $F_a^{(i)}$ is traceless for $a\geq1$. For
$\boldsymbol{\alpha}=(\alpha_1,\ldots,\alpha_n)$ with
$0\leq\alpha_i\leq d_i^2-1$, we set 
$
F_{\boldsymbol{\alpha}}
:=
F_{\alpha_1}^{(1)}\ox\cdots\ox F_{\alpha_n}^{(n)}$,
and
$\supp(\boldsymbol{\alpha})
:=
\{i\in[n]:\alpha_i\neq0\}.
$ 
These product operators form an orthonormal real basis of
$\operatorname{Herm}(\cH)$. For $\cS\subseteq[n]$ with
$|\cS|=k$, the partial trace
$\tr_{A_{\cS^c}}F_{\boldsymbol{\alpha}}$ is nonzero if and only if
$\supp(\boldsymbol{\alpha})\subseteq\cS$.
Consequently, $H\in\cN_k$ if and only if all coefficients of $H$
whose supports have size at most $k$ vanish. Hence
\begin{eqnarray}
\label{eq:N_k_product_basis_dimension}
\cN_k
&=&\Span_{\bbR}\left\{
F_{\boldsymbol{\alpha}}:
|\supp(\boldsymbol{\alpha})|>k
\right\},
\nonumber\\
\dim_{\bbR}\cN_k
&=&
\sum_{\substack{\cS\subseteq[n]\\|\cS|>k}}
\prod_{i\in\cS}(d_i^2-1)
=
\prod_{i=1}^n d_i^2
-
\sum_{\substack{\cS\subseteq[n]\\|\cS|\leq k}}
\prod_{i\in\cS}(d_i^2-1).
\end{eqnarray}
The last equality follows from
$\prod_{i=1}^n d_i^2=\prod_{i=1}^n[1+(d_i^2-1)]$.

For every $H\in\operatorname{Herm}(\cL)$, hermiticity and
$\cR(H)\subseteq\cL$ imply that $Hx=0$ for every
$x\in\cL^\perp$. Thus $\operatorname{Herm}(\cL)$ is naturally
isomorphic to the space of $r\times r$ Hermitian matrices, and hence
$\dim_{\bbR}\operatorname{Herm}(\cL)=r^2$, while
$\dim_{\bbR}\operatorname{Herm}(\cH)=\prod_{i=1}^n d_i^2$.
Applying the dimension formula for two real linear subspaces, we
obtain
\begin{eqnarray}
\label{eq:Herm_L_N_k_intersection}
\dim_{\bbR}\bigl(\operatorname{Herm}(\cL)\cap\cN_k\bigr)
&\geq&
r^2+\dim_{\bbR}\cN_k-\prod_{i=1}^n d_i^2
\nonumber\\
&=&
r^2-
\sum_{\substack{\cS\subseteq[n]\\|\cS|\leq k}}
\prod_{i\in\cS}(d_i^2-1).
\end{eqnarray}
Under \eqref{eq:dimension_non_UDA_criterion}, there is therefore a
nonzero operator $H\in\operatorname{Herm}(\cL)\cap\cN_k$.

Let $\lambda_{\min}(\rho|_{\cL})>0$ be the smallest eigenvalue of the
restriction of $\rho$ to $\cL$. We denote the operator norm by
$\|\cdot\|_\infty$. For
$0<\epsilon<\lambda_{\min}(\rho|_{\cL})/\|H\|_\infty$, we have
$$
\rho+\epsilon H
\geq
\bigl(\lambda_{\min}(\rho|_{\cL})
-\epsilon\|H\|_\infty\bigr)P_{\cL}
\geq0,
$$
where $P_{\cL}$ is the projection onto $\cL$. Moreover,
$H\in\cN_k$ implies $\tr H=0$ and $\cM_k(H)=0$. Thus
$\rho+\epsilon H$ is a state distinct from $\rho$ and satisfies
$\cM_k(\rho+\epsilon H)=\cM_k(\rho)$. Therefore, $\rho$ is not
$k$-UDA.

(ii) We set $D:=2^n$ and write
\begin{eqnarray}
\label{eq:full-weight-Pauli-space}
\cF_n
:=
\cN_{n-1}
=
\Span_{\bbR}\left\{
P_{\boldsymbol{\alpha}}:
\operatorname{wt}(\boldsymbol{\alpha})=n
\right\}.
\end{eqnarray}
Thus $\dim_{\bbR}\cF_n=3^n$. By
\eqref{eq:n-qubit-spin-flip-Pauli}, every $H\in\cF_n$ satisfies
$\Theta_nH\Theta_n^{-1}=(-1)^nH$. We introduce the real linear space
\begin{eqnarray}
\label{eq:spin-parity-Hermitian-space}
\cA_n
:=
\left\{
H\in\operatorname{Herm}(\cH_n):
\Theta_nH\Theta_n^{-1}=(-1)^nH
\right\}.
\end{eqnarray}
Then $P_{\boldsymbol{\alpha}}\in \cA_n$ if and only if $\operatorname{wt}(\boldsymbol{\alpha})\equiv n \mod 2$, and hence 
\begin{eqnarray}
\label{eq:spin-parity-Hermitian-dimension}
\dim_{\bbR}\cA_n
&=&
\sum_{\substack{0\leq j\leq n\\j\equiv n \mod 2}}
\binom{n}{j}3^j
=
\frac{4^n+2^n}{2}
=
\frac{D(D+1)}{2}.
\end{eqnarray} 

On the other hand, the subspace $\widehat{\cK}$ is invariant under $\Theta_n$. We define
\begin{eqnarray}
\label{eq:spin-parity-annihilator-space}
\cW_{\widehat{\cK}}
:=
\left\{H\in\cA_n:H|_{\widehat{\cK}}=0\right\}.
\end{eqnarray}
We show the dimension of $\cW_{\widehat{\cK}}$ as follows. If $n$ is even, then
$\Theta_n^2=I$. We choose a $\Theta_n$-real orthonormal basis $e_1, ..., e_D$, such that $\Theta_n e_j=e_j$ for $j=1,2, ..., D$. Then $\Theta_n H \Theta_n^{-1}=\overline{H}$. From $H\in \cA_n$, we have $\overline{H}=H$. Combining with $H=H^\dg$, we obtain $H=\overline{H}=H^T$.   Therefore, in this basis, the elements of $\cA_n$ are real
symmetric matrices. The condition $H|_{\widehat{\cK}}=0$ implies
$He_j=0$ for $j=1,\ldots,t$, where
$t=\dim\widehat{\cK}$, and thus removes the first $t$
rows and columns. Hence
$\dim_{\bbR}\cW_{\widehat{\cK}}=(D-t)(D-t+1)/2$.
Suppose  $n$ is odd and  $D=2m$. Since
$\Theta_n^2=-I$, the $\Theta_n$-invariant space $\widehat{\cK}$ has even
dimension, which we denote by $t=2p$. We choose a Kramers orthonormal
basis $e_1,\Theta_ne_1,\ldots,e_m,\Theta_ne_m$ adapted to
$\widehat{\cK}$. Up
to an ordering of this basis, every element of $\cA_n$ has the form
$
H=
\bma
A&B\\
B^\dagger&-\overline A
\ema$, for 
$A=A^\dagger$,
and
$B=B^T$. 
This space has real dimension $m^2+m(m+1)=D(D+1)/2$. We choose the
first $p$ Kramers pairs to span $\widehat{\cK}$. The condition
$H|_{\widehat{\cK}}=0$ leaves the same block form on the $(D-t)$-dimensional
orthogonal complement. Therefore, in both parity cases, we have 
\begin{eqnarray}
\label{eq:spin-refined-codimension}
\dim_{\bbR}\cW_{\widehat{\cK}}
=
\frac{(D-t)(D-t+1)}{2},
\qquad
\operatorname{codim}_{\cA_n}\cW_{\widehat{\cK}}
=
\frac{t(2D-t+1)}{2}.
\end{eqnarray}

Since $\cF_n\subseteq\cA_n$, \eqref{eq:spin-refined-codimension}
implies
\begin{eqnarray}
\label{eq:full-weight-annihilator-lower-bound}
\dim_{\bbR}(\cF_n\cap\cW_{\widehat{\cK}})
&\geq & \dim _\bbR \cF_n+ \dim _\bbR \cW_{\widehat{\cK}}-\dim_\bbR \cA_n
 \\
&=&  \dim _\bbR \cF_n- \operatorname{codim}_{\cA_n}\cW_{\widehat{\cK}}
\nonumber \\
&=&3^n-\frac{t(2D-t+1)}{2}. \nonumber
\end{eqnarray}
Under the strict inequality in
\eqref{eq:spin-flip-refined-obstruction}, we choose a nonzero
$H\in\cF_n\cap\cW_{\widehat{\cK}}$. Let
$\cL:=\cR(\rho)=\cK^\perp$. Since
$H\cK=0$ and $H$ is Hermitian, we have
$H=P_{\cL}HP_{\cL}$. For all sufficiently small nonzero real
$\epsilon$, the operator $\rho+\epsilon H$ is positive semidefinite.
Moreover, $H\in\cN_{n-1}$ is traceless and has vanishing
$(n-1)$-marginals. Thus $\rho+\epsilon H$ is a state distinct from
and $(n-1)$-compatible with $\rho$.

 It remains to consider equality in
\eqref{eq:spin-flip-refined-obstruction} when $n$ is even. We consider
the compression space
$$
\left\{P_{\cK}HP_{\cK}:H\in\cF_n\right\}
\subseteq\operatorname{Herm}(\cK).
$$
Here $P_{\cK}$ denotes the orthogonal projection onto $\cK$.
Suppose first that this space contains a positive definite operator.
We choose $H\in\cF_n$ and write, relative to
$\cH_n=\cL\oplus\cK$, we have 
$
\rho=
\bma
\rho_{\cL}&0\\
0&0
\ema$,
and
$H=
\bma
A&B^\dagger\\
B&C
\ema,
$
where $
\rho_{\cL}>0$ and 
$ C>0.$
For all sufficiently small $\epsilon>0$, the lower-right block of
$\rho+\epsilon H$ is positive definite and its Schur complement is
$
\rho_{\cL}
+\epsilon\left(A-B^\dagger C^{-1}B\right)>0.
$
Hence $\rho+\epsilon H$ is a distinct compatible state.

Suppose now that the compression space contains no positive definite
operator. By finite-dimensional semidefinite separation, there is a
nonzero operator $\tau\geq0$ with $\cR(\tau)\subseteq\cK$ such that
\begin{eqnarray}
\label{eq:spin-refined-separator}
\tr(\tau H)=0,
\qquad
H\in\cF_n.
\end{eqnarray}
We set $\Xi:=\tau+\Theta_n\tau\Theta_n^{-1}$. Since $n$ is even,
$\Xi$ is a nonzero element of $\cA_n$ satisfying
$\cR(\Xi)\subseteq\widehat{\cK}$. Eqs.~
\eqref{eq:n-qubit-spin-flip-Pauli} and
\eqref{eq:spin-refined-separator} imply that
$\tr(\Xi H)=0$ for every $H\in\cF_n$. Indeed, the second summand
contributes
$\overline{\tr(\tau\Theta_n^{-1}H\Theta_n)}
=\overline{\tr(\tau H)}=0$.
The range condition for $\Xi$ also implies
$\tr(\Xi H)=0$ for every $H\in\cW_{\widehat{\cK}}$. Consequently,
both $\cF_n$ and $\cW_{\widehat{\cK}}$ are contained in the proper
hyperplane $\Xi^\perp:=\{H\in \cA_n: \tr(\Xi H)=0\}$, and hence $\dim_{\bbR} \Xi^\perp=\dim_\bbR \cA_n -1$. At
equality, 
\eqref{eq:spin-refined-codimension} yields
$$
\dim_{\bbR}\cF_n+\dim_{\bbR}\cW_{\widehat{\cK}}
=\dim_{\bbR}\cA_n.
$$
It follows that
$\dim_{\bbR}(\cF_n\cap\cW_{\widehat{\cK}})\geq1$. A nonzero
operator in this intersection produces a range-constrained
compatible perturbation by the preceding argument. This completes
the proof.
\end{resultproof}

\begin{resultproof}{Theorem~\ref{th:high_rank_dimension_obstruction}}
(i) We use the equal-local-dimension case of
Lemma~\ref{le:dimension_non_UDA_criterion}(i) with $d=2$.
Let $r=\rank\rho$, and suppose $1\le k\le n-2$.
Since $r\ge2^n-n$, we obtain
\begin{eqnarray}
\label{eq:dimension_gap_k_at_most_n_minus_2}
r^2-\sum_{j=0}^{k}\binom{n}{j}3^j
&\ge&
(2^n-n)^2-\sum_{j=0}^{n-2}\binom{n}{j}3^j
\nonumber\\
&=&
3^n+n3^{n-1}-2n2^n+n^2.
\end{eqnarray}
For every $n\ge3$, the inequalities
$3^{n-1}\ge2^n$ and $3^n\ge n2^n$ hold. It follows that the last
expression in \eqref{eq:dimension_gap_k_at_most_n_minus_2} is at least
$n^2>0$. Hence Lemma~\ref{le:dimension_non_UDA_criterion}(i) implies that
$\rho$ is not $k$-UDA for every $1\le k\le n-2$.

 (ii) Suppose $n\geq6$. Let $\cK:=\ker\rho$ and
$t:=\dim(\cK+\Theta_n\cK)$. The rank assumption implies
$\dim\cK\leq n$, and hence $t\leq2n$. The function
$t(2^{n+1}-t+1)/2$ is increasing for $0\leq t\leq2n$. We claim that
\begin{eqnarray}
\label{eq:spin-refined-rank-threshold-n-at-least-six}
3^n>
n\bigl(2^{n+1}-2n+1\bigr),
\qquad n\geq6.
\end{eqnarray}
For $n=6$, Eq.~\eqref{eq:spin-refined-rank-threshold-n-at-least-six}
is $729>702$. Suppose it holds for some $n\geq6$. We have
$$
3n(2^{n+1}-2n+1)
-(n+1)(2^{n+2}-2n-1)
=(n-2)2^{n+1}-4n^2+6n+1>0.
$$
The last expression is positive at $n=6$. The difference between its
values at $n+1$ and $n$ is $n2^{n+1}-8n+2>0$, so it strictly
increases. Therefore, the claim follows by induction. We now obtain
$$
3^n>
n(2^{n+1}-2n+1)
\geq
\frac{t(2^{n+1}-t+1)}2.
$$
Lemma~\ref{le:dimension_non_UDA_criterion}(ii) shows that $\rho$ is
not $(n-1)$-UDA.

Next, we consider the case $n=4$ and $\rank\rho\geq13$. We set
$\cK:=\ker\rho$. Then $\dim\cK\leq3$, and hence
$t:=\dim(\cK+\Theta_4\cK)\leq6$. The function
$t(2^5-t+1)/2$ is increasing for $0\leq t\leq6$, and
$
\frac{6(2^5-6+1)}2=81=3^4,
$
the right-hand side of
\eqref{eq:spin-flip-refined-obstruction} is at most $3^4$. It is
strictly smaller when $t<6$, while equality is covered because $n$ is
even. Lemma~\ref{le:dimension_non_UDA_criterion}(ii) therefore shows
that $\rho$ is not 3-UDA. 
Finally, we assume that $n=5$ and $\rank\rho\geq28$. We set
$\cK:=\ker\rho$. Then $\dim\cK\leq4$ and
$t:=\dim(\cK+\Theta_5\cK)\leq8$. The function
$t(2^6-t+1)/2$ is increasing for $0\leq t\leq8$, and
$
\frac{8(2^6-8+1)}2=228<243=3^5,
$
so Lemma~\ref{le:dimension_non_UDA_criterion}(ii) shows that $\rho$ is
not 4-UDA. In each case, equality of the $(n-1)$-marginals implies
equality of every $k$-marginal with $k\leq n-1$. Hence $\rho$ is not
$k$-UDA for any $1\leq k\leq n-1$. 
\end{resultproof}

\section{Proofs for Sec. ~\ref{subsec:multipartite-non-2UDA}}
\label{app:proofs-multipartite-non-2UDA}

\begin{resultproof}{Lemma~\ref{le:rho_sigma_compatible}}
By \eqref{eq:rho_sigma_channel_states}, the difference between the two states is
$$
\sigma_{\Lambda,p}^{(n)}-\rho_{\Lambda,p}^{(n)}
=\sum_{j\neq k}\sqrt{p_jp_k}
\big(\ket{j}\bra{k}\big)^{\ox(n-1)}
\ox\Lambda(\ket{j}\bra{k}).
$$
Tracing out any one of $A_1,\ldots,A_{n-1}$ eliminates every term in this sum. Tracing out $A_n$ also eliminates every term because $\Lambda$ is trace preserving and $\tr\Lambda(\ket{j}\bra{k})=0$ for $j\neq k$. Hence the two states have the same $(n-1)$-marginals. Since $n\geq 3$, they are also 2-compatible. Lemma \ref{lem:unique_channel_from_diagonal} characterizes when the choice of $\Lambda$ can make these states distinct for every $s\geq 2$.
\end{resultproof}

\begin{resultproof}{Lemma~\ref{lem:unique_channel_from_diagonal}}
(i) We define the measure-and-prepare channel $\Lambda_0$ by \eqref{eq:Lambda(X)=}. This channel belongs to the set in \eqref{eq:channel_set_diagonal}, so the set is nonempty.

(ii) For any $\Lambda\in\cA(\alpha_1,\ldots,\alpha_d)$,
we write $B_{j,k}:=\Lambda(\ket{j}\bra{k})$. Its Choi matrix is
\begin{eqnarray}
\label{eq:choi_block_matrix}
J(\Lambda)=\sum_{j,k=1}^d\ket{j}\bra{k}\ox B_{j,k}
=(B_{j,k})_{j,k=1}^d\geq 0.
\end{eqnarray}
Since $\Lambda$ is trace preserving and $B_{j,j}=\alpha_j$, we have
\begin{eqnarray}
\label{eq:trace_off_diagonal_zero}
B_{j,j}=\alpha_j,\qquad
\tr B_{j,k}=\delta_{j,k}.
\end{eqnarray}
We denote by $P_j$ the projection onto $\cR(\alpha_j)$. The positivity
in \eqref{eq:choi_block_matrix} implies
\begin{eqnarray}
\label{eq:support_restriction}
B_{j,k}=P_jB_{j,k}P_k,\qquad j,k=1,\ldots,d.
\end{eqnarray}
Indeed, for any $x\in\ker\alpha_j$, the positivity of the principal block with diagonal blocks $\alpha_j$ and $\alpha_k$ implies $B_{j,k}^{\dagger}x=0$. Applying the same argument to $y\in\ker\alpha_k$ implies $B_{j,k}y=0$, which proves \eqref{eq:support_restriction}.

For $j\neq k$, we define the linear space
\begin{eqnarray}
\label{eq:off_diagonal_support_kernel}
\cG_{j,k}
:=\{X\in\bbM_s:\ X=P_jXP_k,\ \tr X=0\}.
\end{eqnarray}
Let $r_j:=\dim\cR(\alpha_j)$. The space $P_j\bbM_sP_k$ has complex
dimension $r_jr_k$, and $\cG_{j,k}$ is the kernel of the trace
functional restricted to this space. Consequently,
$\cG_{j,k}=\{0\}$ if and only if $r_j=r_k=1$ and this trace
functional is nonzero. If $\alpha_j=\ket{a_j}\bra{a_j}$ and
$\alpha_k=\ket{a_k}\bra{a_k}$ are pure, then $P_j\bbM_sP_k$ is
spanned by $\ket{a_j}\bra{a_k}$, whose trace is
$\braket{a_k}{a_j}$. Thus
\begin{eqnarray}
\label{eq:kernel_zero_criterion}
\cG_{j,k}=\{0\}
\quad\Longleftrightarrow\quad
\alpha_j,\alpha_k\text{ are pure and }\tr(\alpha_j\alpha_k)>0.
\end{eqnarray}
If the condition in (ii) holds, \eqref{eq:support_restriction}, \eqref{eq:trace_off_diagonal_zero}, and \eqref{eq:kernel_zero_criterion} imply that $B_{j,k}=0$ for all $j\neq k$. Therefore $\Lambda=\Lambda_0$ in \eqref{eq:Lambda(X)=}, which proves uniqueness.

(iii) Suppose the condition in (ii) fails. By
\eqref{eq:kernel_zero_criterion}, there exist $r\neq t$ and a nonzero
$X\in\cG_{r,t}$. More explicitly, if some $\alpha_r$ is mixed,
we choose any $t\neq r$. Then $r_rr_t\geq 2$, so the kernel in
\eqref{eq:off_diagonal_support_kernel} is nonzero. If all $\alpha_j$
are pure, the failure of the condition shows an orthogonal pair, for
which \eqref{eq:kernel_zero_criterion} again applies.

For $z\in\bbC$, we define
\begin{eqnarray}
\label{eq:perturbed_choi}
J_z
=J(\Lambda_0)+z\ket{r}\bra{t}\ox X
+\bar z\ket{t}\bra{r}\ox X^\dagger.
\end{eqnarray}
Since $X=P_rXP_t$, there is an operator
$C:\cR(\alpha_t)\ra\cR(\alpha_r)$ such that
$X=\alpha_r^{1/2}C\alpha_t^{1/2}$. For
$|z|\|C\|_\infty\leq 1$, the only modified principal block of
\eqref{eq:perturbed_choi} is positive because
$$
\bma
\alpha_r&zX\\
\bar zX^\dagger&\alpha_t
\ema
=
\bma
\alpha_r^{1/2}&0\\
0&\alpha_t^{1/2}
\ema
\bma
P_r&zC\\
\bar zC^\dagger&P_t
\ema
\bma
\alpha_r^{1/2}&0\\
0&\alpha_t^{1/2}
\ema
\geq 0.
$$
Hence $J_z\geq 0$. Moreover, $\tr X=0$ by
\eqref{eq:off_diagonal_support_kernel}, so the partial trace of $J_z$
over the output system is $I_d$. Thus $J_z$ is the Choi matrix of a
channel $\Lambda_z\in\cA(\alpha_1,\ldots,\alpha_d)$.
For every sufficiently small nonzero $z$, we have
$\Lambda_z(\ket{r}\bra{t})=zX\neq 0$. By varying $z$, we obtain
infinitely many distinct channels. This also proves the necessity in
(ii) and completes the proof.
\end{resultproof}

\begin{resultproof}{Proposition~\ref{pro:UDA_diagonal(n-1)}}
(i) We first prove necessity by contraposition. We define
$\Lambda_0:\bbM_d\ra\bbM_s$ by \eqref{eq:Lambda(X)=}. Suppose the
condition in (i) fails. By Lemma~\ref{lem:unique_channel_from_diagonal}
(iii), there are a channel
$\Lambda\in\cA(\alpha_1,\ldots,\alpha_d)$ and indices
$r\neq t$ such that $\Lambda(\ket{r}\bra{t})\neq 0$. We define
\begin{eqnarray}
\label{eq:sigma_alpha_p_n}
\sigma_{\Lambda,p}^{(n)}
=\sum_{j,k=1}^d\sqrt{p_jp_k}
\big(\ket{j}\bra{k}\big)^{\ox(n-1)}
\ox\Lambda(\ket{j}\bra{k}).
\end{eqnarray}
By \eqref{eq:rho_sigma_channel_states}, the operator in \eqref{eq:sigma_alpha_p_n} is a state. Lemma \ref{le:rho_sigma_compatible} shows that $\sigma_{\Lambda,p}^{(n)}$ and $\rho_{\alpha,p}^{(n)}$ have the same $(n-1)$-marginals. They are distinct because the $(r,t)$ block in \eqref{eq:sigma_alpha_p_n} is nonzero. Hence $\rho_{\alpha,p}^{(n)}$ is not $(n-1)$-UDA.

We next prove sufficiency. Suppose every $\alpha_j$ is pure and
$\tr(\alpha_j\alpha_k)>0$ for any $j\neq k$. Let $\tau$ be a state
having the same $(n-1)$-marginals as $\rho_{\alpha,p}^{(n)}$. We
define
$\cQ:=\Span\{\ket{j\cdots j}:j=1,\ldots,d\}
\subset(\bbC^d)^{\ox(n-1)}$, and denote by $P_{\cQ}$ the projection
onto $\cQ$. The equality of the marginal on
$A_1\cdots A_{n-1}$ shows
\begin{eqnarray}
\label{eq:tau_support_Q}
\tau=(P_{\cQ}\ox I_s)\tau(P_{\cQ}\ox I_s).
\end{eqnarray}
Therefore, there exist $B_{j,k}\in\bbM_s$ such that
\begin{eqnarray}
\label{eq:tau_block_decomposition}
\tau
=\sum_{j,k=1}^d
\ket{j\cdots j}\bra{k\cdots k}\ox B_{j,k}.
\end{eqnarray}
For any $1\leq \ell\leq n-1$, tracing out $A_\ell$ in \eqref{eq:tau_block_decomposition} and comparing the result with the corresponding marginal of \eqref{eq:rho_alpha_p_n} shows $B_{j,j}=p_j\alpha_j$. Comparing the marginals on $A_1\cdots A_{n-1}$ shows
\begin{eqnarray}
\label{eq:B_constraints_UDA}
B_{j,j}=p_j\alpha_j,\qquad
\tr B_{j,k}=p_j\delta_{j,k}.
\end{eqnarray}
We define the linear map $\Gamma:\bbM_d\ra\bbM_s$ by
\begin{eqnarray}
\label{eq:Gamma_from_tau}
\Gamma(\ket{j}\bra{k})
=\frac{B_{j,k}}{\sqrt{p_jp_k}},
\qquad j,k=1,\ldots,d.
\end{eqnarray}
Let $V:\bbC^d\ra(\bbC^d)^{\ox(n-1)}$ be the isometry defined by $V\ket{j}=\ket{j\cdots j}$, and let $D:=\sum_{j=1}^d p_j\ket{j}\bra{j}$. By \eqref{eq:tau_support_Q}, the Choi matrix of $\Gamma$ satisfies
\begin{eqnarray}
\label{eq:Choi_Gamma_positive}
J(\Gamma)
=(D^{-1/2}V^\dagger\ox I_s)\tau
(VD^{-1/2}\ox I_s)\geq 0.
\end{eqnarray}
Thus $\Gamma$ is completely positive. Eqs.~\eqref{eq:B_constraints_UDA} and \eqref{eq:Gamma_from_tau} show that $\tr\Gamma(\ket{j}\bra{k})=\delta_{j,k}$ and $\Gamma(\ket{j}\bra{j})=\alpha_j$. Therefore, $\Gamma$ is a channel in the set \eqref{eq:channel_set_diagonal}. Lemma \ref{lem:unique_channel_from_diagonal} (ii) implies that $\Gamma=\Lambda_0$. It follows from \eqref{eq:Gamma_from_tau} that $B_{j,k}=0$ for $j\neq k$. Substituting this into \eqref{eq:tau_block_decomposition} and using \eqref{eq:B_constraints_UDA}, we obtain $\tau=\rho_{\alpha,p}^{(n)}$. Hence $\rho_{\alpha,p}^{(n)}$ is $(n-1)$-UDA.

(ii) If $k=n-1$ and the condition in (i) holds, the conclusion follows from (i). Conversely, suppose $\rho_{\alpha,p}^{(n)}$ is $k$-UDA for a fixed $1\leq k\leq n-1$. Since equality of all $(n-1)$-marginals implies equality of all $k$-marginals, $k$-UDA implies $(n-1)$-UDA. Part (i) then shows that every $\alpha_j$ is pure. We write $\alpha_j=\ket{a_j}\bra{a_j}$.

It remains to exclude $1\leq k\leq n-2$. We define $\ket{\phi_j}:=\ket{j\cdots j}\ox\ket{a_j}$. For any $r\neq t$ and $0<|\epsilon|\leq 1$, we define
\begin{eqnarray}
\label{eq:rho_epsilon_lower_order}
\sigma_\epsilon
=\rho_{\alpha,p}^{(n)}
+\epsilon\sqrt{p_rp_t}\ket{\phi_r}\bra{\phi_t}
+\bar\epsilon\sqrt{p_rp_t}\ket{\phi_t}\bra{\phi_r}.
\end{eqnarray}
The vectors $\ket{\phi_j}$ are mutually orthogonal. On $\Span\{\ket{\phi_r},\ket{\phi_t}\}$, the modified coefficient matrix in \eqref{eq:rho_epsilon_lower_order} has determinant $p_rp_t(1-|\epsilon|^2)\geq 0$, so $\sigma_\epsilon$ is a state. It is distinct from $\rho_{\alpha,p}^{(n)}$. Every $k$-marginal with $k\leq n-2$ is obtained by tracing out at least one of $A_1,\ldots,A_{n-1}$, and the two off-diagonal terms in \eqref{eq:rho_epsilon_lower_order} vanish under this partial trace because $r\neq t$. Thus $\sigma_\epsilon$ and $\rho_{\alpha,p}^{(n)}$ have the same $k$-marginals, so $\rho_{\alpha,p}^{(n)}$ is not $k$-UDA. Therefore, $k=n-1$, and the condition in (i) follows again from (i).
\end{resultproof}

\section{Proofs for Section~\ref{sec:applications-GME}}
\label{app:proofs-GME-applications}

\begin{resultproof}{Proposition~\ref{pro:GME_detection_length_families}}
(i.a) We set $\ket w=b\ket{100}+c\ket{010}+d\ket{001}$ and
$\cL=\cR(\rho)=\Span\{\ket{000},\ket w\}$. For
$\ket\psi=\alpha\ket{000}+\beta\ket w$ with $\beta\neq0$, suitable
minors across $A:BC$, $B:AC$, and $C:AB$ are
$-\beta^2bd$, $-\beta^2bc$, and $-\beta^2cd$, respectively. Hence
$\ket\psi$ is GME, and the only biseparable vectors in $\cL$ are
multiples of $\ket{000}$. These vectors do not span $\cL$, so a state
with range $\cL$ cannot be biseparable. Thus $\rho$ is GME. Moreover,
$\cL\subseteq\cT$, so
Proposition~\ref{pro:T_Ct_2UDA} shows that $\rho$ is 2-UDA, and therefore $\ell_{\rm GME}(\rho)=2$.

(i.b) For $0<\theta\leq\pi/4$, a direct rank-one test shows that every
vector in $\cC_\theta$ that is product across $B:AC$ or $C:AB$ is
fully product. Thus every biseparable vector in $\cC_\theta$ is
product across $A:BC$, and a state whose range is contained in $\cC_\theta$ is
biseparable if and only if it is separable across $A:BC$. Under the
identification
$\Span\{\ket{00},\cos\theta\ket{01}+\sin\theta\ket{10}\}\simeq\bbC^2$,
the positive-partial-transpose (PPT) criterion for $2\ox2$ states
\cite{horodecki1996separability} shows that $\rho$ is GME if and only
if $\rho^{T_A}\not\geq0$. Proposition~\ref{pro:T_Ct_2UDA} again shows
that $\rho$ is 2-UDA. Hence, $\ell_{\rm GME}(\rho)=2$.

(i.c) For $\ket\psi=\alpha u+\beta v$, the minors in columns $(1,2)$
and $(2,4)$ of the coefficient matrices across $A:BC$ and $B:AC$ are
$-\alpha^2$ and $-\beta^2$, while those across $C:AB$ are
$-\alpha^2$ and $-4\beta^2$. Hence $\Span\{u,v\}$ contains no nonzero
biseparable vector, and every state with this range is GME. The state
$\rho_2$ in Example~\ref{exp:rank_two_GHZ_free_UDA} is 2-UDA and has
this range. By Lemma~\ref{le:lr+1-ls}, $\rho$ is also 2-UDA, and thus $\ell_{\rm GME}(\rho)=2$.

(ii) We have
$\cR(\sigma_{\Lambda,p}^{(n)})\subseteq\cQ\ox\bbC^s$. We first
observe that every biseparable vector in this subspace is product
across $\cQ:A_n$. Indeed, this is immediate for the bipartition
$A_1\cdots A_{n-1}:A_n$. Any other bipartition separates some of the
first $n-1$ systems. If we write a vector in the subspace as
$\sum_j\ket{j}^{\ox(n-1)}\ket{x_j}$, productness across such a
bipartition forces at most one $\ket{x_j}$ to be nonzero, since the
labels on both sides are mutually orthogonal. The vector is then
product across $\cQ:A_n$ as well. It follows that entanglement across
$\cQ:A_n$ implies that $\sigma_{\Lambda,p}^{(n)}$ is GME.

By Lemma~\ref{le:rho_sigma_compatible},
$\sigma_{\Lambda,p}^{(n)}$ has the same $(n-1)$-marginals as
$\rho_{\Lambda,p}^{(n)}$. The latter is fully separable, and equality
of the $(n-1)$-marginals also shows equality of every proper
marginal. Hence no collection of proper marginals detects the GME of
$\sigma_{\Lambda,p}^{(n)}$. The full state does detect it, and thus
$\ell_{\rm GME}(\sigma_{\Lambda,p}^{(n)})=n$. Finally, when $d=s=2$,
the effective state on $\cQ\ox\bbC^2$ is a two-qubit state, so the PPT
criterion proves the last assertion.
\end{resultproof}

\begin{resultproof}{Proposition~\ref{pro:GME-from-separable-marginals}}
We first prove that all 2-marginals are separable. For $i<j$, we set
$s_{ij}:=1-x_i-x_j$ and
$\ket{v_{ij}}:=\sqrt{x_i}\ket{10}+\sqrt{x_j}\ket{01}$. The
$(i,j)$-marginal of $\rho_{p;\boldsymbol{x}}$ is
$$
(\rho_{p;\boldsymbol{x}})_{A_iA_j}
=p s_{ij}\proj{00}
+p\proj{v_{ij}}
+(1-p)\proj{11}.
$$
Its partial transpose is positive semidefinite if and only if
$
p(1-p)s_{ij}-p^2x_ix_j\geq0,
$
which is equivalent to
$p\leq s_{ij}/(s_{ij}+x_ix_j)$. Hence the assumption
$p\leq p_*(\boldsymbol{x})$ and the PPT criterion for two-qubit states
\cite{horodecki1996separability} show that every 2-marginal is
separable.

We next prove that these separable 2-marginals jointly certify GME.
Let $\sigma$ be a state 2-compatible with
$\rho_{p;\boldsymbol{x}}$. The kernel of
$(\rho_{p;\boldsymbol{x}})_{A_iA_j}$ is generated by
$$
\ket{\kappa_{ij}}
:=\sqrt{x_i}\ket{01}-\sqrt{x_j}\ket{10}.
$$
Since $\sigma\geq0$ and $(\rho_{p;\boldsymbol{x}})_{A_iA_j}=\sigma_{A_iA_j}$, we have 
\begin{eqnarray}
\label{eq:compatible-support-filtered-symmetric}
\cR(\sigma)\subseteq\cK_{\boldsymbol{x}}
:=\bigcap_{i<j}\ker\bigl(
\bra{\kappa_{ij}}_{A_iA_j}\ox I_{A_{[n]\setminus\{i,j\}}}
\bigr).
\end{eqnarray}
We define
$D_i:=\ket0\bra0+\sqrt{x_i}\ket1\bra1$ and
$D_{\boldsymbol{x}}:=D_1\ox\cdots\ox D_n$. The relations in
\eqref{eq:compatible-support-filtered-symmetric} show that
$$
\cK_{\boldsymbol{x}}
=D_{\boldsymbol{x}}\operatorname{Sym}^n(\bbC^2),
$$
where
$$
\operatorname{Sym}^n(\bbC^2)
=\left\{\ket{\phi}\in(\bbC^2)^{\ox n}:
F_{ij}\ket{\phi}=\ket{\phi}\text{ for all }1\leq i<j\leq n\right\},
$$
Let $F_{A_iA_j}$ denote the flip operator on
$\cH_{A_i}\ox\cH_{A_j}$, and let $F_{ij}$ denote its natural
extension to the $n$-qubit Hilbert space. Thus $F_{ij}$ exchanges
$A_i$ and $A_j$ in the preceding definition.
 A symmetric pure state that is product across a bipartition is fully
product. Indeed, suppose
$\ket{\phi}=\ket{u}_{A_{\cS}}\ox\ket{v}_{A_{\cS^c}}$ for
$\varnothing\subsetneq\cS\subsetneq[n]$, and choose $i\in\cS$ and
$j\in\cS^c$. Symmetry
implies that all one-qubit reduced states are the same state $\tau$,
while the product form implies that the reduced state on $A_iA_j$ is
$\tau\ox\tau$. We set
$P^-_{A_iA_j}:=(I_{A_iA_j}-F_{A_iA_j})/2$. Then
$$
0=\tr\bigl[P^-_{A_iA_j}(\tau\ox\tau)\bigr]
=\frac12\bigl(1-\tr\tau^2\bigr).
$$
Thus $\tau$ is pure and
$\ket{\phi}=(a\ket0+b\ket1)^{\ox n}$ for some $a,b\in\bbC$. Since
$D_{\boldsymbol{x}}$ is
invertible and local, every
biseparable vector in $\cK_{\boldsymbol{x}}$ therefore has the form
\begin{eqnarray}
\label{eq:filtered-symmetric-product-vector}
\bigotimes_{i=1}^n
\bigl(a\ket0+b\sqrt{x_i}\ket1\bigr)
\end{eqnarray}
for some $a,b\in\bbC$.

Suppose, to the contrary, that $\sigma$ is biseparable. We may
decompose it into normalized states of the form in
\eqref{eq:filtered-symmetric-product-vector}. Let $\ket{\varphi_i}$ be
the normalized state proportional to $a\ket0+b\sqrt{x_i}\ket1$. For
$t=|\frac{b}{a}|^2\in[0,\infty]$, we have 
$
q_i(t):=\abs{\braket{1}{\varphi_i}}^2
=\frac{tx_i}{1+tx_i},
$
with the usual values at $t=0$ and $t=\infty$. Thus, for a probability
measure $\mu$ induced by this decomposition, the parameter of  $\ket{11}$
of the $(i,j)$-marginal is
$\int q_i(t)q_j(t)\,d\mu(t)$. The target marginals require
\begin{eqnarray}
\label{eq:pair-population-compatible-GME}
\int q_i(t)q_j(t)\,d\mu(t)=1-p,
\qquad i<j.
\end{eqnarray}
We choose $r,s\in[n]$ such that $x_r<x_s$ and choose
$j\notin\{r,s\}$. Applying 
\eqref{eq:pair-population-compatible-GME} for $(s,j)$ and $(r,j)$, we have
$
\int q_j(t)\bigl(q_s(t)-q_r(t)\bigr)\,d\mu(t)=0.
$
Indeed,
$
q_j(t)\bigl(q_s(t)-q_r(t)\bigr)
$
 is strictly positive for $0<t<\infty$ and vanishes at the two
endpoints. Since $\mu$ is a positive measure, the vanishing integral
implies that $\mu$ is supported on $\{0,\infty\}$. The endpoint $t=0$
corresponds to $b=0$ in
\eqref{eq:filtered-symmetric-product-vector}, while $t=\infty$
corresponds to $a=0$. Thus every pure-state projector in the preceding
biseparable decomposition is either $\proj{0}^{\ox n}$ or
$\proj{1}^{\ox n}$. Consequently,
$
\sigma=\lambda\proj{0}^{\ox n}+(1-\lambda)\proj{1}^{\ox n}
$
for some $0\leq\lambda\leq1$. Such a mixture has no
$\ket{10}\bra{01}$ term in any 2-marginal, whereas the corresponding
term of $(\rho_{p;\boldsymbol{x}})_{A_iA_j}$ equals
$p\sqrt{x_ix_j}>0$. This contradiction shows that every state
2-compatible with $\rho_{p;\boldsymbol{x}}$ is GME. Therefore, the
full collection of 2-marginals certifies the GME of
$\rho_{p;\boldsymbol{x}}$. On the other hand, its 1-marginals are
compatible with the fully product state
$
\bigotimes_{i=1}^n(\rho_{p;\boldsymbol{x}})_{A_i},
$
and hence no collection of 1-marginals can certify GME. We thus obtain
$\ell_{\rm GME}(\rho_{p;\boldsymbol{x}})=2$.

We finally prove that this certification does not arise from the
2-UDA property when $n\geq4$. We set
$$
H:=\ket{w_{\boldsymbol{x}}}\bra{1}^{\ox n}
+\ket{1}^{\ox n}\bra{w_{\boldsymbol{x}}}.
$$
The two computational-basis vectors in each summand differ on $n-1$
qubits, and hence every 2-marginal of $H$ vanishes. Thus
$H\in\cN_2$. Moreover, for every nonzero real $\varepsilon$ with
$|\varepsilon|\leq\sqrt{p(1-p)}$, the restriction of
$\rho_{p;\boldsymbol{x}}+\varepsilon H$ to
$\Span\{\ket{w_{\boldsymbol{x}}},\ket{1}^{\ox n}\}$ is
$
\bma
p&\varepsilon\\
\varepsilon&1-p
\ema
\geq0.
$
Therefore $\rho_{p;\boldsymbol{x}}+\varepsilon H$ is a state distinct
from and 2-compatible with $\rho_{p;\boldsymbol{x}}$, where we have
also used $\tr H=0$. Hence
$\rho_{p;\boldsymbol{x}}$ is not 2-UDA. By the preceding part, every
such compatible state is nevertheless GME. Thus the 2-marginals
certify the entanglement property without determining the global state. 
\end{resultproof}

\bibliographystyle{unsrt}
\bibliography{UDA}
\end{document}